\documentclass[a4paper,UKenglish,cleveref, autoref, thm-restate]{lipics-v2019}
\usepackage{amsmath, amsthm, amssymb}
\usepackage{graphicx}
\usepackage{subcaption}
\usepackage{gensymb}
\usepackage{mathtools}
\usepackage{cleveref}
\usepackage{blindtext}
\usepackage[ruled,vlined]{algorithm2e}
\usepackage{mathtools}
\usepackage{enumitem} 

\SetCommentSty{mycommfont}
\SetArgSty{textup}

\graphicspath{{./graphics/}}

\title{On Efficient Connectivity-Preserving Transformations in a Grid} 

\titlerunning{On Efficient Connectivity-Preserving Transformations in a Grid} 

\author{Abdullah Almethen}{Department of Computer Science, University of Liverpool, UK}{A.Almethen@liverpool.ac.uk}{}{}
\author{Othon Michail}{Department of Computer Science, University of Liverpool, UK}{Othon.Michail@liverpool.ac.uk}{}{}
\author{Igor Potapov}{Department of Computer Science, University of Liverpool, UK}{Potapov@liverpool.ac.uk}{}{}

\authorrunning{A. Almethen, O. Michail  and I. Potapov} 

\Copyright{Abdullah Almethen, Othon Michail and Igor Potapov}



\category{} 

\supplement{}

\nolinenumbers 

\EventEditors{}
\EventNoEds{2}
\EventLongTitle{}
\EventShortTitle{}
\EventAcronym{}
\EventYear{}
\EventDate{}
\EventLocation{}
\EventLogo{}
\SeriesVolume{}
\ArticleNo{}

\begin{document}

\maketitle

\begin{abstract}
We consider a discrete system of $n$ devices lying on a 2-dimensional square grid and forming an initial connected shape $S_I$. Each device is equipped with a linear-strength mechanism which enables it to move a whole line of consecutive devices in a single time-step. We study the problem of transforming $S_I$ into a given connected target shape $S_F$ of the same number of devices, via a finite sequence of \emph{line moves}. Our focus is on designing \emph{centralised} transformations aiming at \emph{minimising the total number of moves} subject to the constraint of \emph{preserving connectivity} of the shape throughout the course of the transformation. We first give very fast connectivity-preserving transformations for the case in which the \emph{associated graphs} of $ S_I $ and $ S_F $ are isomorphic to a Hamiltonian line. In particular, our transformations make $ O(n \log n $) moves, which is asymptotically equal to the best known running time of connectivity-breaking transformations. Our most general result is then a connectivity-preserving \emph{universal transformation} that can transform any initial connected shape $ S_I $ into any target connected shape $ S_F $, through a sequence of $O(n\sqrt{n})$ moves. Finally, we establish $\Omega(n \log n)$ lower bounds for two restricted sets of transformations. These are the first lower bounds for this model and are matching the best known $ O(n \log n)  $ upper bounds.
\end{abstract}

\section{Introduction}
\label{sec:Intro}

Over the past few years, many fascinating systems have been developed, leveraging advanced
technology in order to deploy large collections of tiny monads. Each monad is typically a highly restricted micro-robotic entity, equipped with a microcontroller and some actuation and sensing capabilities. Through its collaborative complexity, the collection of monads can carry out tasks which are well beyond the capabilities of individual monads. The vision is the development of materials that will be able to algorithmically change their physical properties, such as their shape, colour, conductivity and density, based on transformations executed by an underlying program. These efforts are currently shaping the research area of \emph{programmable matter}, which has attracted much theoretical and practical interest.   

The implementation indicates whether the monads are operated centrally or through local decentralised control. In \textit{centralised} systems, there is an external program which globally controls all monads with full knowledge of the entire system. On the other hand, \textit{decentralised} systems provide each individual monad with enough autonomy to communicate with its neighbours and move locally. There are an impressive number of recent developments for collective robotic systems, demonstrating their potential and feasibility, starting from the scale of milli or micro \cite{BG15,GKR10,KCL12,RCN14}  down to nano size of individual monads \cite{DDL09,Ro06}. 

Recent research has highlighted the need for the development of an algorithmic theory of such systems. Michail and Spirakis \cite{MS18} and Michail \emph{et al.} \cite{MSS19} emphasised an apparent lack of a formal theoretical study of this prospective, including modelling, possibilities/limitations, algorithms and complexity.  The development of a formal theory is a crucial step for further progress in those systems. Consequently, multiple theoretical computer science sub-fields have appeared, such as metamorphic systems \cite{DSY04b,NGY00,WWA04}, mobile robotics \cite{CFPS12,CDP19,DFSY15,DiLuna2019,SY10}, reconfigurable robotics \cite{ABD13,BKR04,DDG18,DGRSS16,YSS07}, passively-mobile systems \cite{AADFP06, AAER07,MS16a,MS18}, DNA self-assembly \cite{Do12,RW00,Wi98,WCG13}, and the latest emerging sub-area of ``Algorithmic Foundations of Programmable Matter'' \cite{FRRS16}.

Consider a system deployed on a two-dimensional square grid in which a collection of spherical devices are typically connected to each other, forming a shape $S_I$. By a finite number of valid individual moves, $S_I$  can be transformed into a desired target shape $S_F$. In this prospective, a number of models are designed and introduced in the literature for such systems. For example, Dumitrescu and Pach \cite{DP04}, Dumitrescu \emph{et al.} \cite{DSY04a,DSY04b} and Michail \emph{et al.} \cite{MSS19} consider mechanisms where an individual device is capable to move over and turn around its neighbours through empty space. Transformations based on similar moves being assisted by small seeds, have also been considered in \cite{ADDDFKPP19}. A new linear-strength mechanism was introduced by Almethen \emph{et al.} in \cite{AMP19}, where a whole line of consecutive devices can, in a single time-step, move by one position in a given direction. 

In this paper, we embark from the \textit{line-pushing} model of \cite{AMP19}, which provided sub-quadratic centralised transformations that may, though, arbitrarily break connectivity of the shape during their course. The only connectivity-preserving transformation in \cite{AMP19} was an $ O(n\sqrt{n}) $-time transformation for a single pair of shapes of order $ n $, namely from a diagonal into a straight line. All transformations that we provide in the present study preserve connectivity of the shape during the transformations. We first give very fast connectivity-preserving transformations for the case in which the \emph{associated graphs} of $ S_I $ and $ S_F $ are isomorphic to a Hamiltonian line. In particular, our transformations make $ O(n \log n $) moves, which is asymptotically equal to the best known running time of connectivity-breaking transformation. Our most general result is then a connectivity-preserving \emph{universal transformation} that can transform any initial connected shape $ S_I $ into any target connected shape $ S_F $, through a sequence of $O(n\sqrt{n})$ moves. Finally, we establish $\Omega(n \log n)$ lower bounds for two restricted sets of transformations. These are the first lower bounds for this model and are matching the best known $ O(n \log n)  $ upper bounds.    

\subsection{Related Work}

For the models of individual moves where only one node moves in a single time-step, \cite{DP04,MSS19} show universality of transforming any pair of connected shapes $(A,B)$ having the same number of devices (called \textit{nodes} throughout this paper) to each other via sliding and rotation mechanisms. By allowing only rotation, \cite{MSS19} proves that the problem of deciding transformability is in $\mathbf{P}$. It can be shown that in all models of constant-distance individual moves, $\Omega(n^2)$ moves are required to transform some pairs of connected shapes, due the inherent distance between them \cite{MSS19}. This motivates the study of alternative types of moves that are reasonable with respect to practical implementations and allow for sub-quadratic reconfiguration time in the worst case.  

There are attempts in the literature to provide alternatives for more efficient reconfiguration. The first main approach is to explore parallel transformations, where multiple nodes move together in a single time-step. This is a natural step to tackle such a problem, especially in distributed systems where nodes can make independent decisions and move locally in parallel to other nodes. There are a number of theoretical studies on parallel and distributed transformations \cite{DDG18,DGPR16,DiLuna2019,DSY04b,MSS19,YUY16} as well as practical implementations \cite{RCN14}. For example, it can be shown that a connected shape can transform into any other connected shape, by performing in the worst case $O(n)$ parallel moves around the perimeter of the shape \cite{MSS19}.

The second approach aims to equip nodes in the system with a more powerful mechanism which enables them to reduce the inherent distance by a factor greater than a constant in a single time-step. There are a number of models in the literature in which individual nodes are equipped with strong actuation mechanisms, such as linear-strength mechanisms.  Aloupis \emph{et al.} \cite{ABD13,ACD08} provide a node with arms that are capable to extend and contract a neighbour, a subset of the nodes or even the whole shape as a consequence of such an operation. Further, Woods \emph{et al.} \cite{WCG13} proposed an alternative linear-strength mechanism, where a node has the ability to rotate a whole line of consecutive nodes.

Recently, the \textit{line-pushing} model of \cite{AMP19} follows a similar approach in which a single node can move a whole line of consecutive nodes by simultaneously (i.e., in a single time-step) pushing them towards an empty position. The line-pushing model can simulate the rotation and sliding based transformations of \cite{DP04,MSS19} with at most a 2-factor increase in their worst-case running time. This implies that all transformations established for individual nodes, transfer in the line-pushing model and their universality and reversibility properties still hold true.  They achieved sub-quadratic time transformations, including an $O(n \log n)$-time universal transformation which does not preserve connectivity and a connectivity-preserving $O(n \sqrt{n} )$-time transformation for the special case of transforming a diagonal  into a straight line.       

Another relevant line of research has considered a single moving robot that transforms an otherwise static shape by carrying its tiles one at a time \cite{CDP19,FGHKSS18,GHK19}. Those models are partially centralised as a single robot (usually a finite automaton) controls the transformation, but, in contrast to our perspective, control in that case is local and lacking global information.

\subsection{Our Contribution}

In this work, we build upon the findings of \cite{AMP19} aiming to design very efficient and general transformations that are additionally able to keep the shape connected throughout their course. 

We first give an $O(n \log n)$-time transformation, called \textit{Walk-Through-Path}, that works for all pairs of shapes $(S_I,S_F)$ that have the same order and belong to the family of \emph{Hamiltonian shapes}. A \textit{Hamiltonian shape} is any connected shape $S$ whose \emph{associated graph} $G(S)$ is isomorphic to a Hamiltonian path (see also \cite{IPS82}). 
At the heart of our transformation is a recursive successive doubling technique, which starts from one endpoint of the Hamiltonian path and proceeds in $\log n$ phases (where $n$ denotes the order of the input shape $S_I$, throughout this paper). In every phase $i$, it moves a terminal line $L_i$ of length $2^i$ a distance $2^i$ higher on the Hamiltonian path through a \emph{LineWalk} operation. This leaves a new terminal sub-path $S_i$ of the Hamiltonian path, of length $2^i$. Then the general procedure is recursively called on $S_i$ to transform it into a straight line $L^{\prime}_i$ of length $2^i$. Finally, the two straight lines $L_i$ and $L^{\prime}_{i}$ which are perpendicular to each other are combined into a new straight line $L_{i+1}$ of length $2^{i+1}$ and the next phase begins. 

A core technical challenge in making the above transformation work is that Hamiltonian shapes do not necessarily provide free space for the \emph{LineWalk} operation. Thus, moving a line has to take place through the remaining configuration of nodes while at the same time ensuring that it does not break their and its own connectivity, including keeping itself connected to the rest of the shape. We manage to overcome this by revealing a nice property of line moves, according to which a line $L$ can \emph{transparently} walk through \emph{any} configuration $S$ (independently of the latter's density) in a way that: (i) preserves connectivity of both $L$ and $S$ and (ii) as soon as $L$ has gone through it, $S$ has been restored to its original state, that is, all of its nodes are lying in their original positions. This property is formally proved in Proposition \ref{prop:LineAlongPath} (Section \ref{sec:Prelim}).       

We next develop a \emph{universal transformation}, called \emph{UC-Box}, that within $O(n\sqrt{n})$ moves transforms any pair of connected shapes of the same order to each other, while preserving connectivity throughout its course. Starting from the initial shape $S_I$, we first compute a spanning tree $T$ of $S_I$. Then we enclose the shape into a square box of size $n$ and divide it into sub-boxes of size $\sqrt{n}$, each of which contains at least one sub-tree of $T$. By moving lines in a way that does not break connectivity, we compress the nodes in a sub-box into an adjacent sub-box towards a parent sub-tree. By carefully repeating this we manage to arrive at a final configuration which is always a compressed square shape. The latter is a type of a \textit{nice} shape (a family of connected shapes introduced in \cite{AMP19}), which can be transformed into a straight line in linear time. We provide an analysis of this strategy based on the number of \emph{charging phases}, which turns out to be $\sqrt{n}$, each making at most $n$ moves, for a total of $O(n\sqrt{n})$ moves.

Finally, we establish $\Omega(n \log n)$ lower bounds for two restricted sets of transformations. These are the first lower bounds for this model and are matching the best known $ O(n \log n)  $ upper bounds. The first set consists of transformations from an initial diagonal into a target straight line. If every node on the diagonal only meets through shortest paths with other nodes at their original positions and every such meeting results in an irreversible merging (i.e., nodes that merge cannot split in future steps), then it can be shown that any such transformation has a labelled tree representation. The nodes of the tree are the nodes of the shape, the edges represent mergings between the corresponding nodes at some point in the transformation and the labels of the edges represent the shortest path distances between the original positions of the corresponding nodes. Then the total cost of the transformation is equal to the sum of the labels plus the sum of the sizes of all sub-trees of the tree. The latter additive factors capture the cost of \emph{turning} (i.e., changing the orientation of) a line of merged nodes, which is always equal to its length, and every meeting through a shortest path on the grid requires at least one turn.

 We further restrict attention to the sub-set of those transformations in which every leaf-to-root path has length at most 2. This captures tranformations in which every node must reach its final destination (on the target straight line) through at most 1 meeting-hop and at most 2 hops in total. Interestingly, by disregarding the sub-tree additive costs and the fact that our initial and target instances have specific geometric arrangements, it is known that computing a 2-HOPS MST in the Euclidean 2-dimensional space is a hard optimisation problem and the best known result is a PTAS by Arora \emph{et al.} \cite{ARR98} (cf. also \cite{CDL07}). By working on the tree representation, we show that any transformation in this set requires at least $\Omega(n \log n)$ moves. Our second lower bound is also $\Omega(n \log n)$ time, for an alternative set of one-way transformations.    

Section \ref{sec:Prelim} formally defines the model and the problems under consideration and proves a basic proposition which is a core technical tool in one of our transformations. Section \ref{sec:Hamiltonian_Shapes} presents our $O(n\log n)$-time transformation for Hamiltonian shapes. Section \ref{sec:nsqrtn_Universal_Transformation} discusses our universal $O(n\sqrt{n})$-time transformation. In Section \ref{sec:Lower_Bounds}, our lower bounds are proved. Finally, in Section  \ref{sec:Conclusion} we conclude and discuss interesting problems left open by our work.

\section{Preliminaries}
\label{sec:Prelim}

All transformations in this study operate on a two-dimensional square grid, in which each cell  has a unique position of non-negative integer coordinates $ (x, y) $, where $x$ represents columns and $y$ denotes rows in the grid. A set of $n$ nodes on the grid forms a shape $ S $ (of the order $n$), where every single node $u \in S$ occupies only one cell $cell(u) = (u_x, u_y)$. A node $u$ can be indicated at any given time by the coordinates $(u_x, u_y)$ of the unique cell that it occupies at that time. A node $v \in S$ is a \textit{neighbour} of (or \textit{adjacent} to) a node $u \in S$ if and only if their coordinates satisfy $u_x -1 \le v_x \le u_x +1$ and $u_y -1 \le v_y \le u_y +1$ (i.e., their cells are adjacent vertically, horizontally or diagonally).

\begin{definition}\label{def:associated_graph} 
	A graph $G(S) = (V,E)$ is \emph{associated} with a shape $ S $, where $u\in V$ iff $u$ is a node of $S$ and $(u,v)\in E$ iff $u$ and $v$ are neighbours in $S$. 
\end{definition}
 
A shape $S$ is connected iff $G(S)$ is a connected graph. We denote by $T(S)$ (or just $T$ when clear from context) a spanning tree of $G(S)$, and whenever we state that such a tree is given we make use of the fact that $T(S)$ can be computed in polynomial time.  

\begin{definition}  [A tree]	 \label{def:RootedTree} 
	A tree $(T, r)$, or $T$ whenever clear in the context, is rooted at a node $r \in V$, such that there is a unique path $\mathcal{P}$ from $r$ to each node $v \in V$ denoted by $\mathcal{P}_T(r,v)$ on which the distance $ \delta_T(r,v) $ is the number of edges between them. A node $v$ is a successor of $u$ iff $\mathcal{P}_T(r,v) \supset  \mathcal{P}_T(r,u)$, and $u$ is a parent of $v$ iff $\delta_T(u,v) = 1$.       
\end{definition}

The size of a tree, $size(T)$, denotes the number of all nodes in $T$, includes the root $r$ and all its successors. In what follows, $n$ denotes the number of nodes in a shape under consideration, and all logarithms are to base 2.  

In this paper, we exploit the linear-strength mechanism of the \textit{line-pushing model} introduced in \cite{AMP19}. A line $ L $ is a sequence of nodes occupying consecutive cells in one direction of the grid, that is, either vertically or horizontally but not diagonally. A \textit{\textbf{line move}} is an operation of moving all nodes of $ L $ together in a single time-step towards a position adjacent to one of $ L $’s endpoints, in a given direction $d$ of the grid, $d \in \{up,down,right,left\}$. A line move may also be referred to as \emph{step}, \emph{move}, or \emph{movement} in this paper. Throughout, the running time of transformations is measured in total number of line moves until completion. A \textit{line move} is formally defined below.

\begin{definition}[A permissible line move] \label{def:permissible_line_move} 
	A line $ L = (x,y), (x+1,y), \ldots , (x+k-1, y ) $ of length $k$, where $1 \leq k \le n $, can push all its $k$ nodes rightwards in a single move to positions $ (x+1,y), (x+2,y), \ldots , (x+k, y) $ iff there exists an empty cell to $(x+k, y)$. The ``down'', ``left'', and ``up'' moves are defined symmetrically, by rotating the whole shape 90\degree, 180\degree, and 270\degree\  clockwise, respectively.
\end{definition}

We next define a family of shapes that are used in one of our transformations.

\begin{definition}[Hamiltonian Shapes] \label{def:All_Line_Shapes}	
	A shape $S$ is called Hamiltonian iff $G(S)=(V,E)$ is isomorphic to a Hamiltonian path, i.e., a path starting from a node $u \in  V$, visiting every node in $V$ exactly once and ending at a node $v \in V$, where $v\neq u$. $\mathcal{H}$ denotes the family of all Hamiltonian shapes.	Figure \ref{fig:Straight_Spning_Line} shows some examples of Hamiltonian shapes.
	\begin{figure}[th!]
		\centering
		\subcaptionbox{A double-spiral.}
		{\includegraphics[scale=0.4]{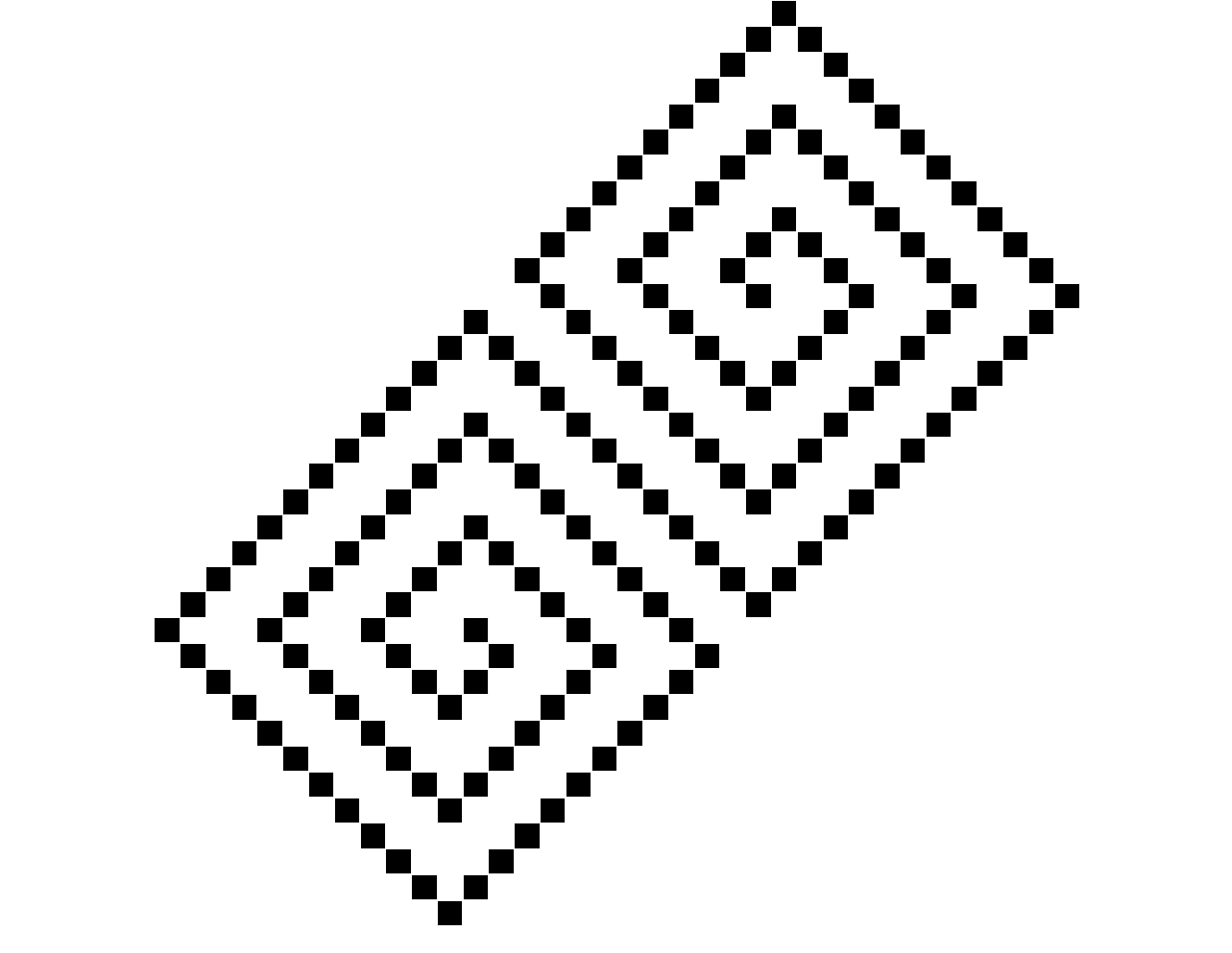}} \qquad \qquad  
		\subcaptionbox{A shape of two different Hamiltonian paths in yellow.}
		{\includegraphics[scale=0.4]{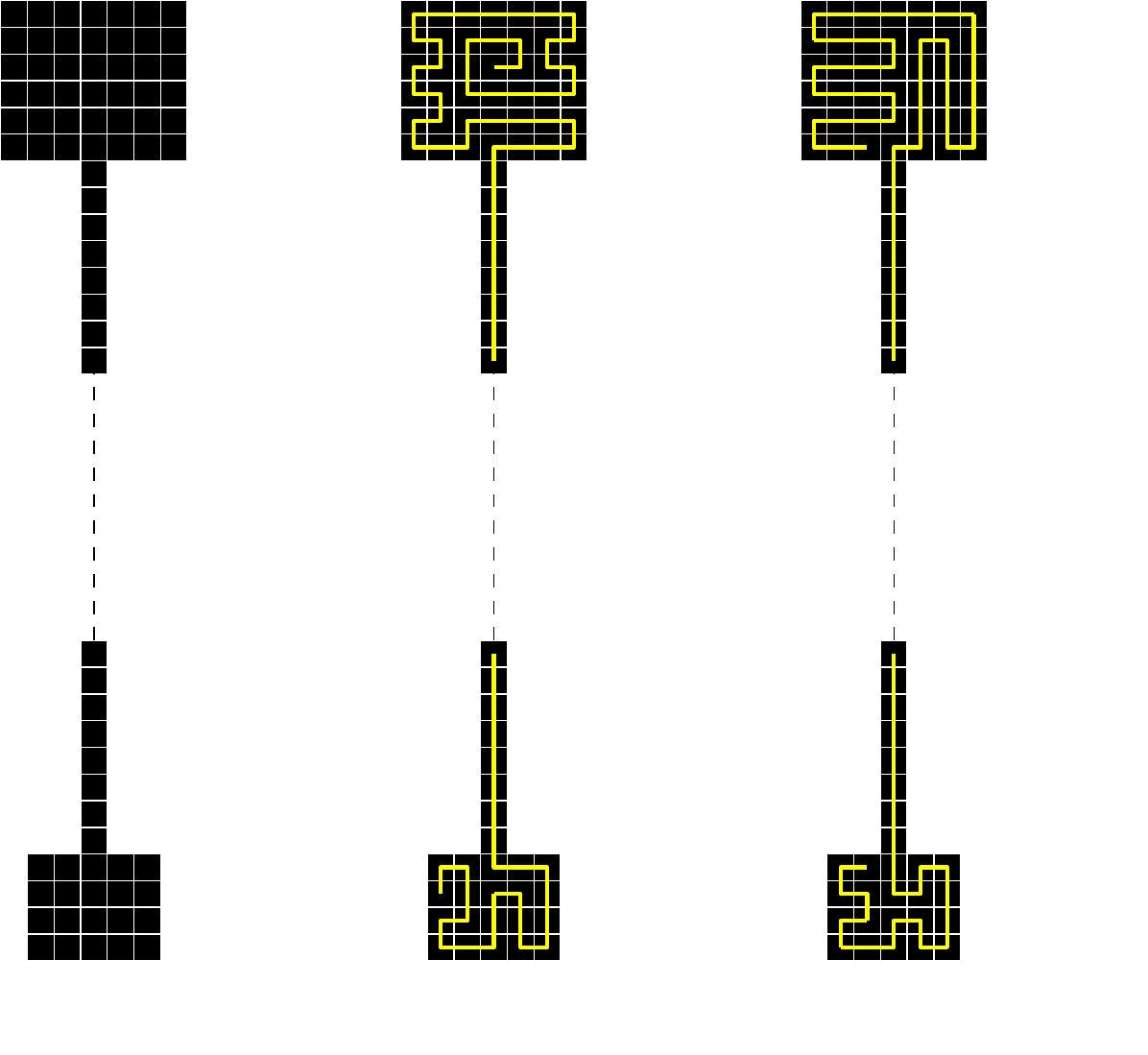}} 
		\caption{Examples of Hamiltonian shapes.}
		\label{fig:Straight_Spning_Line}
	\end{figure}
\end{definition}

The following proposition proves a basic property of line moves which will be a core technical tool in one of our transformation for Hamiltonian shapes.

\begin{proposition} [Transparency of Line Moves] \label{prop:LineAlongPath}
	Let $S$ be any shape, $ L \subseteq S$ any line and $ P $ any path of cells in the grid (under the vertical and horizontal neighbouring relation) starting from a position adjacent to one of $L$'s endpoints. Let $ C(P) $ denote the configuration of $P$ defined by $S$. 
	There is a way to move $ L $ along $ P $, while satisfying all the following properties: 
	\begin{enumerate}[label=(\roman*)]
		\item \emph{No delay:} The number of steps is asymptotically equal to that of an optimum move of $ L $ along $ P $ in the case of  $ C(P) $ being empty (i.e., if no cells were occupied). That is, $ L $ is not delayed, independently of what $ C(P) $ is.
		\item \emph{No effect:} After $ L $'s move along $ P $, $ C^{\prime}(P) = C(P) $, i.e., the cell configuration has remained unchanged. Moreover, no occupied cell in $ C(P) $ is ever emptied during $ L $'s move (but unoccupied cells may be temporarily occupied).
		\item \emph{No break:} $S$ remains connected throughout $L$'s move.
	\end{enumerate}
\end{proposition}
\begin{proof}
	Given $L \subseteq S$ and $P$, place additional nodes that occupy cells in $ P $, possibly with gaps, in any configuration $ C(P) $, see Figure \ref{fig:Line_Path} for example. Whenever $L$ walks through an empty cell $(x,y)$ of $ P $, a node $u\in L$ fills in $(x,y)$. If $L$ pushes the node $u$ of a non-empty cell of  $ P $, a node $v\in L$ takes its place. When $L$ leaves a non-empty cell $(x,y)$ that was originally occupied by node $v$, $L$ restores $(x,y)$ by leaving its endpoint $u \in L$ in $(x,y)$.
	
	\begin{figure}[th!]
		\centering
		{\includegraphics[scale=0.45]{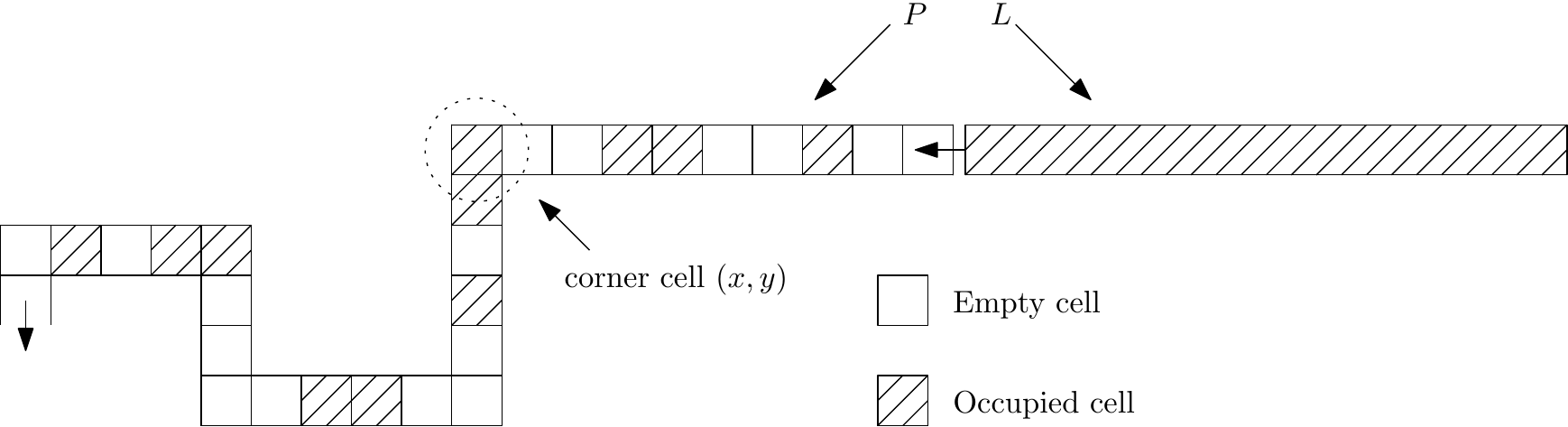}} 
		\caption{A path  $ P $ of a given configuration $ C(P) $. A line $L$ will pass along $p$.}
		\label{fig:Line_Path}
	\end{figure}  
	
	Now assume that $L$ turns at a non-empty corner cell $(x,y)$ of $P$ (say without loss of generality, from horizontal to vertical direction). Typically the node occupying the corner cell$(x,y)$ moves vertically one step along $P$, and then $L$ pushes one move to fill in the empty cell $(x,y)$ by a node $u\in L$. Unless  $(x,y)$ is being only connected diagonally to a non-empty cell that is not a neighbour of any node $u\in L$.  Figure \ref{fig:Line_Path_2} shows how to deal with the case in which $L$ turns at a non-empty corner-cell $(x,y)$ of $P$, which is only connected diagonally to a non-empty cell of $S$ and is not adjacent to any cell occupied by $L$. 
	
	\begin{figure}[th!]
		\centering
		\subcaptionbox{}
		{\includegraphics[scale=0.45]{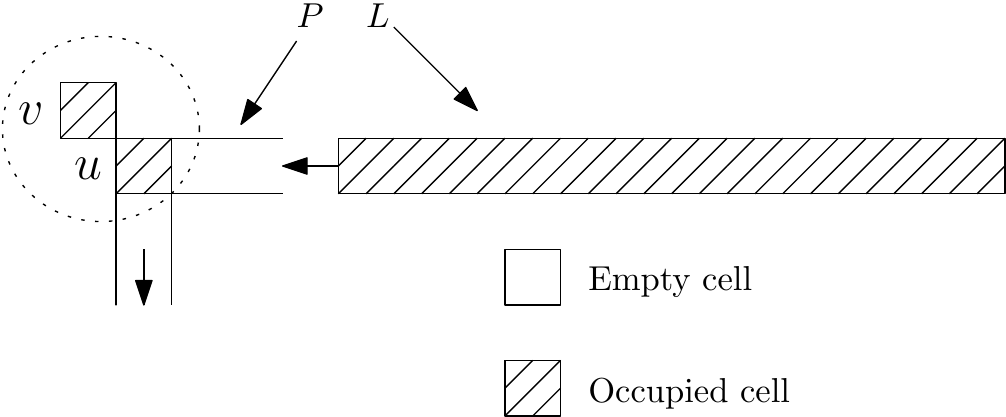}} \qquad 
		\subcaptionbox{}
		{\includegraphics[scale=0.45]{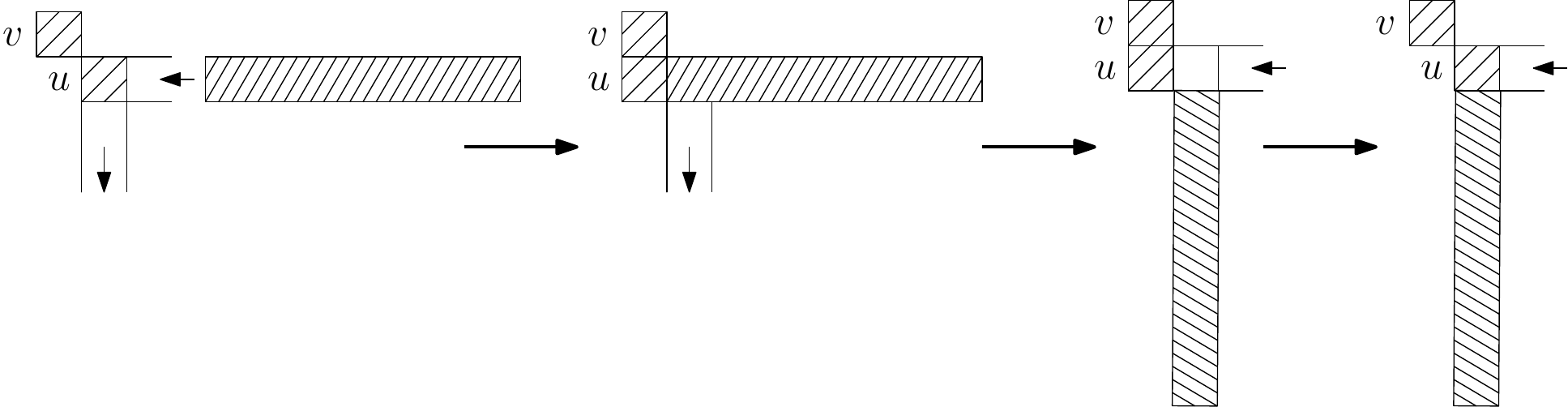}} 
		\caption{A line $L$ moving through a path $P$ and arriving at a turning point of $P$. $u$ occupies a corner cell of $P$ and $v$ occupies a cell of $S$ and is only connected diagonally to $u$ while not being adjacent to any cell occupied by $L$. $L$ pushes $u$ one position horizontally and turns all of its nodes vertically. Then $u$ moves back to its original position in $P$. All other orientations are symmetric and follow by rotating the shape $ 90\degree $, $ 180\degree $ or $ 270\degree$. }
		\label{fig:Line_Path_2}
	\end{figure} 
	
	Therefore,  it always temporarily maintain global connectivity and restores all of those nodes to their original positions. Hence, $L$'s move takes a number of moves to pass through any $ C(P) $ equal to or even less than its optimum move in the case of empty $C(P)$. Therefore, $L$ can \emph{transparently} walk through \emph{any} configuration $S$ (independently of the latter's density) in a way that: (i) preserves connectivity of both $L$ and $S$ and (ii) as soon as $L$ has gone through it, $S$ has been restored to its original state, that is, all of its nodes are lying in their original positions.
\end{proof}

We now formally define all problems considered in this work.\\

\noindent\textbf{{\sc HamiltonianConnected}.} Given a pair of connected Hamiltonian shapes $(S_I, S_F)$ of the same order, where $S_I$ is the initial shape and $S_F$ the target shape, transform $S_I$ into $S_F$ while preserving connectivity throughout the transformation. \\

\noindent\textbf{{\sc DiagonalToLineConnected}.} A special case of {\sc HamiltonianConnected} in which $S_I$ is a diagonal line and $S_F$ is a straight line. \\ 

\noindent\textbf{{\sc UniversalConnected}.} Given \emph{any} pair of connected shapes $(S_I, S_F)$ of the same order, where $S_I$ is the initial shape and $S_F$ the target shape, transform $S_I$ into $S_F$ while preserving connectivity throughout the transformation.

\section{$ O(n \log n) $-time Transformations for Hamiltonian Shapes}
\label{sec:Hamiltonian_Shapes}

In this section, we present a strategy for {\sc HamiltonianConnected}, called \textit{Walk-Through-Path}. It transforms any pair of shapes $S_I,S_F\in\mathcal{H}$ of the same order to each other within $O(n\log n)$ moves while preserving connectivity of the shape throughout the transformation. Recall that $\mathcal{H}$ is the family of all Hamiltonian shapes. Our transformation starts from one endpoint of the Hamiltonian path of $S_I$ and applies a recursive successive doubling technique to transform $S_I$ into a straight line $S_L$ in $O(n\log n)$ time. By replacing $S_I$ with $S_F$ in \textit{Walk-Through-Path} and reversing the resulting transformation, one can then go from $S_I$ to $S_F$ in the same asymptotic time. 

We first demonstrate the core recursive technique of this strategy in a special case which is sufficiently sparse to allow local reconfigurations without the risk of affecting the connectivity of the rest of the shape. In this special case, $S_I$ is a diagonal of any order and observe that $S_I,S_F\in\mathcal{H}$ holds for this case. We then generalise this recursive technique to work for any $S_I\in\mathcal{H}$ and add to it the necessary sub-procedures that can perform local reconfiguration in \emph{any} area (independently of how dense it is), while ensuring that global connectivity is always preserved.

\begin{figure}[th!]
	\centering
	\captionsetup{justification=centering}
	\subcaptionbox{First phase.}
	{\includegraphics[scale=0.44]{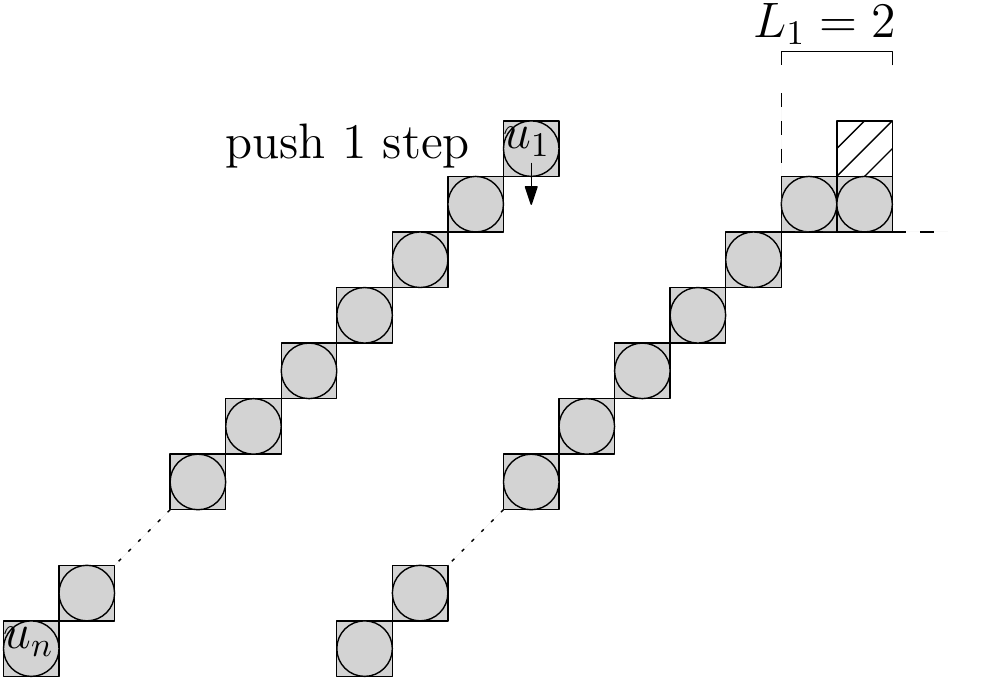}} \quad
	\vspace{10px}
	\subcaptionbox{Second phase.}
	{\includegraphics[scale=0.44]{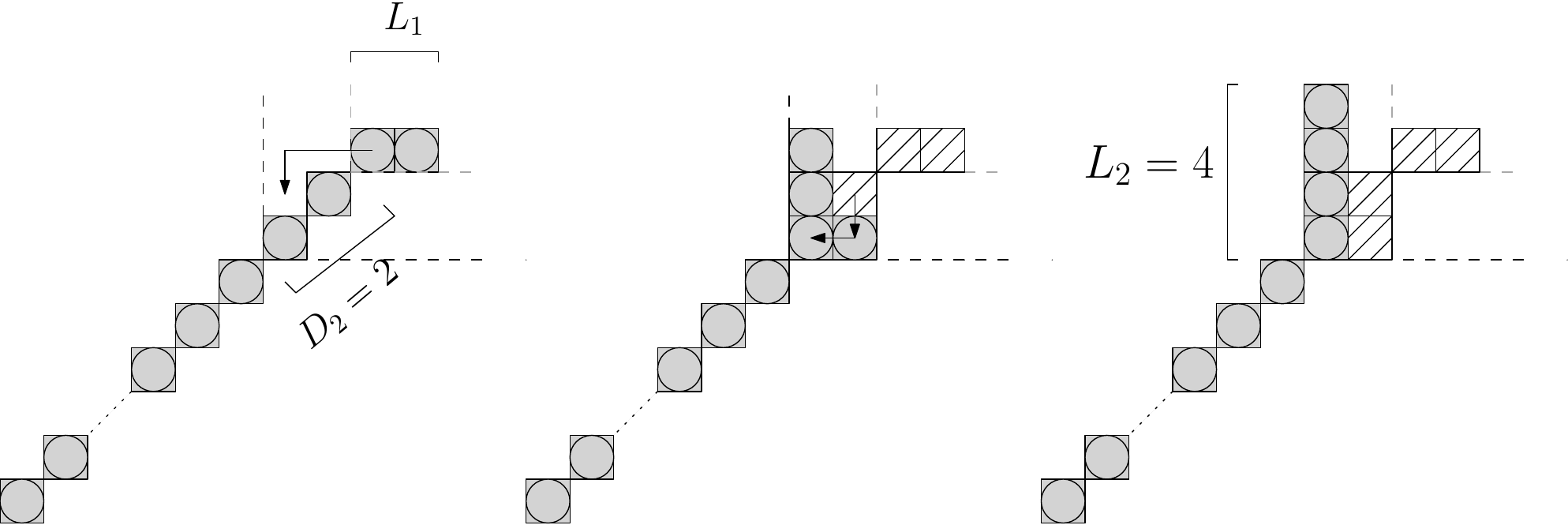}}
	\caption{First and second phase of \textit{Walk-Through-Path} on the diagonal shape.}
	\label{fig:FirstPhase}
\end{figure} 

Let $ S_I$ be a diagonal of $n$ nodes $u_n,u_{n-1},\ldots,u_1$, occupying cells $(x,y),(x+1,y+1),\ldots,(x+n-1,y+n-1)$, respectively. Assume for simplicity of exposition that $n$ is a power of 2; this can be dropped later. As argued above, it is sufficient to show how $S_I$ can be transformed into a straight line $S_L$. In phase $i =0$,  the top node $u_1$ moves one position to align with $u_2$ and form a line $L_1$ of length 2, as depicted in Figure \ref{fig:FirstPhase} (a). Next phase,  $L_1$ moves two positions and turns to align with $u_4$, then repeat whatever done in phase $i =0$ again on nodes $u_3$ and $u_4$ (where both form a diagonal segment $D_1$) to create a line $L^{\prime}_1$, and then combine the two perpendicular line $L_1$ and $L^{\prime}_1$ into a line $L_2$ of length 4, as shown in Figure \ref{fig:FirstPhase} (b).

In any phase $i$, for $1 \le i \le \log n$, a line $L_{i}$ occupies $2^{i}$ consecutive cells in a terminal subset of $S_I$ (Figure \ref{fig:I_Phase} (a)). $L_i$ moves through a shortest path towards the far endpoint of the next diagonal segment $D_i$ of length $2^{i}$ (Figure \ref{fig:I_Phase} (b)). Note that for general shapes, this move shall be replaced by a more general \textit{Line-Walk} operation (defined in the sequel). By a recursive call on $D_i$, $D_i$ transforms into a line $L_i^{\prime}$ (Figure \ref{fig:I_Phase} (c)). Finally, the two perpendicular lines $L_{i}$ and $L_i^{\prime}$ are combined in linear time into a straight line $L_{i+1}$ of length $2^{i+1}$ (Figure \ref{fig:I_Phase} (d)). By the end of phase $\log n$, a straight line $ S_L $ of order $n$ has been formed.  

\begin{figure}[th!]
	\centering
	\subcaptionbox{A line $L_{i}$ and a diagonal segment $D_i$ both of length $2^i$.}
	{\includegraphics[scale=0.25]{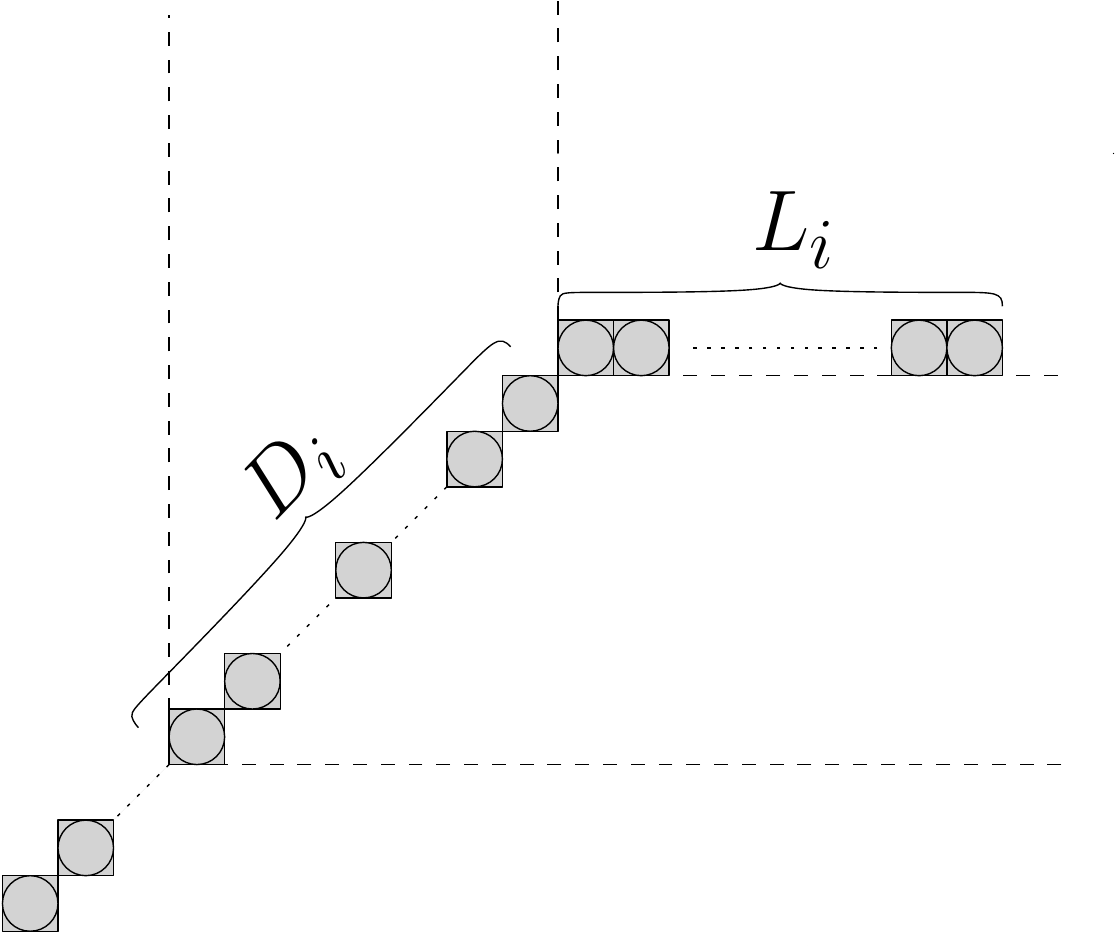}} \qquad
	\subcaptionbox{ $L_{i}$ moves through a shortest path towards the far endpoint of $D_i$.}
	{\includegraphics[scale=0.25]{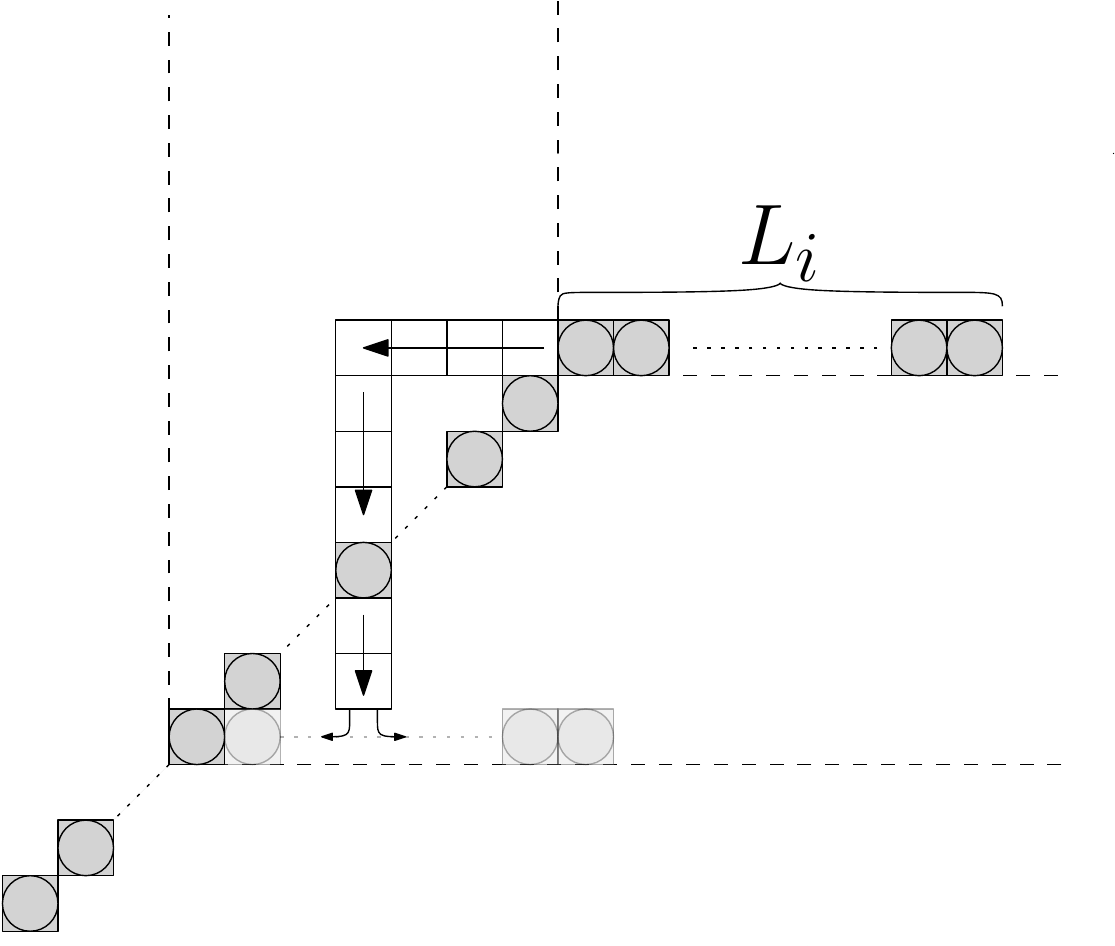}} \qquad
	\subcaptionbox{$D_i$ recursively transforms into a line $L_i^{\prime}$ .}
	{\includegraphics[scale=0.25]{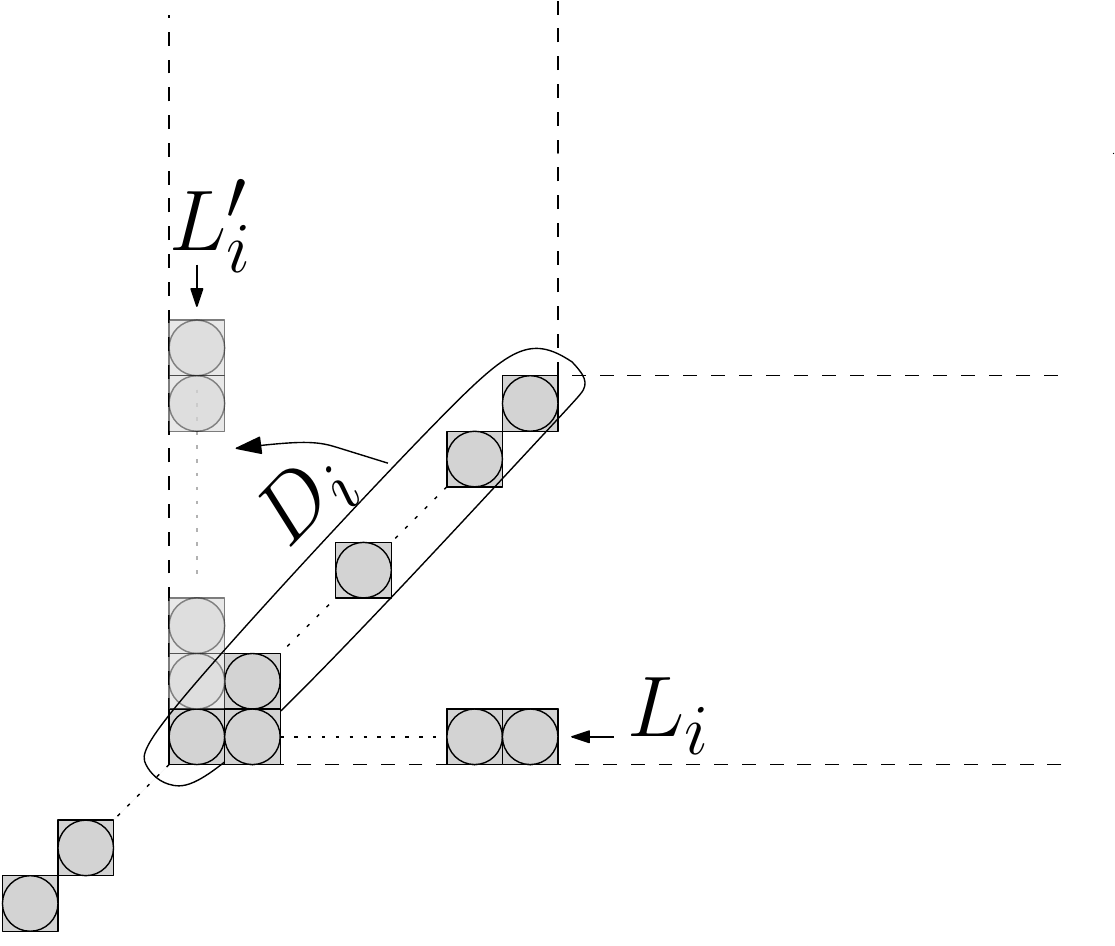}}\qquad
	\subcaptionbox{A  line $L_{i+1}$ of length $2^{i+1}$ formed by combining $L_{i}$ and $L_i^{\prime}$.}
	{\includegraphics[scale=0.25]{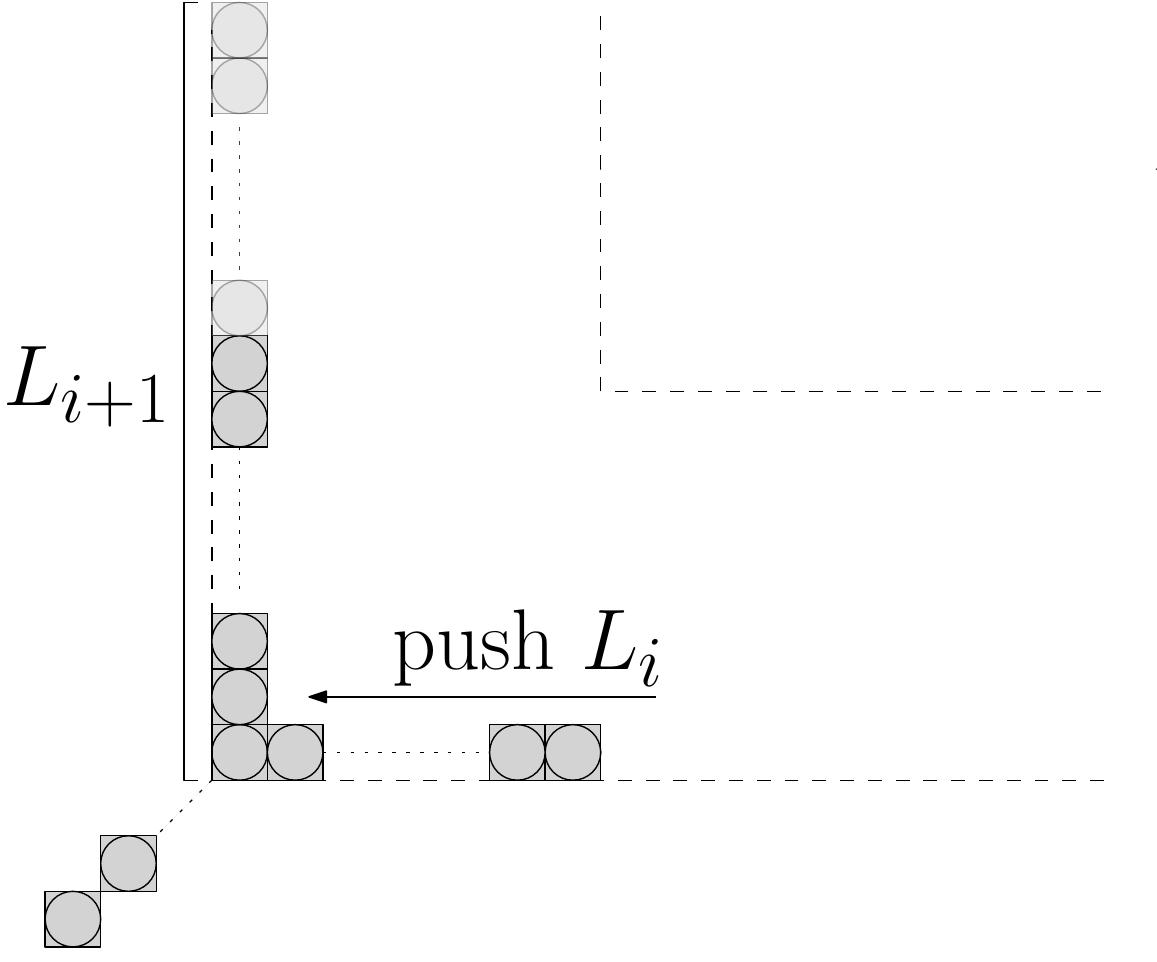}}
	\caption{A snapshot of phase $i$ of \textit{Walk-Through-Path} applied on a diagonal. Light grey cells represent the ending positions of the corresponding moves depicted in each sub-figure.}
	\label{fig:I_Phase}
\end{figure} 

A core technical challenge in making the above transformation work in the general case, is that Hamiltonian shapes do not necessarily provide free space, thus, moving a line has to take place through the remaining configuration of nodes while at the same time ensuring that it does not break their and its own connectivity. In the more general \emph{LineWalk} operation that we now describe, we manage to overcome this by exploiting \emph{transparency} of line moves, according to which a line $L$ can \emph{transparently} walk through any configuration $S$ (independently of the latter's density); see Proposition \ref{prop:LineAlongPath}.

\textbf{\textit{LineWalk.}} At the beginning of any phase $i$, there is a terminal straight line $L_i$ of length $2^i$ containing the nodes $v_1,\ldots, v_{2^{i}}$, which is connected to an $ S_i \subseteq S_I$, such that $S_i$ consists of the $ 2^i $ subsequent nodes, that is $v_{2^{i}+1},\ldots, v_{2^{i+1}}$. Observe that $S_i$ is the next terminal sub-path of the remaining Hamiltonian path of $S_I$. We distinguish the following cases: (1) If $L_i$ and $S_i$ are already forming a straight line, then go to phase $i+1$.  (2) If $S_i$ is a line perpendicular to $L_i$, then combine them into a straight line by pushing $L_i$ to extend $S_i$ and go to phase $i+1$. Otherwise, (3) check if the (Manhattan) distance between $v_{2^{i}}$ and $v_{2^{i+1}}$ is  $\delta(v_{2^i} ,  v_{2^{i+1}}) \le 2^{i}$, then $L_i$ moves from $v_{2^i} = (x,y)$ vertically or horizontally towards either node $(x, y^{\prime})$ or $(x^{\prime}, y)$ in which $L_i$ turns and keeps moving to $v_{2^{i+1}} = (x^{\prime},y^{\prime})$ on the other side of $S_I$. If not, (4) $L_i$ must first pass through a middle node of $S_I$ at $v_{2^{i} + 2^{i-1}}  = (x^{\prime\prime}, y^{\prime\prime})$, therefore $L_i$ repeats (3) twice, from $v_{2^{i}}$ to $v_{2^{i} + 2^{i-1}}$ and then towards $v_{2^{i+1}} $.

Note that cases (3) and (4) ensure that $L_i$ is not disconnected from the rest of the shape. Moreover, moving $L_i$ must be performed in a way that respects transparency (Proposition \ref{prop:LineAlongPath}), so that connectivity of the remaining shape is always preserved and its configuration is restored to its original state. These details are described later in this section. 

Algorithm \ref{algo2}, {\sc HamiltonianToLine}, gives a general strategy to transform any  Hamiltonian shape $ S_I \in \mathcal{H}$ into a straight line in $O(n\log n)$ moves. In every phase $i$, it moves a terminal line $L_i$ of length $2^i$ a distance $2^i$ higher on the Hamiltonian path through a \emph{LineWalk} operation. This leaves a new terminal sub-path $S_i$ of the Hamiltonian path, of length $2^i$. Then the general procedure is recursively called on $S_i$ to transform it into a straight line $L^{\prime}_i$ of length $2^i$. Finally, the two straight lines $L_i$ and $L^{\prime}_{i}$ which are perpendicular to each other are combined into a new straight line $L_{i+1}$ of length $2^{i+1}$ and the next phase begins. The output of {\sc HamiltonianToLine} is a straight line $S_L$ of order $n$. 

\vspace{10pt}
\begin{algorithm}[H]
	$ S = (u_0,u_1,...,u_{|S|-1}) $ is a Hamiltonian shape\\
	\SetAlgoLined
	\DontPrintSemicolon	
	Initial conditions: $ S \gets S_I$ 	and $ L_0 \gets \{u_0\} $\\
	\vspace{7pt}
	\For{$ i = 0, \ldots,$ $ \log |S| $}
	{
		LineWalk($ L_i $)\\
		$ S_i $ $\gets$  select($ 2^i $) \tcp{select the next terminal subset of $ 2^i $ consecutive nodes of $ S $}
		$ L^{\prime}_i $ $\gets$ HamiltonianToLine($ S_i $) \tcp{recursive call on $S_i$}
		$ L_{i+1} $ $\gets$ combine$( L_i, L^{\prime}_i )$ \tcp{combines $L_i$ and $L^{\prime}_i$ into a new straight line $L_{i+1}$}
	}
	\KwOut{a straight line $S_L$}
	\caption{{\sc HamiltonianToLine}($ S $)}
	\label{algo2}
\end{algorithm}
\vspace{10pt}

Now, we are ready to show correctness of \textit{Walk-Through-Path} in the following lemmas.
\begin{lemma} \label{lem:Walk-Through-Path}
	Starting from an initial  Hamiltonian shape $S_I\in\mathcal{H}$  of  order $n$,  {\sc HamiltonianToLine} forms a straight line $S_L\in\mathcal{H}$ of length $n$.
\end{lemma}   
\begin{proof}
	By the beginning of the final phase, the shape configuration consists of two parts, a straight line $L$ of length $2^{\log n -1 }$ and a shape $S$ of $2^{\log n -1 }$ nodes. During this phase,   $L$ performs a \emph{LineWalk} operation, $S$ transforms recursively into $L^{\prime}$  and then  $L$ combines with $L^{\prime}$ into a straight line $S_L$ of length $2^{log n} =n$. Consequently, $S_L$  shall occupy $n$ consecutive cells on the grid, either vertically or horizontally.  
\end{proof}
\begin{lemma} \label{lem:CorrectnessOf_LineWalk}
	The operation of \textit{Line-Walk} preserves the whole connectivity of the shape during phase $i$. 
\end{lemma}   
\begin{proof}
	Let $S_I\in\mathcal{H}$ be a Hamiltonian shape of  order $n$ in phase $i$, which terminates at a straight line $L_i$ of length $2^{i}$ nodes, starting from $v_1$ to $v_{2^i}$. During phase $i$, this transformation doubles the size of $L_i$ by merging its nodes with the following $2^i$ nodes on the Hamiltonian path that are forming a shape $S_i$  from $v_{2^i+1}$ to $v_{2^{i+1}}$. 
	
	We now show  case (1) and (2) of the \textit{Line-Walk} operation on a horizontal $L_i$ (\textit{the other cases are symmetric by rotating the shape $ 90\degree $, $ 180\degree $ or $ 270\degree $}). In case 1, $L_i$ and $S_i$ are already forming a straight line $L_{i+1}$ of length $2^{i+1}$, hence the whole configuration of the shape left unchanged. In case (2), $L_i$  and $S_i$ are forming two perpendicular straight lines in which $L_i$ can easily push into $S_i$  and extend it by $2^i$. As $L_i$ pushes and  $S_i$ extends, they are replacing and restoring any occupied cell along their way through \emph{any} configuration (independently of how density is) by exploiting \emph{transparency} of line moves in Proposition \ref{prop:LineAlongPath}. As a result, the \textit{Line-Walk} operation preserves connectivity of  $L_i$, $S_i$ and the whole shape. 	

	Now, let $L_i$ and $S_i$ be of the same configuration of case (3) or (4) described above, where $L_i$ has a length of  $2^{i}$ and $S_i$ consists of $2^{i }$ nodes $v_{2^{i}+1},\ldots, v_{2^{i+1}}$ that occupy multiple rows and columns.  Assume that $L_i$ is horizontal and occupies $ (x,y), (x+1,y), \ldots , (x+2^i, y)$, this is sufficient as the other cases are symmetric if one rotates the whole shape $ 90\degree $, $ 180\degree $ or $ 270\degree $. Observe that $S_i$ is the next terminal sub-path of the remaining Hamiltonian path. Consequently, the Manhattan distance between $ v_ {2 ^ i} $ and $ v_ {2 ^ {i + 1}} $ specifies the path that $ L_i $ will follow to meet and align with the far endpoint of $ S_i $.
	
	Recall that the minimum Manhattan (taxicab) distance of any path in a square grid, which starts from point  $u$ and ends at $v$, $\delta(u, v) = |u_x - v_x| + |u_y - v_y|  $, will always have the same length and this transamination picks a path of minimum turns (aiming for low cost). Hence,  there are two feasible L-shaped paths from $u$ to $v$, each of which has one turn. The first path starts horizontally from point $(u_x, u_y)$  towards $(v_x, u_y)$ then turns vertically to $(v_x, v_y)$, and the second one starts from $(u_x, u_y)$ vertically to  $(u_x, v_y)$ then turns horizontally towards $(v_x, v_y)$.  
	
	In case (3), the Manhattan distance between $v_{2^i}$ and $v_{2^{i+1}}$ is $\delta(v_{2^i}, v_{2^{i+1}}) \le 2^{i}$, then $L_i$ moves horizontally from $v_{2^i} = (x,y)$ along $(x^{\prime}, y)$ in which $L_i$ changes its direction towards $v_{2^{i+1}} = (x^{\prime}, y^{\prime})$. In a worst-case configuration, a path may consists of  at $2^{i}$ empty cells $L_i$ must pass to reach the destination cell $(x^{\prime}, y^{\prime})$. Recall that $L_i$ contains $2^i$ nodes, hence $L_i$ shall arrive at $(x^{\prime}, y^{\prime})$, occupy all $2^{i}$ cells and still connected. Once $L_i$ arrived there, it can safely change its direction  to line up with $v_{2^{i+1}} $ and occupy the column $x^{\prime}$, while being connected too. Moreover, assume the path along which $L_i$ has moved contains  non-empty cells, therefore all of them are restored by \emph{transparency} of line moves shown in Proposition \ref{prop:LineAlongPath}.
	
	That is, as $L_i$ moves along a path of non-empty cells within phase $i$, it pushes a node $u \notin L_i$  and replaces it by node $u \in L_i$.  When $L_i$ leaves this path during phase $i+1$, it rosters any non-empty cell occupied by a pre-existing node $u\notin L$. 	The same argument holds for (4) by applying (3) twice. Figure \ref{fig:StraightLine_SpanningLine} demonstrates an example of case (3) and (4).	As a result, The operation of \textit{Line-Walk} keeps the whole shape connected during any phase $i$ of this transformation. 

\end{proof}
\begin{figure}[th!]
	\centering
	\subcaptionbox{The case when $\delta(v_{2^i}, v_{2^{i+1}}) \le 2^{i}$.}
	{\includegraphics[scale=0.57]{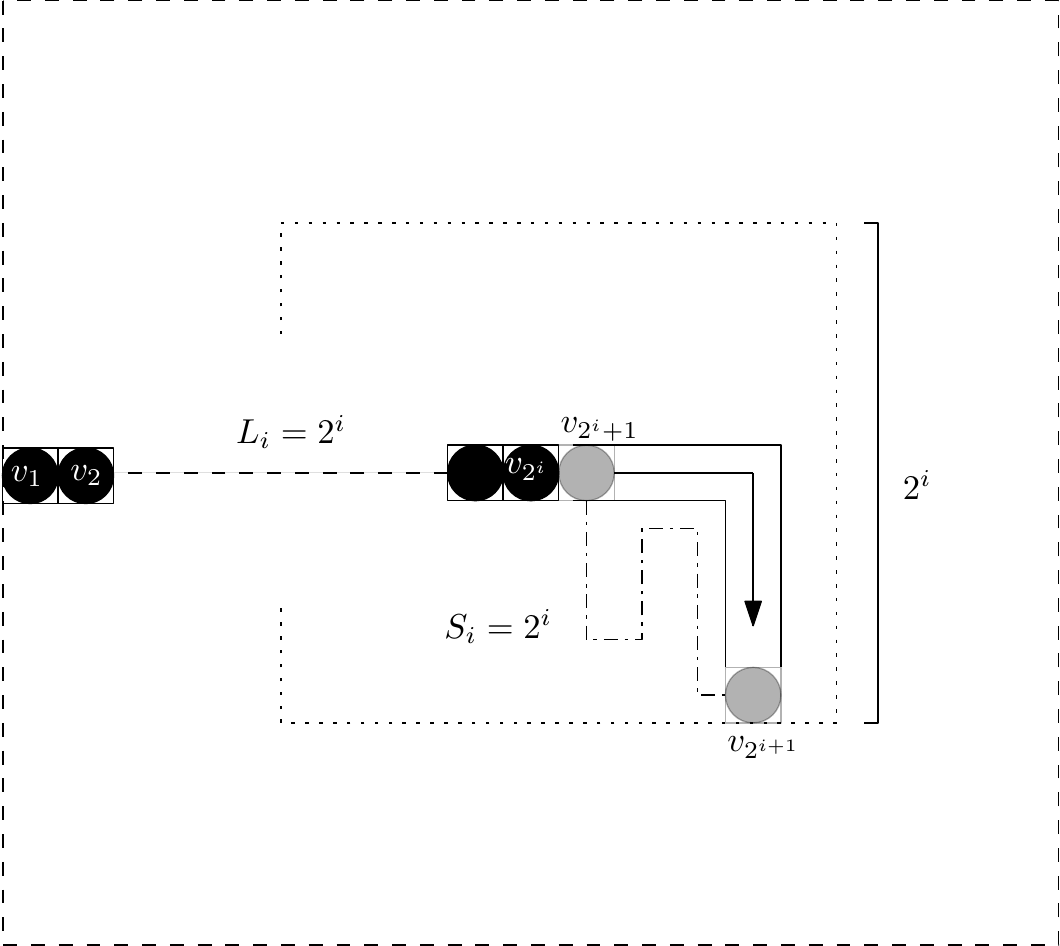}} \qquad
	\subcaptionbox{The case when $\delta(v_{2^i}, v_{2^{i+1}}) >2^{i}$, where node  $v_{2^{i}+ 2^{i-1}} $ at the middle of $S_i$.}
	{\includegraphics[scale=0.57]{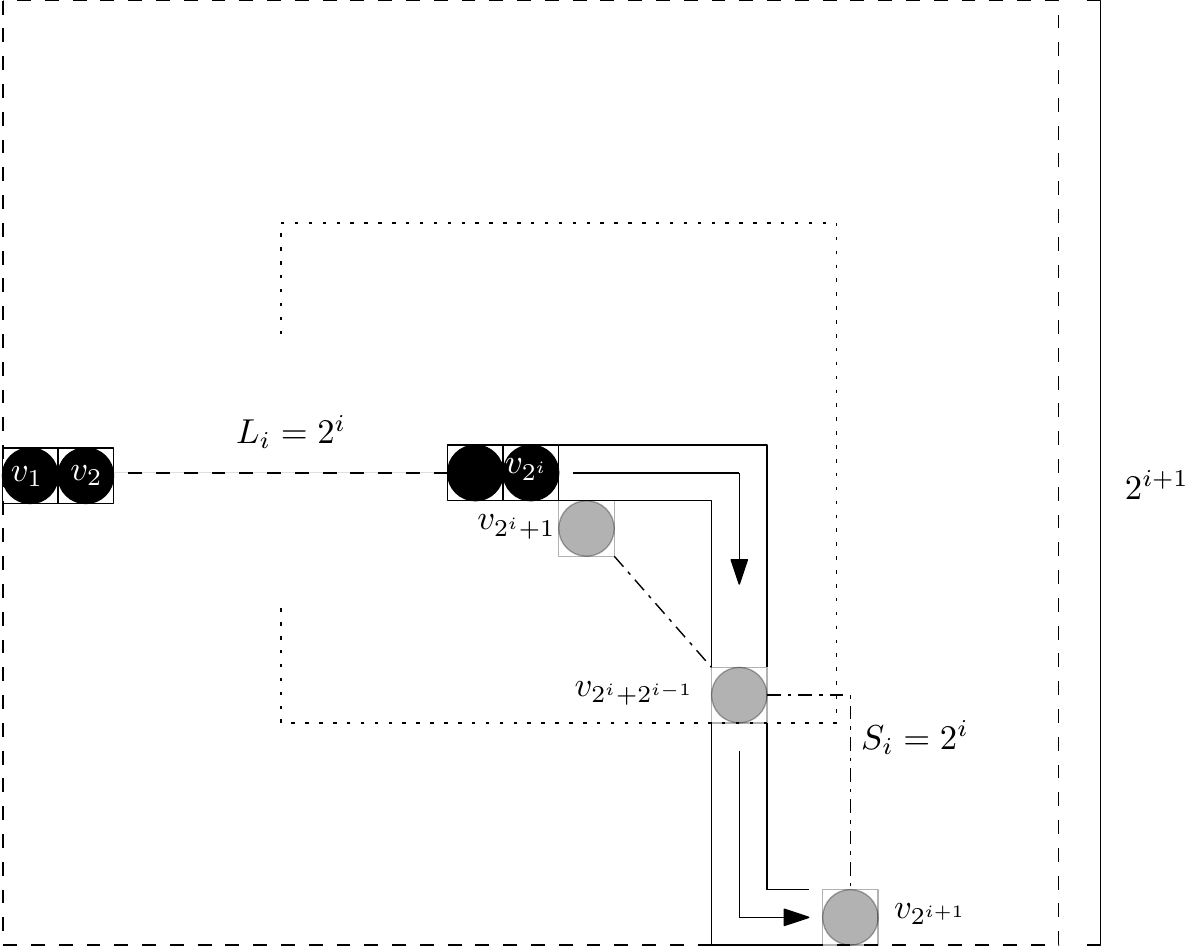}}
	\caption{The two cases of applying \textit{Line-Walk} operation on $L$. }
	\label{fig:StraightLine_SpanningLine}
\end{figure}

As a result of  Lemma \ref{lem:Walk-Through-Path} and \ref{lem:CorrectnessOf_LineWalk}, we obtain the following lemma: 

\begin{lemma} \label{lem:Walk-Through-PathConnectivity}
	Given an initial Hamiltonian shape  $S_{I} \in \mathcal{H}$ of order $n$, {\sc HamiltonianToLine} transforms $S_{I} $ into a straight line $S_L$ in $O(n\log n)$ moves, without breaking connectivity during the transformation. 
\end{lemma}

Now, we are ready to analyse the running time of {\sc HamiltonianToLine}. 

\begin{lemma} \label{lem:Walk-Through-PathRunningTime}
	By the end of phase $i$, for all $0\le i \le \log n$, {\sc HamiltonianToLine} forms a straight line $L$ of $2^i$ nodes in at most $O(n\log n)$ steps, without breaking connectivity of  the whole shape. 
\end{lemma}
\begin{proof}
	 The bound $O(n)$ trivially holds for case (1) and (2), so we analyse a worst-case in which the transformation matches the maximum running time in every phase $i$, for all $1 \le i \le \log n$. In phase $i$, a straight line $L_{i}$ of length $2^{i}$ traverses along a path of at most $2\cdot(2^{i}-2) = 2^{i+1} - 4$ cells in which $L_{i}$ changes its direction twice by at most $ 2^{i+2} - 4$ moves. There is an additive factor of 2 for the special-case of turning $L_{i}$ on a non-empty corner as in Figure \ref{fig:Line_Path_2}. Then the operation of \textit{Line-Walk} takes at total moves of at most:
	 
	 \begin{align*}
	 (t_1)_i  &= (2^{i+1} - 4) + (2^{i+2} - 4)  + 2=  6(2^i -1).
	 \end{align*}
	 
	Next, a recursive call of the algorithm, {\sc HamiltonianToLine}, on $S_i$ of  $2^{i}$ to transform it into a straight line $L^{\prime}_i$, requires  the total sum given by:
	
	\begin{align*}
	(t_2)_i  &= \sum_{i=1}^{i-1} T(i-1).
	\end{align*}	
	
	By the end of phase $i$,  $L_i$ and $L^{\prime}_i$ combine together into a straight line $L_{i+1}$ of  length $2^{i+1}$, in a total cost of at most:
	
	\begin{align*}
	(t_3)_i  &= 2(2^i - 1 ),
	\end{align*}
	
	steps. Hence, {\sc HamiltonianToLine} completes phase $i$ in a total moves $T(i)$ of at most:
	
	\begin{align*}
	T(i)&= (t_1)_i  + (t_2)_i  + (t_3)_i \\
	&= 6(2^i -1)+ \Big(\sum_{i=1}^{i-1} T(i-1)\Big)   +  2(2^i - 1 ) \\
	&\approx 2(2^i)+ \Big(\sum_{i=1}^{i-1} T(i-1)\Big)
	\end{align*}
	
	Now, we compute the recursion of $(t_2)_i $ as follows:
	\begin{align*}
	T(1)  &=  2(2)   \\
	T(2)  &= 2(2^2 )+2(2) = 2(2^2+2)\\  
	. \\
	. \\
	T(i-1) &= 2\Big(  2^{i -1}+ 2^{i-2} + 2(2^{i-3}) + 2^2 (2^{i-4}) + \ldots + 2^{i-4}(2^2)+2^{i-3}(2)\Big)\\
	&< 2\Big(2^{i -1} + 2^{i -1}  + 2^{i -1}  + 2^{i -1}  + \ldots + 2^{i -1} +2^{i -1} \Big)\\
	&= 2\big(  2^{i-1}(i -1) \big).        
	\end{align*}
	Finally, in phase $i$, {\sc HamiltonianToLine} takes a total moves $T(i)$ at most:
	\begin{align*}
	T(i)&= (t_1)_{i} + (t_2)_{i} + (t_3)_{i} \\
	&=   6(2^i -1)+ 2\big(  2^{i-1}(i -1)+ 2(2^i - 1 )\\
	&\le  2^{\log n-1}(\log n-1) - 2^{\log n} = \dfrac{n(\log n-1)}{2} - n =  \dfrac{n\log n-n}{2} - n \\
	&= O(n \log n),
	\end{align*}
	steps. 
\end{proof}

\begin{lemma}\label{lem:Walk-Through-PathTotalRunningTime}
	Given an initial Hamiltonian shape  $S_{I} \in \mathcal{H}$ of order $n$, {\sc HamiltonianToLine} transforms $S_{I} $ into a straight line $S_L$ in $O(n\log n)$ moves, without breaking connectivity during the transformation. 
\end{lemma}     	
\begin{proof}
	By Lemma \ref{lem:Walk-Through-PathRunningTime}, we use induction to analyse the running time of this transformation. The base case is  holds trivially for the first phase. Assume that it holds for phase $i$, and we prove this must hold also for phase $i+1$.  
	\begin{align*}
	T(i+1)&= (2^{(i+1)-1}\big((i+1)-1\big) - 2^{i+1} =  2^i(i) - (2^i\cdot 2) = 2^i (i -2) \\
	&\le 2^{\log n}(\log n -2) = n \log n - 2n\\
	&= O(n \log n).
	\end{align*} 
	The assumption is also true for phase $i$.  Hence, {\sc HamiltonianToLine} makes a total number of moves  bounded by:
	\begin{align*}
	T &= \sum_{i=1}^{\log n} T(i) 
	= \sum_{i=1}^{\log n}  2^{i-1}(i-1) - 2^i 
	= \sum_{i=1}^{\log n-1}(i-2)2^i - 2^{\log n}\\
	&\le \sum_{i=1}^{\log n-1} i \cdot 2^i - n
	\le  \sum_{j=1}^{\log n} \sum_{i=j}^{\log n} 2^i -n 
	\le \sum_{j=1}^{\log n} n -  n \le n \log n - n\\
	&\le O(n\log n).
	\end{align*}

\end{proof}

Finally, reversibility of  line moves \cite{AMP19}, Lemmas \ref{lem:Walk-Through-PathConnectivity} and \ref{lem:Walk-Through-PathTotalRunningTime} together imply that:

\begin{theorem} \label{theo:Walk_Through}
	For any pair of Hamiltonian shapes $S_I, S_F \in \mathcal{H}$ of the same order $ n $, \textit{Walk-Through-Path} transforms $S_I$ into $S_F$ (and $S_F$  into $S_I$) in $ O(n\log n) $ moves, while preserving connectivity of the shape during its course.
\end{theorem}

\section{$O(n \sqrt{n})$-time Universal Transformation}
\label{sec:nsqrtn_Universal_Transformation}

In this section, we introduce a transformation that solves the {\sc UniversalConnected} problem in $O(n \sqrt{n})$ moves. It is called \emph{UC-Box} and transforms any pair of connected shapes $(S_I, S_F)$ of the same order to each other, while preserving \textit{connectivity} during its course.

Starting from the initial shape $S_I$ of order $n$ with an associated graph $G(S_I)$, compute a spanning tree $T$ of $G(S_I)$. Then enclose the shape into an $n\times n$ square box and divide it into $\sqrt{n}\times \sqrt{n}$ square sub-boxes. Each occupied sub-box contains one or more maximal sub-trees of $T$. Each such sub-tree corresponds to a sub-shape of $S_I$, which from now on we call a \emph{component}. Pick a leaf sub-tree $T_l$, let $C_l$ be the component with which it is associated, and $B_l$ their sub-box. Let also $B_p$ be the sub-box adjacent to $B_l$ containing the unique parent sub-tree $T_p$ of $T_l$. Then compress all nodes of $C_l$ into $B_p$ through line moves, while keeping the nodes of $C_p$ (the component of $T_p$) within $B_p$. Once compression is completed and $C_p$ and $C_l$ have been \emph{combined} into a single component $C^{\prime}_{p}$, compute a new sub-tree $T^{\prime}_{p}$ spanning $G(C^{\prime}_{p})$. Repeat until the whole shape is compressed into a $\sqrt{n}\times \sqrt{n}$ square. The latter belongs to the family of \textit{nice} shapes (a family of connected shapes introduced in \cite{AMP19}) and can, thus, be transformed into a straight line in linear time.  

Given that, the main technical challenges in making this strategy work universally is that a connected shape might have many different configurations inside the sub-boxes it occupies, while the shape needs to remain connected during the transformation. In the following, we describe the \textit{compression} operation, which successfully tackles all of these issues by exploiting the linear strength of line moves.  

\textbf{\textit{Compress.}} Let  $C_l\subseteq S_I$  be a leaf  component containing nodes  $v_1,\ldots, v_{k}$  inside a sub-box $B_l$ of size $\sqrt{n} \times \sqrt{n} $, where $1\le k \le n$, and $C_p\subseteq S_I$ the unique parent component of $C_l$ occupying an adjacent sub-box $B_p$. If the direction of connectivity between $B_l$ and $B_p$ is vertical or horizontal, push all lines of $C_l$  one move towards $B_p$ sequentially one after the other, starting from the line furthest from $B_p$. Repeat the same procedure to first align all lines perpendicularly to the boundary between $B_l$ and $B_p$ (Figure \ref{fig:DTransferB1ToB2}(b)) and then to transfer them completely into $B_p$ (e.g., Figure \ref{fig:DTransferB1ToB2}(c)). Hence, $C_l$ and $C_p$ are combined into $C^{\prime}_{p}$, and the next round begins.  The above steps are performed in a way which ensures that all lines (in $C_l$ or $C_p$) which are being pushed by this operation do not exceed the boundary of $B_p$ (e.g., Figure \ref{fig:DTransferB1ToB2}(d)). While $C_l$ compresses vertically or horizontally, it may collide with a component $C_r \subseteq S_I$  inside $B_l$. In this case $C_l$ stops compressing and combines with $C_r$ into $C^{\prime}_{r}$. Then the next round begins. If $C_l$ compresses diagonally towards $C_p$ (vertically then horizontally or visa versa) via an intermediate adjacent sub-box $B_m$ and collides with $C_m \subseteq S_I$ inside $B_m$, then $C_l$ completes compression into $B_m$  and combines with $C_m$ into $C^{\prime}_{m}$.  Figure \ref{fig:DTransferB1ToB2} shows how to compress a leaf component into its parent component occupying a diagonal adjacent sub-box.

Examples \ref{ex:HandVCompress} and \ref{ex:DiaginalCompress} depict the compression in different directions. The formal description of \emph{UC-Box} is illustrated in Algorithm \ref{algo}.    

\begin{example} [Horizontal and vertical compression] \label{ex:HandVCompress}
	Let $C_l$ and $C_p$ be components occupying two horizontal sub-boxes, $B_l$ and $B_p$, respectively. $C_l$ transfers completely to join $C_p$ in $B_p$, as in Figure \ref{fig:TransferB1ToB2}. The vertical compression holds by rotating the system $90\degree$ clockwise or counter-clockwise. 
	\begin{figure}[th!]
		\centering	
		\subcaptionbox{}	
		{\includegraphics[scale=0.39]{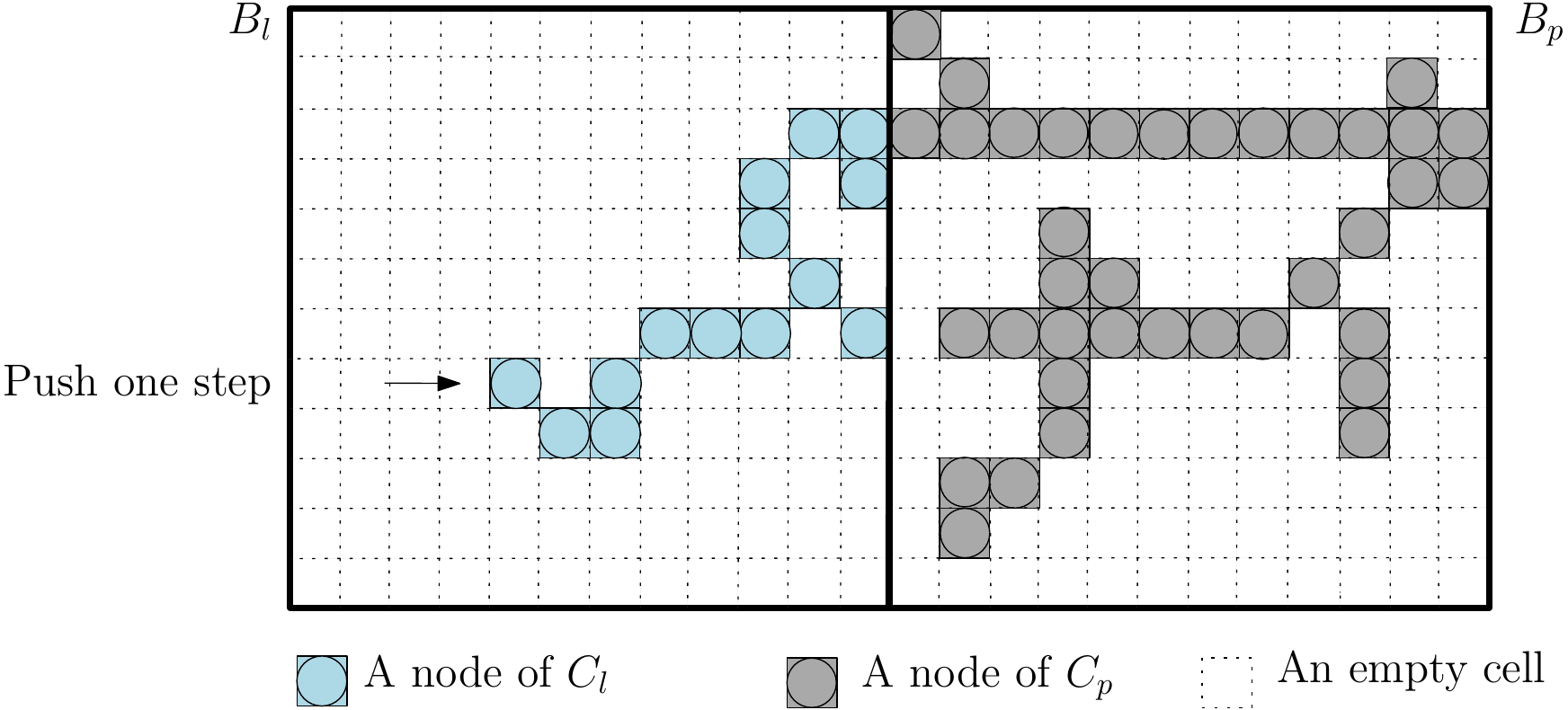}}
		\subcaptionbox{}
		{\includegraphics[scale=0.39]{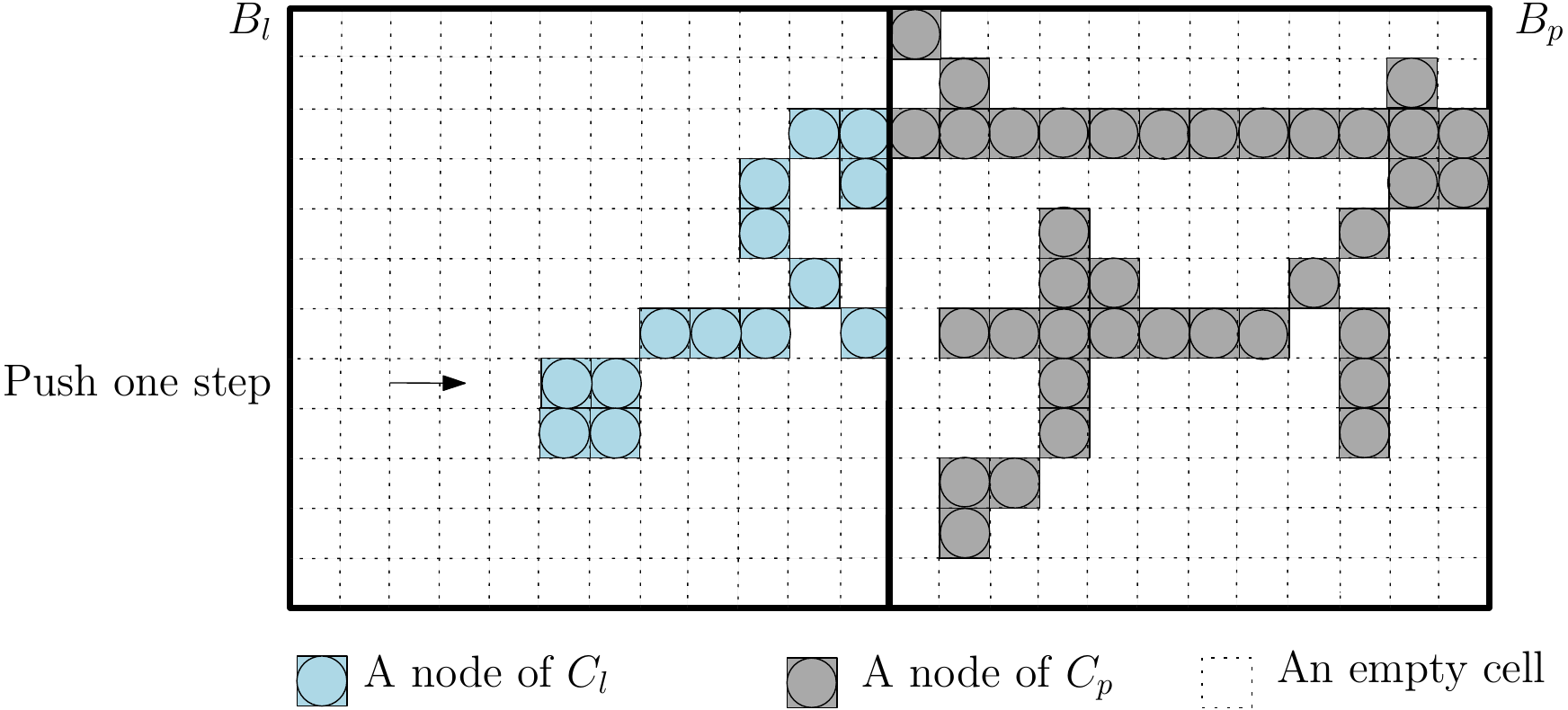}}
		\subcaptionbox{}	
		{\includegraphics[scale=0.38]{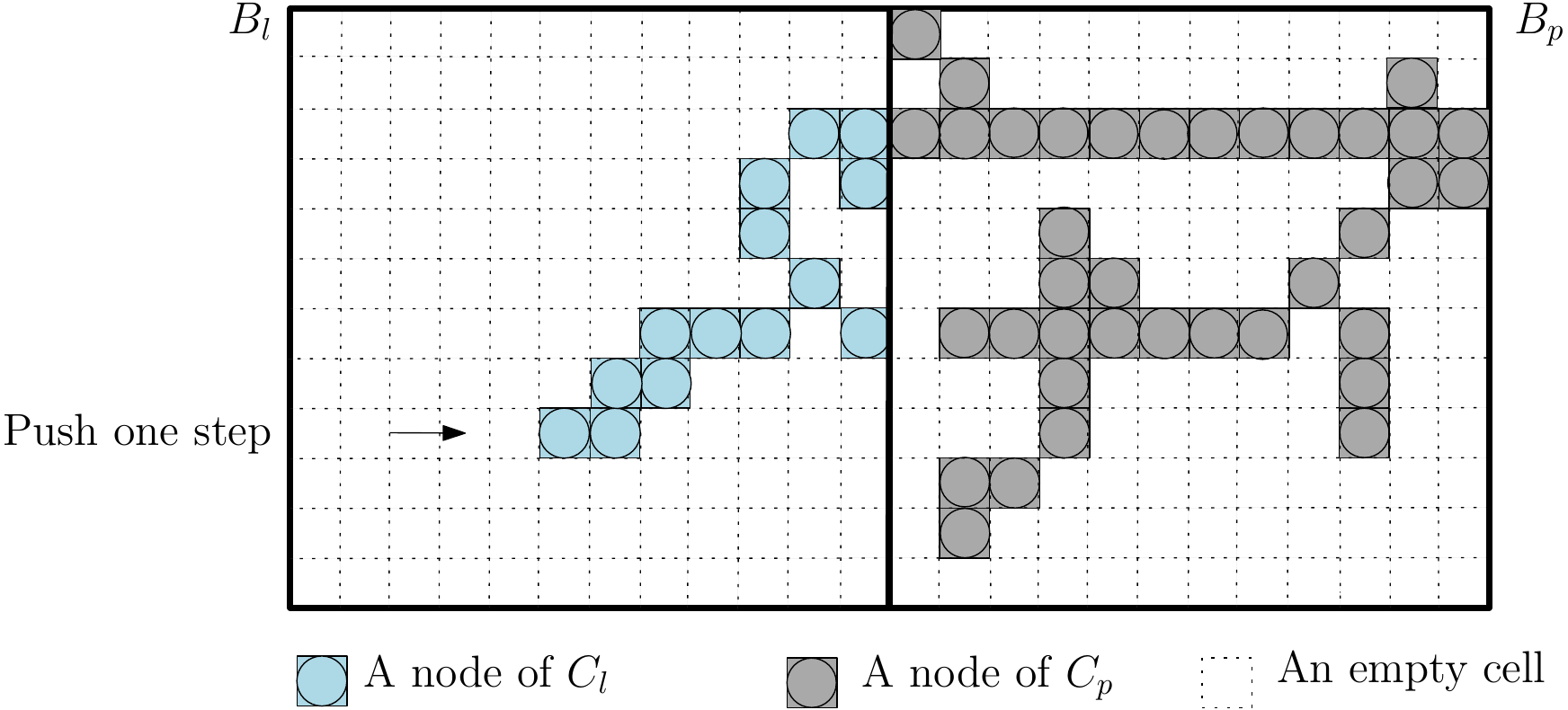}}
		\subcaptionbox{}	
		{\includegraphics[scale=0.38]{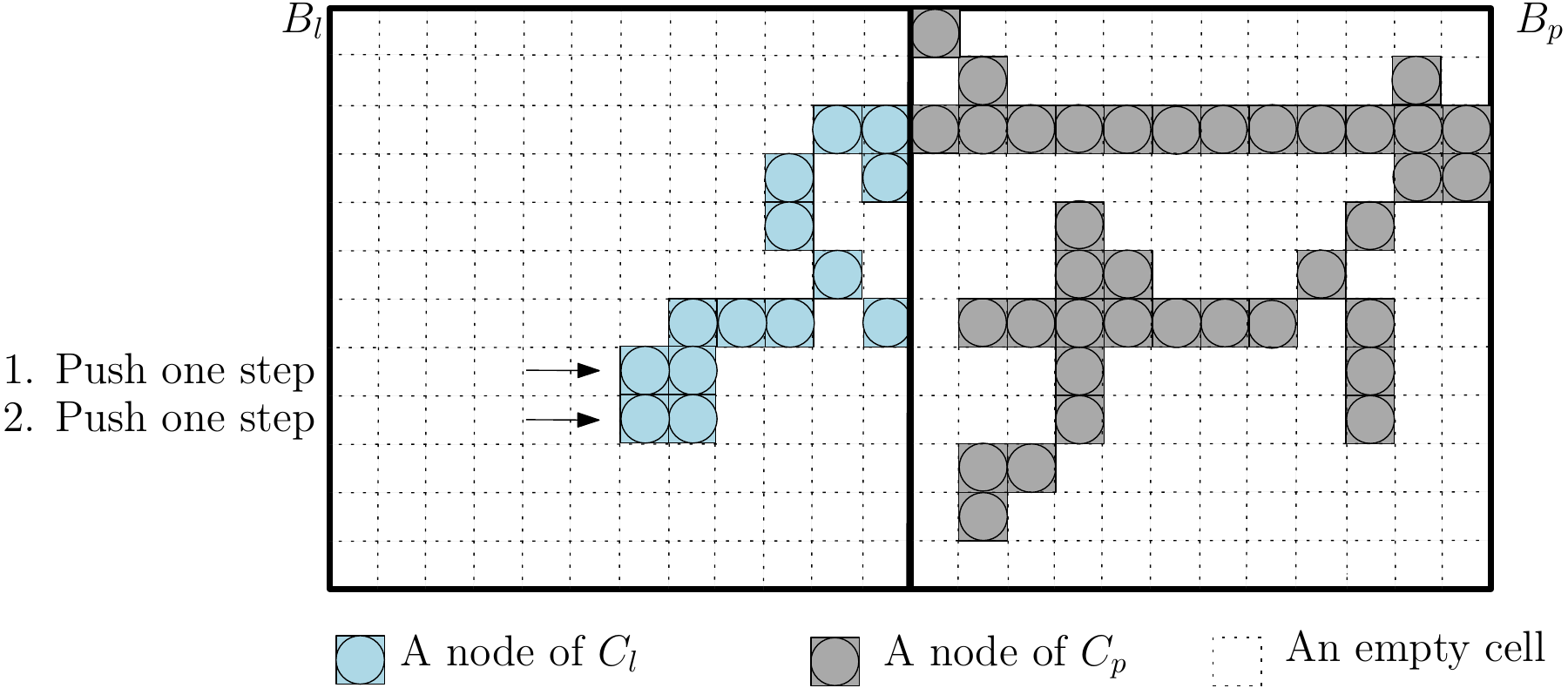}}	
		\subcaptionbox{}
		{\includegraphics[scale=0.38]{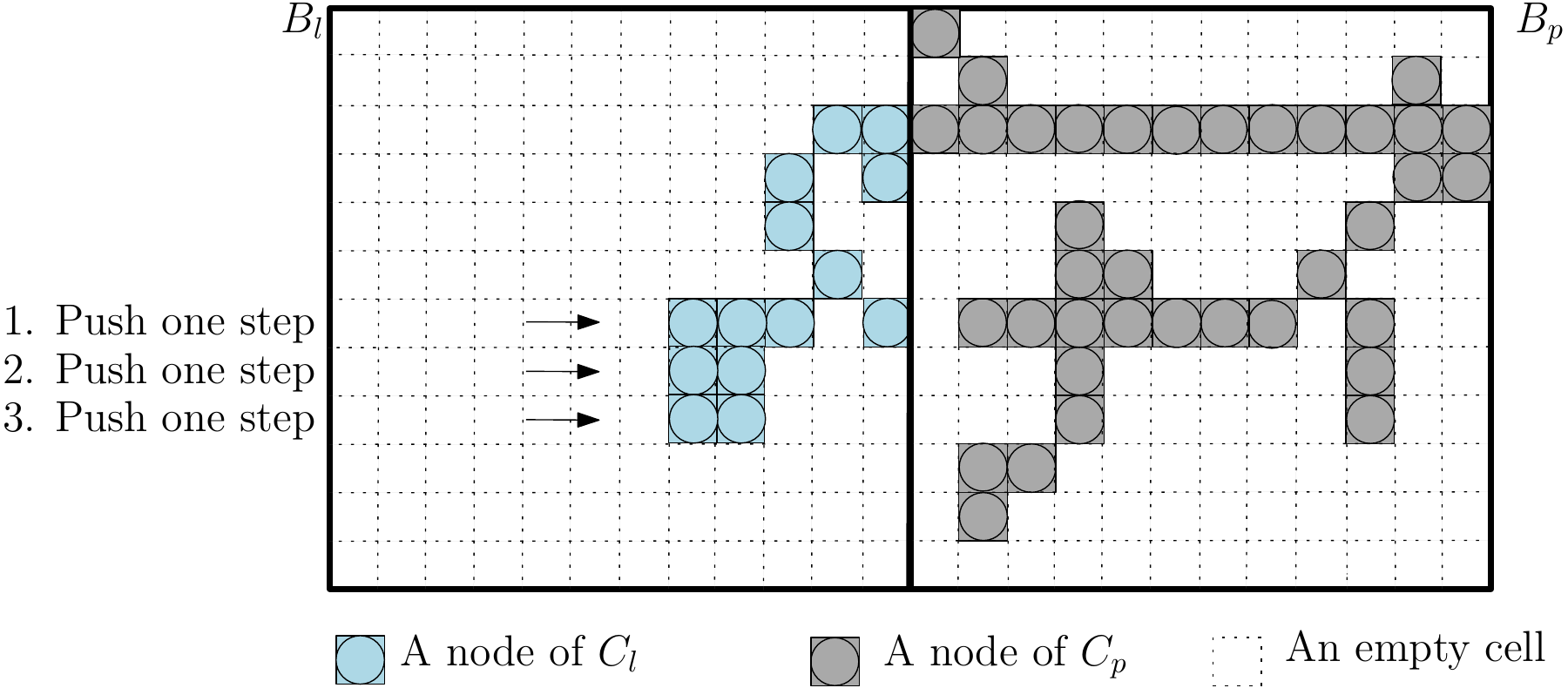}}
		\subcaptionbox{}
		{\includegraphics[scale=0.38]{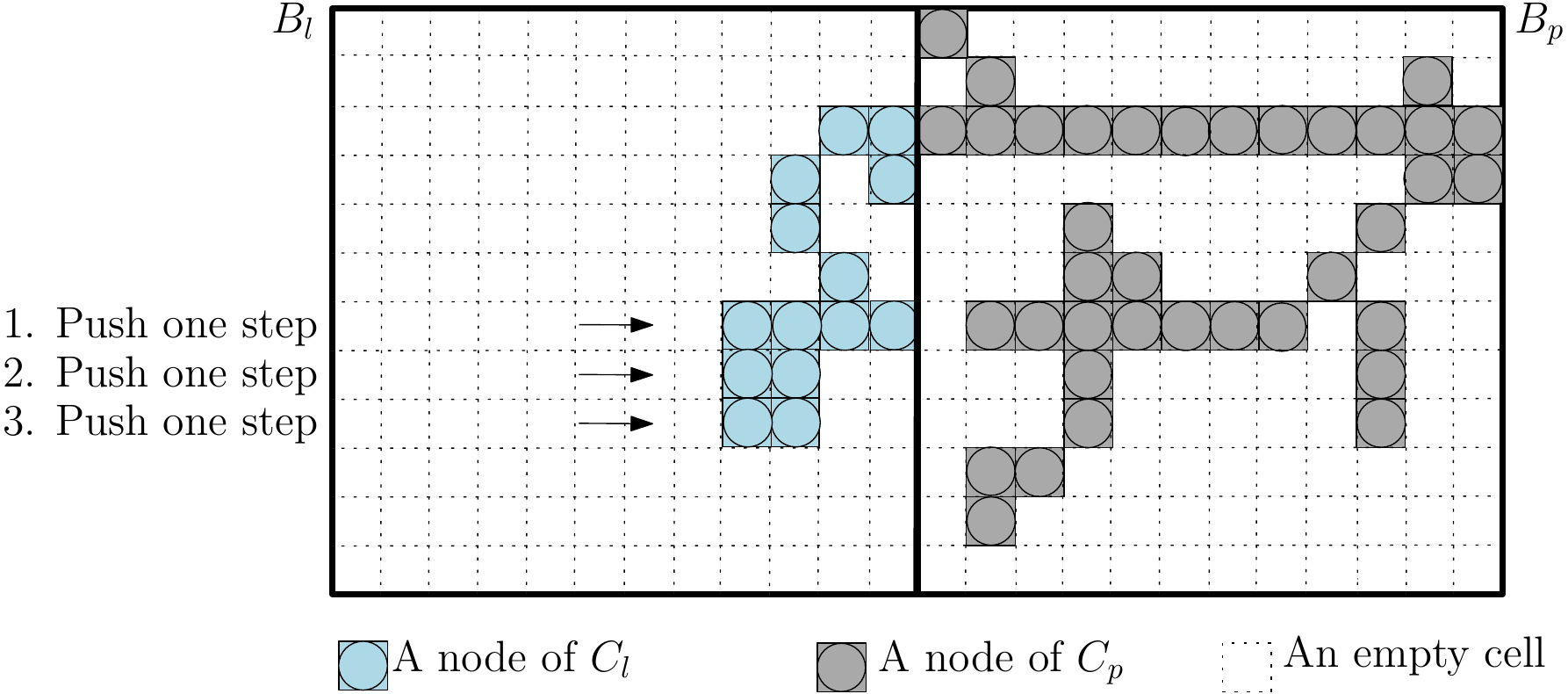}}
		\subcaptionbox{}
		{\includegraphics[scale=0.38]{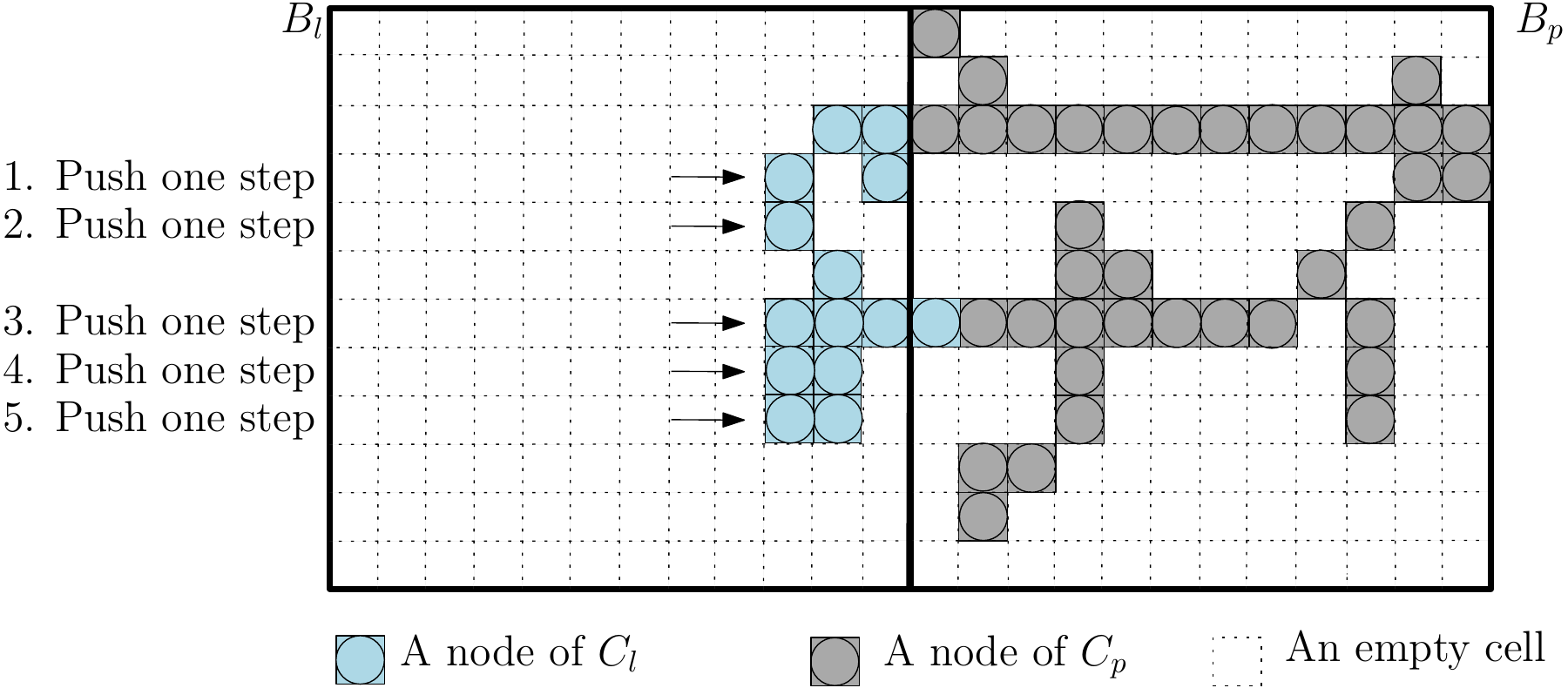}}
		\subcaptionbox{}
		{\includegraphics[scale=0.38]{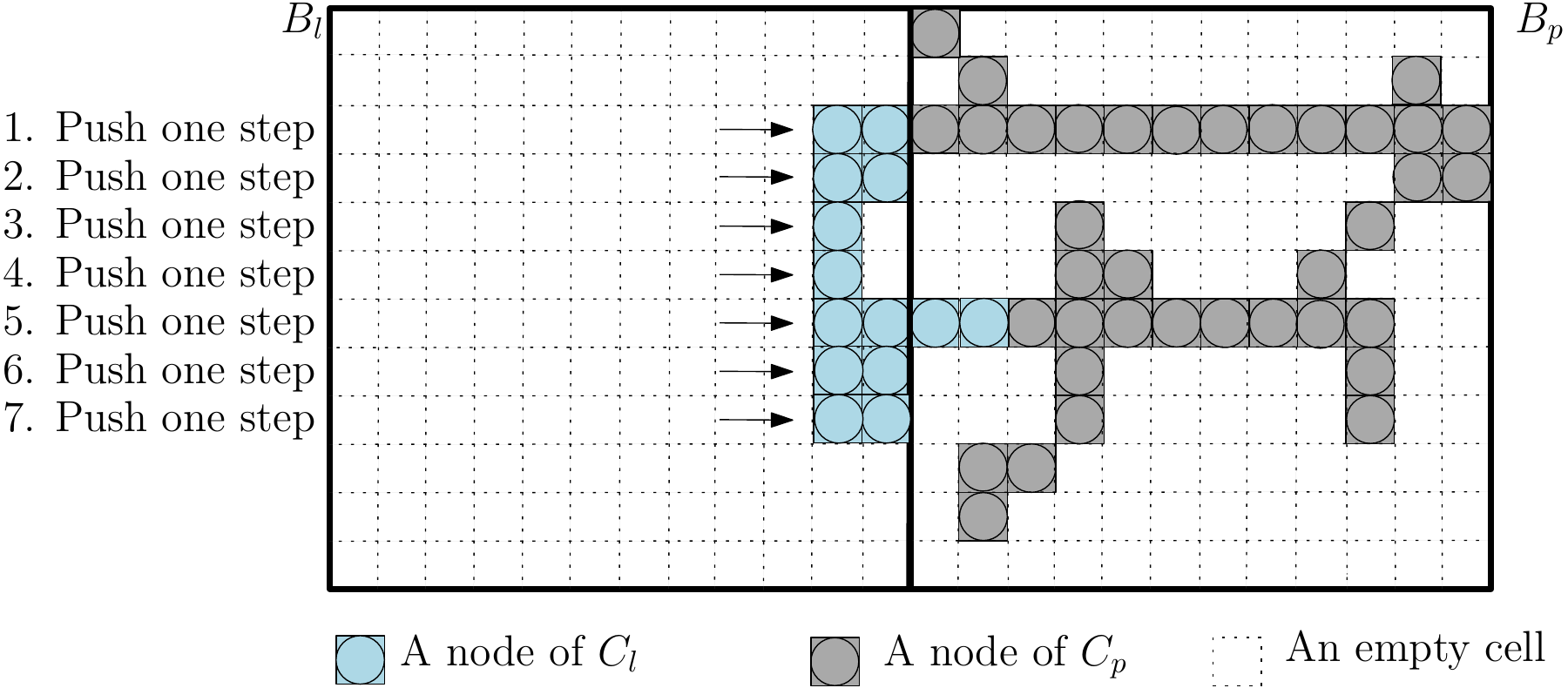}}
		\subcaptionbox{}
		{\includegraphics[scale=0.38]{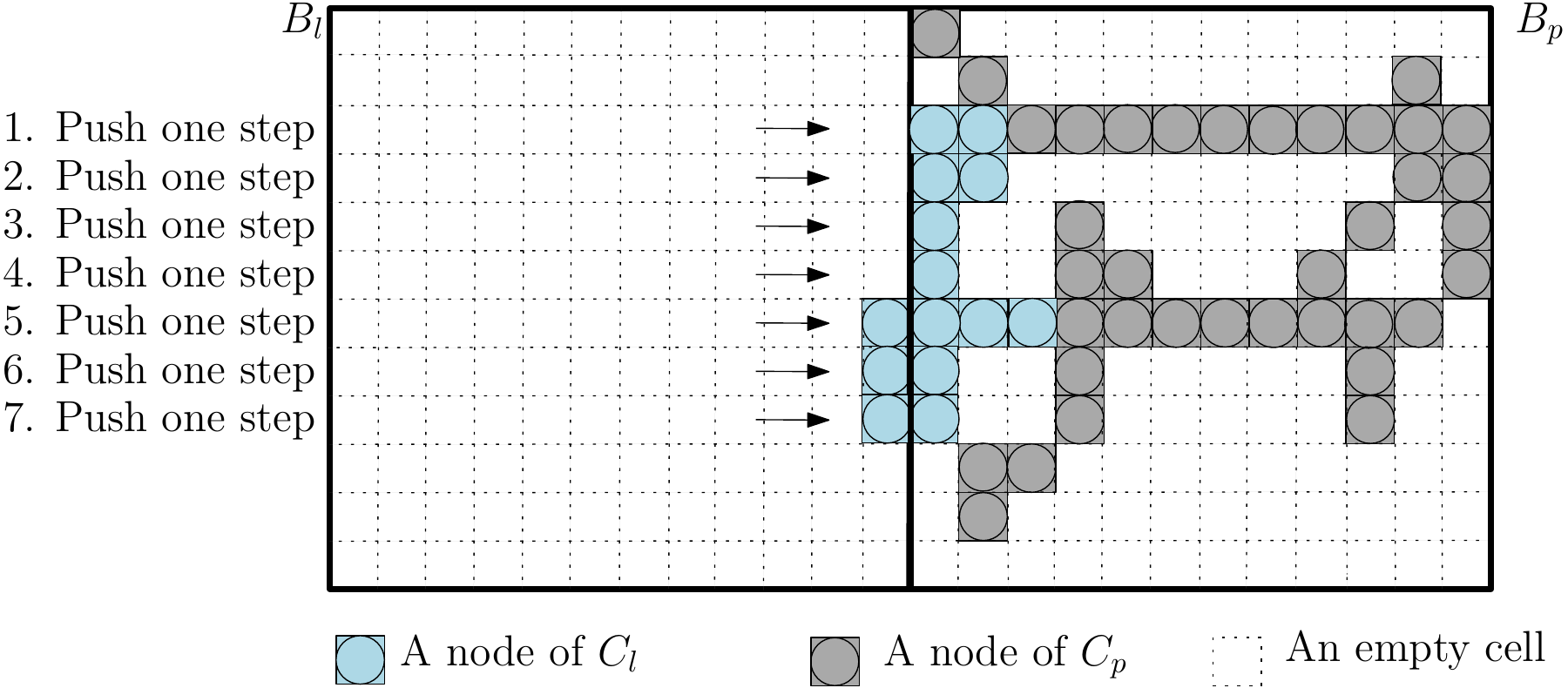}}
		\subcaptionbox{}
		{\includegraphics[scale=0.38]{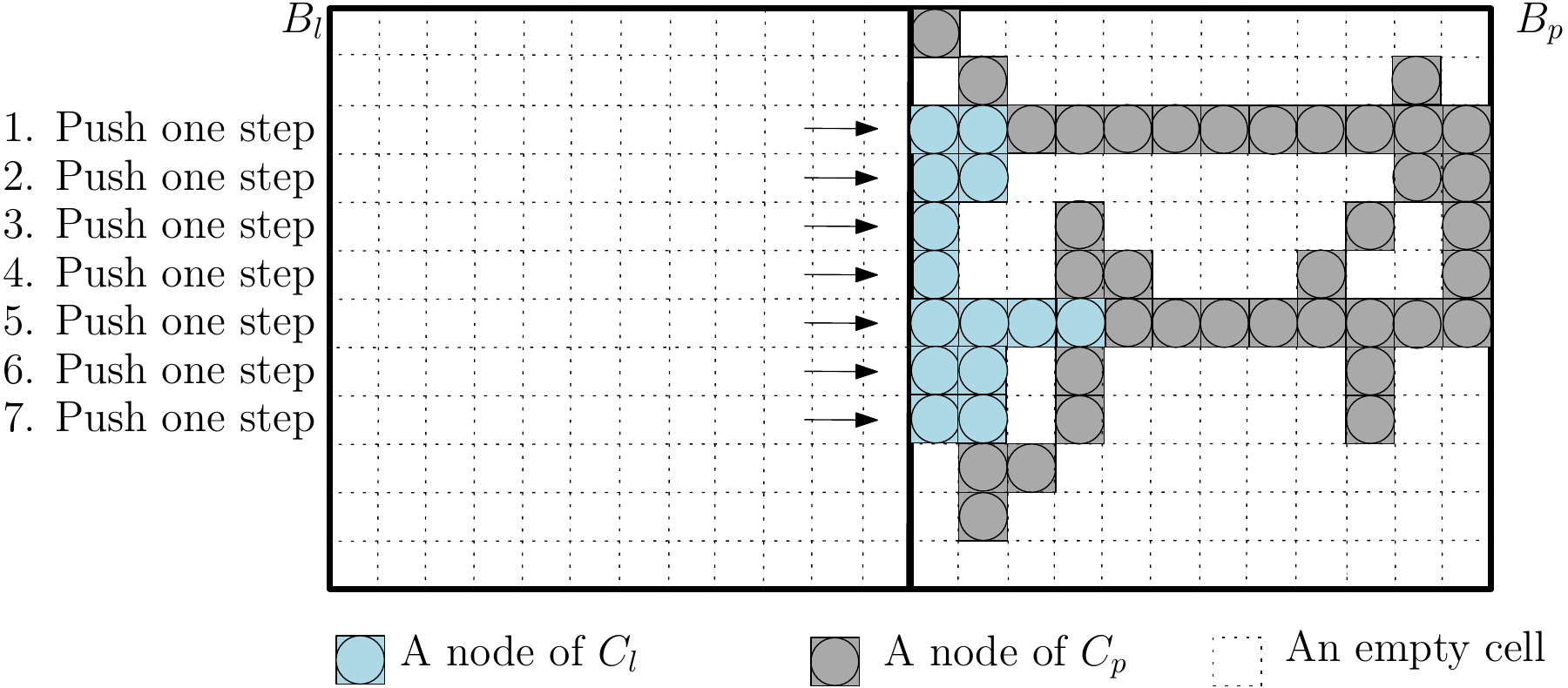}}
		\caption{Horizontal and vertical compression.}
		\label{fig:TransferB1ToB2}
	\end{figure}		
\end{example}
\begin{example}  [Diagonal compression] \label{ex:DiaginalCompress}
	Let $C_l$ and $C_p$ be components occupying two sub-boxes, $B_l$ and $B_p$ that are connected diagonally, respectively. $C_l$ transfers completely via an intermediate sub-box $B_m$, as shown in Figure \ref{fig:DTransferB1ToB2}.
	\begin{figure}[th!]
		\centering		
		\subcaptionbox{}
		{\includegraphics[scale=0.35]{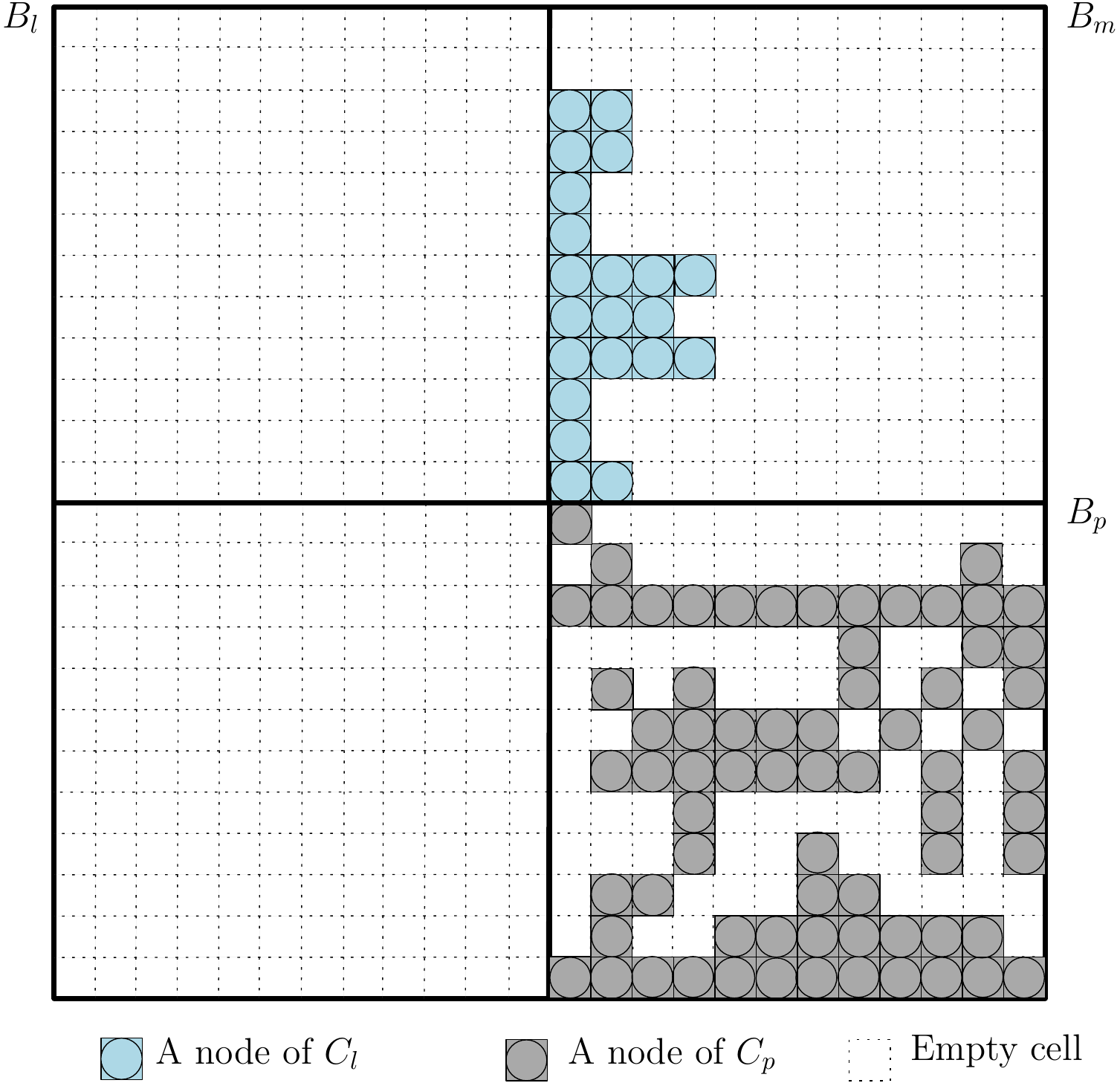}} \qquad 
		\subcaptionbox{}
		{\includegraphics[scale=0.35]{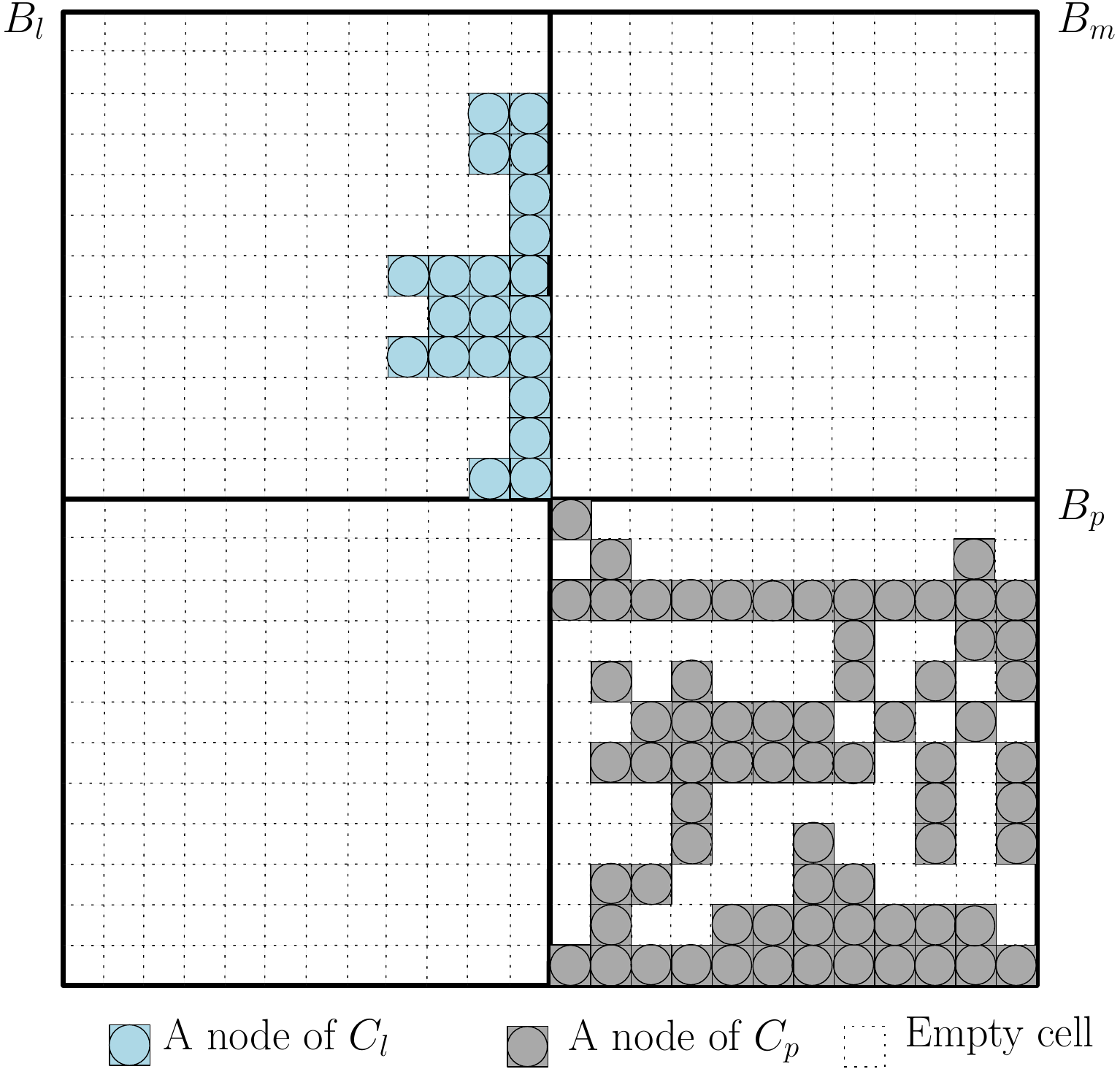}} 
		\subcaptionbox{}
		{\includegraphics[scale=0.35]{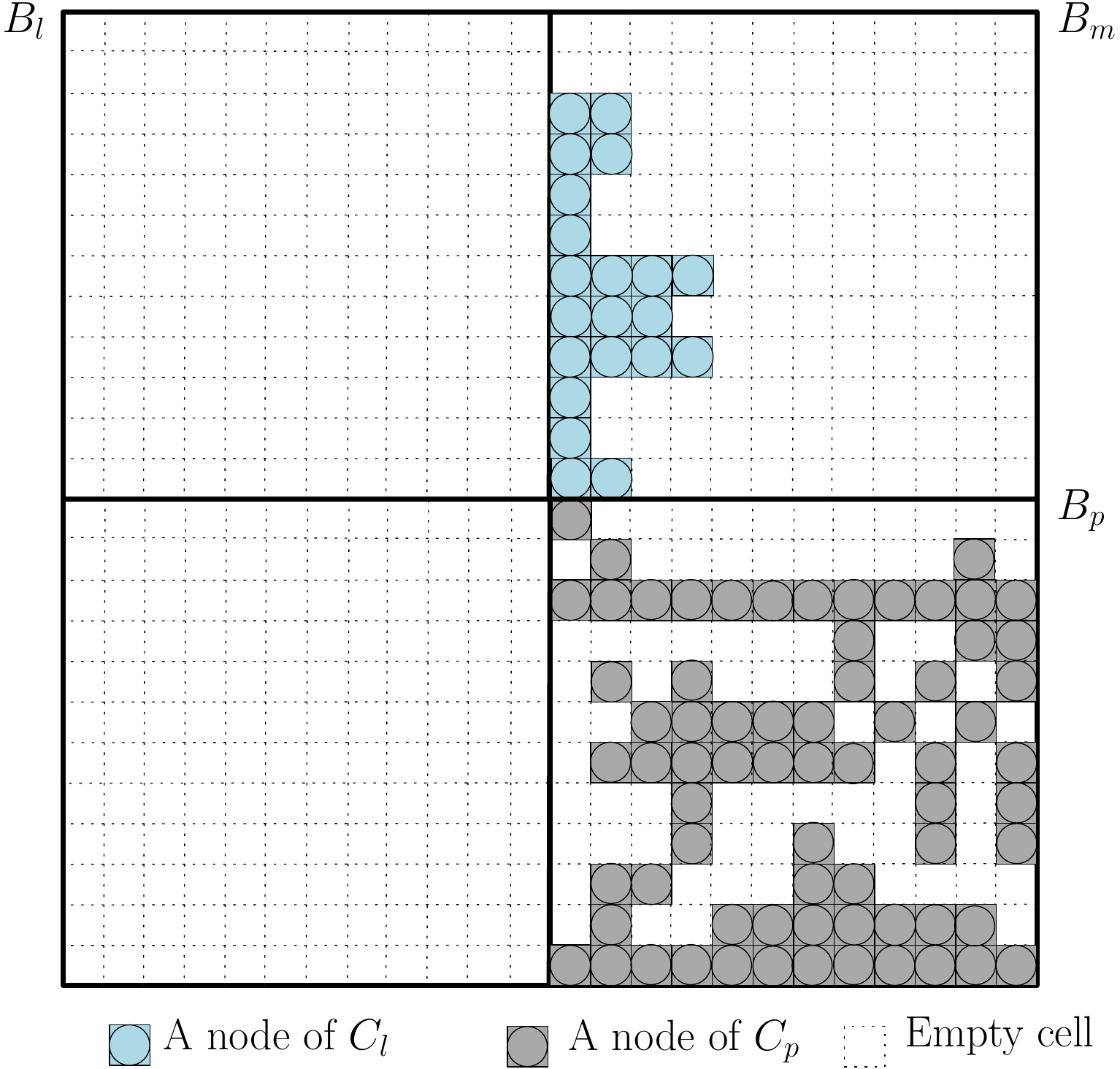}} \qquad 
		\subcaptionbox{}	
		{\includegraphics[scale=0.35]{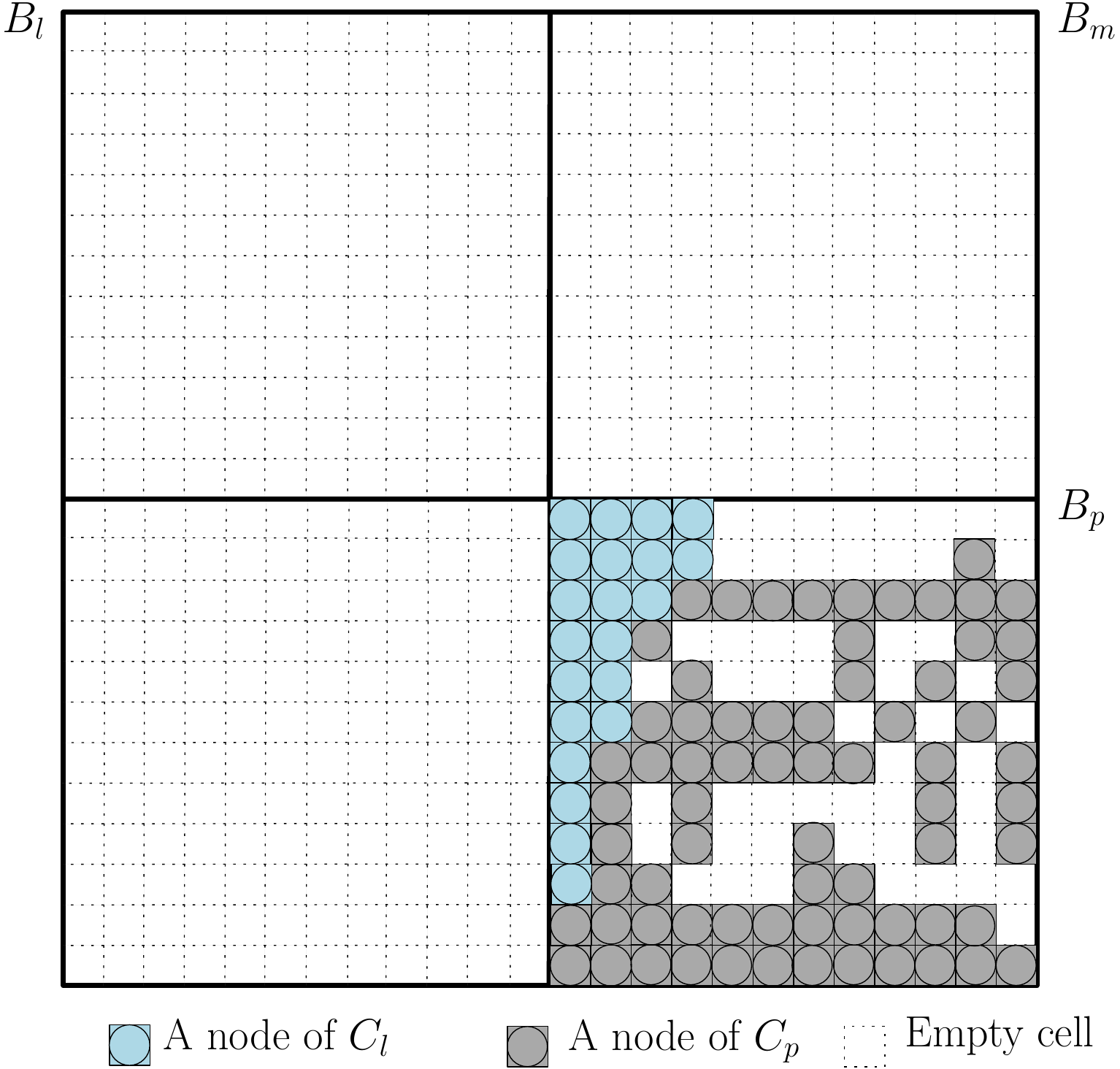}}	
		\caption{Diagonal compression.}
		\label{fig:DTransferB1ToB2}
	\end{figure}		
\end{example}

\begin{algorithm}[H]
	$ S  = (u_1,u_2,...,u_{|S|}) $ is a connected shape, $T$ is a spanning tree of $G(S)$
	\SetAlgoLined
	\DontPrintSemicolon
	
	\vspace{7pt} 
	\Repeat{\textsf{the whole shape is compressed into a $\sqrt{n}\times \sqrt{n}$ square}}
	{
		$C_l$ $\gets$ pick($T_l$)  \tcp{select a leaf component associated with a leaf sub-tree}
		Compress($C_l$)        \tcp{start compressing the leaf component}
		\eIf{$C_l$ collides} 
		{$C^{\prime}_{r}$ $\gets$ combine($C_r, C_l$)  or  $C^{\prime}_{m}$ $\gets$ combine($C_m, C_l$) \tcp{as described in text} }
		{$ C^{\prime}_{p}$ $\gets$ combine($C_p, C_l$) \tcp{combine $C_l$ with a parent component}} 
		update($T$)   \tcp{update sub-trees and remove cycles after compression}
	}
	
	\KwOut{a square shape $S_C$}
	\caption{{\sc Compress}($ S $)}
	\label{algo}
\end{algorithm}

Algorithm \ref{algo}, {\sc Compress}, provides a universal procedure to transform an initial connected shape $S_I$ of any order into a compressed square shape of the same order. It takes two arguments: $S_I$ and the spanning tree $T$ of the \textit{associated graph} $G(S_I)$. In any round: Pick a leaf sub-tree of $T_l$ corresponding to $C_l$ inside a sub-box $B_l$. Compress $C_l$ into an adjacent sub-box $B_p$ towards its parent component $C_p$ associated with parent sub-tree $T_p$. If $C_l$ compressed with no collision, perform combine($C_p, C_l$) which combines $C_l$ with $C_p$ into one component $ C^{\prime}_{p} $. If $C_l$ collides with another component $C_r $ inside $B_l$, then perform combine($C_r, C_l$) into $C^{\prime}_{r}$. If not, as in the diagonal compression in which $C_l$ collides with $C_m $  in an intermediate sub-box $B_m$, then $C_l$ compresses completely into $B_m$ and performs combine($C_m, C_l$) into  $C^{\prime}_{m}$. Once compression is completed, update($T$) computes a new sub-tree and removes any cycles. The algorithm terminates when $T$ matches a single component of $n$ nodes compressed into a single sub-box.

\subsection{Correctness}
In this section, we show that all properties of connectivity-preserving, transformability and universality hold in \emph{UC-Box}, which is capable to transform any pair of connected shapes $(S_I, S_F)$ of the same order $n$ to each other, without breaking connectivity during its course. 

Given a an initial connected shape $S_I$ holding $n$ nodes, then $S_I$ can be always bounded by a square box of size $n \times n$, placed in a appropriate position to include all nodes in $S_I$. This box can be divided into at most $\sqrt{n}$ sub-boxes (proved in \cite{AMP19}), $B_1, B_2, \cdots , B_{\sqrt{n}}$, of size $\sqrt{n} \times \sqrt{n}$, each occupied sub-box may contain one or more sub-shapes (called components) of at least one node $u \in S_I$.  As the shape $S_I$ is \emph{connected}, all occupied sub-boxes are \emph{connected} too. This relation of connectivity can be defined as follows;

\begin{definition}[Connectivity of sub-boxes]\label{def:Sub-boxesConnectivity} 
	By the above partitioning, two occupied sub-boxes, $B_1$ and $B_2$, are \textit{connected} iff there are two distinct nodes $u_1,u_2 \in S_I$, such that $u_1$ occupies $B_1$  and $u_2$ occupies $B_2$ where $u_1$ and $u_2$ are two adjacent neighbours connected vertically,  horizontally or diagonally. 
\end{definition}

Next, we define connectivity between the components (sub-shapes).

\begin{definition}[Connectivity of components]\label{def:ConnectivityBetweenComponenets} 
	By the above partitioning, two connected components, $C_1, C_2  \in S_1$ are \textit{connected} iff there are two distinct elements $u \in C_1$ and $v \in C_2$, such that $u$ and $v$ are two adjacent neighbours connected vertically,  horizontally or diagonally. 
\end{definition} 

\begin{corollary} \label{oberv:ShapeToTree} 
	Given the above partitioning dividing $S_I$ into a number of components. Then, it holds that all components can be computed into a spanning tree $T$.           	
\end{corollary}

In the following lemmas, we prove that any connected shape $S_1$ of $n$ nodes can be compressed into a square box of dimension $ \sqrt{n} $.

\begin{lemma} \label{lem:ConCompInSubBox}
	Any square box of size $ \sqrt{n} $ can hold at most $2\sqrt{n}$ connected components.    
\end{lemma}
\begin{proof}
	Assume $S_I$ is a connected shape enclosed by a box of size $n$ that is partitioned into $\sqrt{n}$ square sub-boxes of dimension $\sqrt{n}$. Then, a component $C\subseteq S_I$ of at least 1 node can occupy a sub-box, $B$. The component $C$ must be connected to one of the four length-$ \sqrt{n} $ boundaries of $B$. Assume for the sake of contradiction that  $C$ is not connected to any boundaries. This means that $S_1$ is disconnected and therefore $C \nsubseteq S_I$, which contracts our assumption. Observe that based in our setting, $C$ can be connected via a path to any of the four length-$ \sqrt{n} $ boundaries through at most $ \sqrt{n}/2 $ cells, as shown in Figure \ref{fig:ConCompInSub-box}. Thus, one boundary can hold $ \sqrt{n}/2  $ distinct components, resulting in $ 2\sqrt{n} $ for the four boundaries.  Therefore, the sub-box $B$ can contain at most $2\sqrt{n} $ disconnected components. 
	\begin{figure}[th!]
		\centering
		{\includegraphics[scale=0.6]{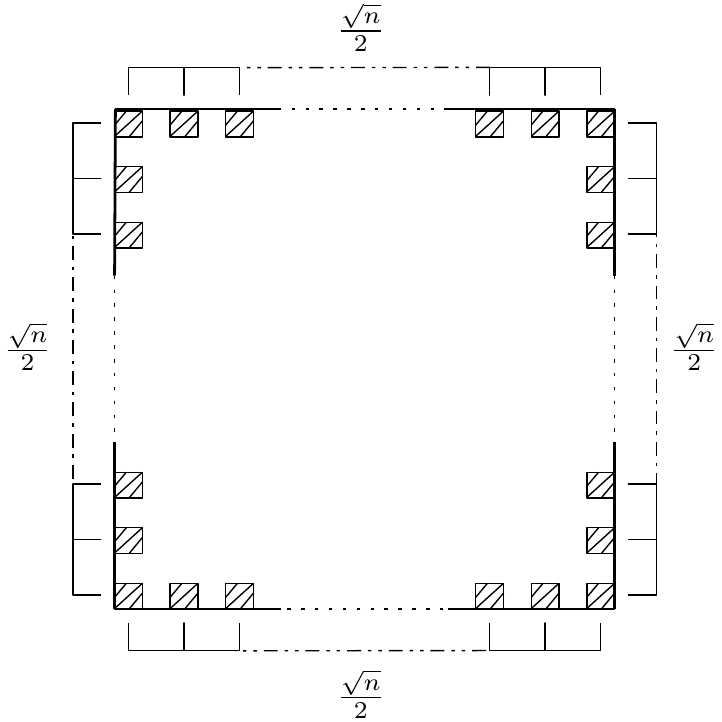}}	
		\caption{A square box of four length-$ \sqrt{n} $ boundaries, each  can hold up to $2\sqrt{n}$ different components.}
		\label{fig:ConCompInSub-box}
	\end{figure}
\end{proof} 

\begin{lemma} \label{lem:PossibleityofCom}
	Let $S_{I}$ be a connected shape of order $n$ occupies $ \sqrt{n}$  sub-boxes of size $ \sqrt{n} \times \sqrt{n} $ each. Then, it is always possible to compress all $n$ nodes into a single sub-box.    
\end{lemma}
\begin{proof}
	It is sufficient to show that the number of cells inside any sub-box, $\sqrt{n} \times \sqrt{n} = n$ is enough to be filled by at most $n$ nodes.
\end{proof} 

Now, we show that transformation \emph{UC-Box} form a nice shape by the end of the final phase. 

\begin{lemma} \label{lem:CompNice}
	Starting from any connected shape $S_{I}$ of order $n$, {\sc Compress} forms a nice shape of order $n$.   
\end{lemma}
\begin{proof}
	The strategy will eventually compress all components of $ n \in S_I$ nodes into a $\sqrt{n} \times \sqrt{n}$ square sub-box. Regardless on which sub-box the final compressing is, the resulting final shape will be a compressed square of size $\sqrt{n}$, which is a nice shape. 
\end{proof}

\begin{lemma} \label{lem:CompressingPreservesConnectivity}
	Starting from an initial connected shape $S_I$ of order $n$ divided into $\sqrt{n}$ square sub-boxes of size $\sqrt{n}$, {\sc Compress} compresses a leaf component $C_l \subseteq S_I$ of $k \ge 1$ nodes, while preserving the global connectivity of the shape. 
\end{lemma}
\begin{proof}
	Given an initial connected shape $S_{I}$ of order $n$ enclosed into a box of length $n$, which is divided into  $\sqrt{n} $ sub-boxes of size $\sqrt{n} \times \sqrt{n}$, each occupied sub-box contains at least one component of a total $C$, for all  $1 \le C \le n$. By Corollary \ref{oberv:ShapeToTree}, $S_I$ is computed into a spanning tree $T = (V, E)$ of its associated graph $G(S_I)$, where $V$ represents components $C$ inside the sub-boxes and $E$ is the \textit{neighbouring relation of connectivity} between those sub-boxes (see Definitions \ref{def:Sub-boxesConnectivity} and \ref{def:ConnectivityBetweenComponenets}). Say that a component $C_l \in C$, occupies a sub-box $B_l$ and represented by a leaf $v \in V$, compresses into a parent component $C_p$ occupies an adjacent sub-box $B_p$ and corresponds to a parent $u \in V$. We shall discuss all possible cases of moving all $k \in C_l$ lines from $B_l$ towards $B_p$ vertically, horizontally and diagonally, for all $1 \le k \le \sqrt{n}$. 
	Due to symmetry, we only present all transformations in one direction, which holds for all other directions by rotating the shape 90\degree, 180\degree, and 270\degree. 
	
	Assume a left $B_l$ and right sub-box $B_p$ are connected horizontally. Then, all horizontal lines (rows) $k \in C_l$ push a single move right towards $B_p$ sequentially one after the other, starting from the furthest line from the boundary between $B_l$ and $B_p$. A single line $l \in k$ of length $i$, $1 \le i \le \sqrt{n}$, can occupy a row in $B_l$ in one of the following cases: 
	
	\begin{itemize}
		\item \textbf{Case 1.} The line $l$ of length $\sqrt{n}$ starts from the left and finishes at the right boundary of $B_l$. Regardless of the current configuration, $l$ pushes one move the right from $(x,y), \ldots, (x+\sqrt{n}, y)$ to $(x+1,y), \ldots, (x+\sqrt{n}+1, y)$ and decreases its length by 1. This move is just like simple position permutations of the $l$'s elements to their right neighbours positions. As a result, $l$ stays connected to any nodes at cells $(x,y\pm1), \ldots, (x+\sqrt{n}, y\pm1)$, creates an empty cell at $(x+1,y)$ and dose not break \textit{connectivity} of all other lines in $S_I$. See an example in Figure \ref{fig:CompressingPreservesConnectivity} (a) and (b).
		\begin{figure}[th!]
			\centering
			\subcaptionbox{A line $l$ of length $\sqrt{n}$ occupyis a whole row in $B_l$.}
			{\includegraphics[scale=0.4]{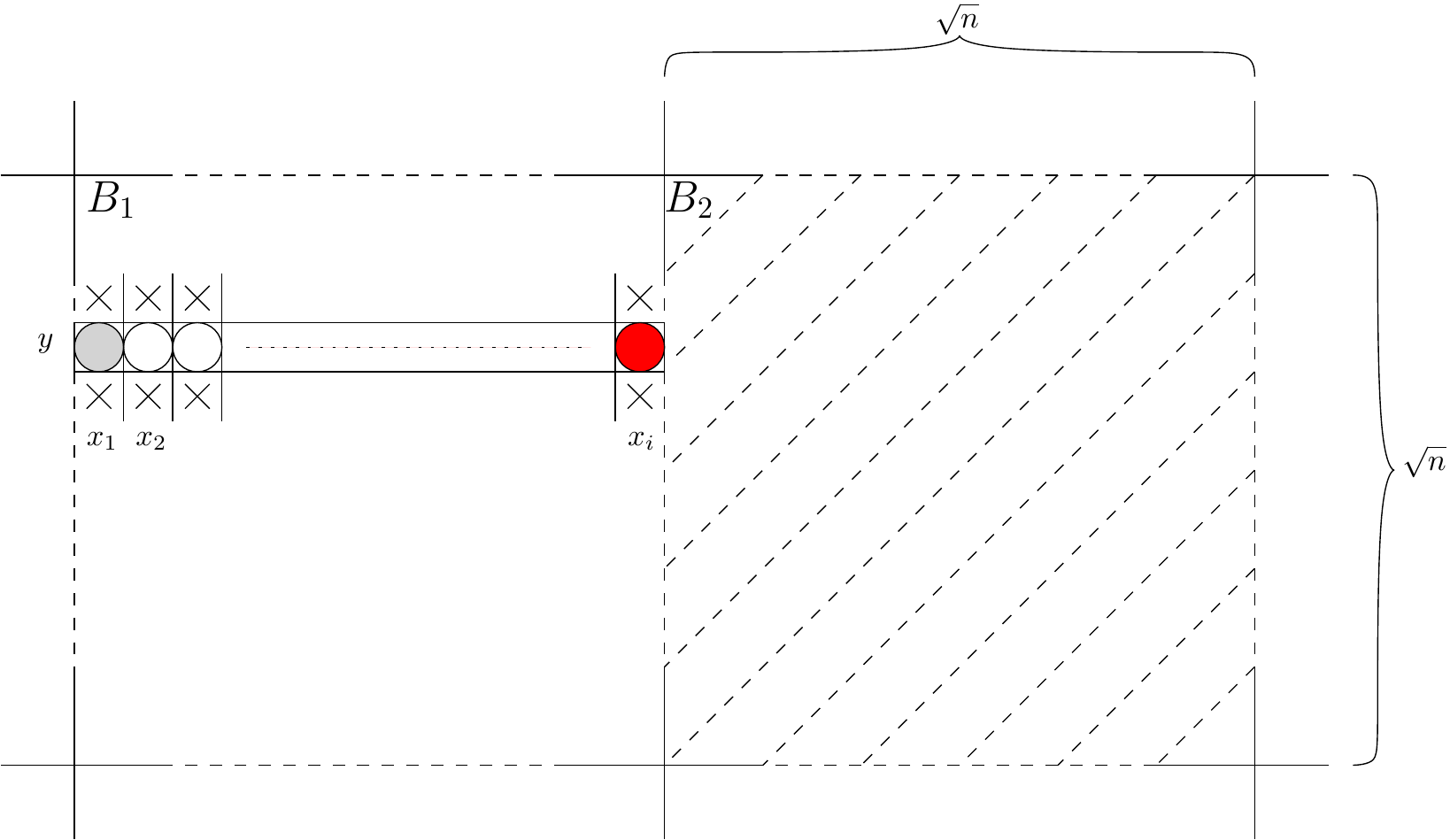}}  \quad 
			\subcaptionbox{$l$ pushes one move towards $B_l$.}
			{\includegraphics[scale=0.4]{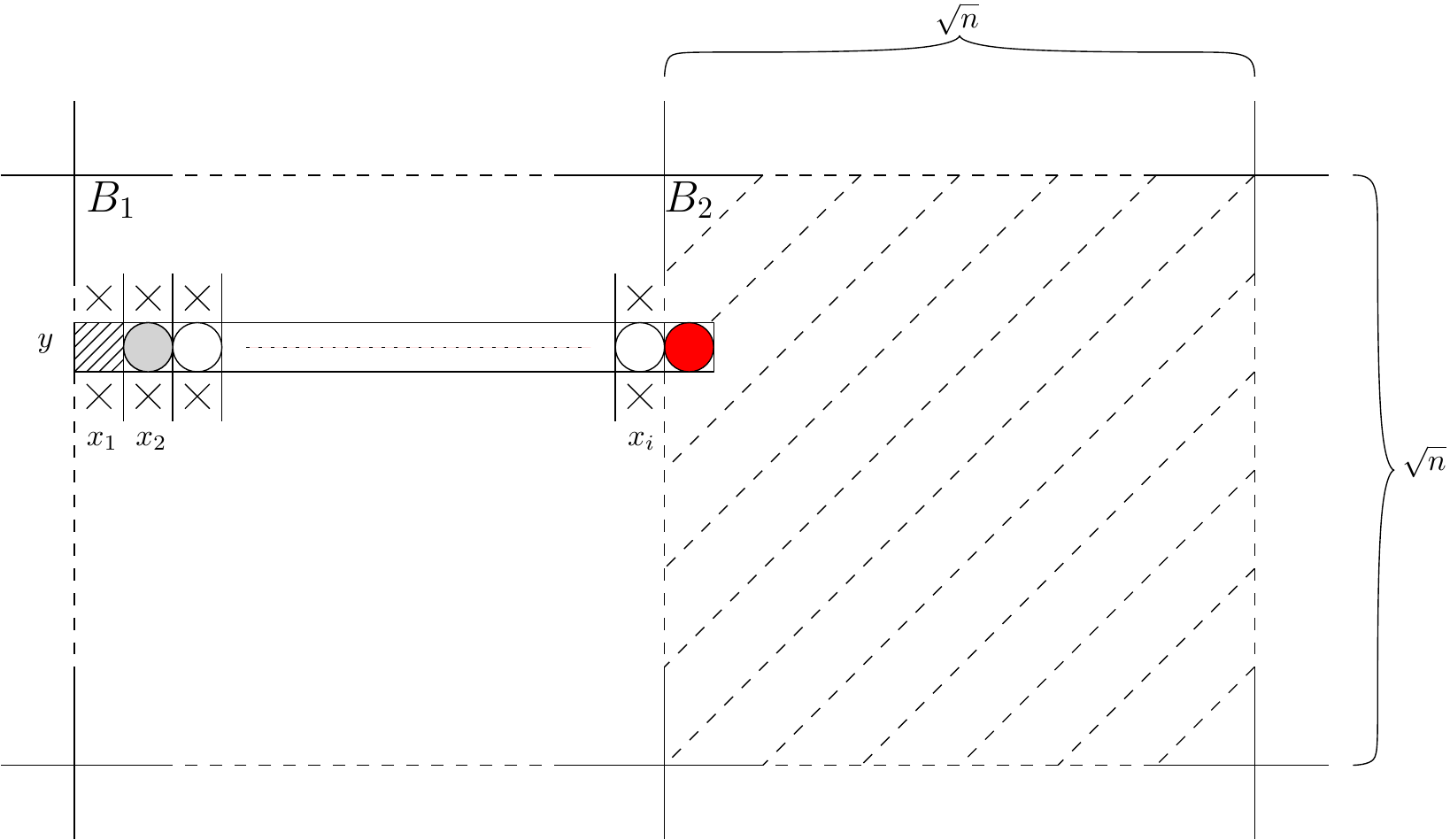}} 
			\caption{Case 1. A line $l$ of length $\sqrt{n}$ of a leaf component that occupies the whole dimension of a sub-box.}
			\label{fig:CompressingPreservesConnectivity}
		\end{figure} 
		\item \textbf{Case 2.} Similar of \textbf{Case 1} but with a line $l$ of length less than $\sqrt{n}$. $l$ pushes one move the right, and the length of $l$ dose not decrease in this case. Therefore, the whole connectivity of the shape is not effected. See Figure \ref{fig:CompressingPreservesCase1b}.
		\begin{figure}[th!]
			\centering
			\subcaptionbox{A line $l$ of length $i <\sqrt{n}$}
			{\includegraphics[scale=0.4]{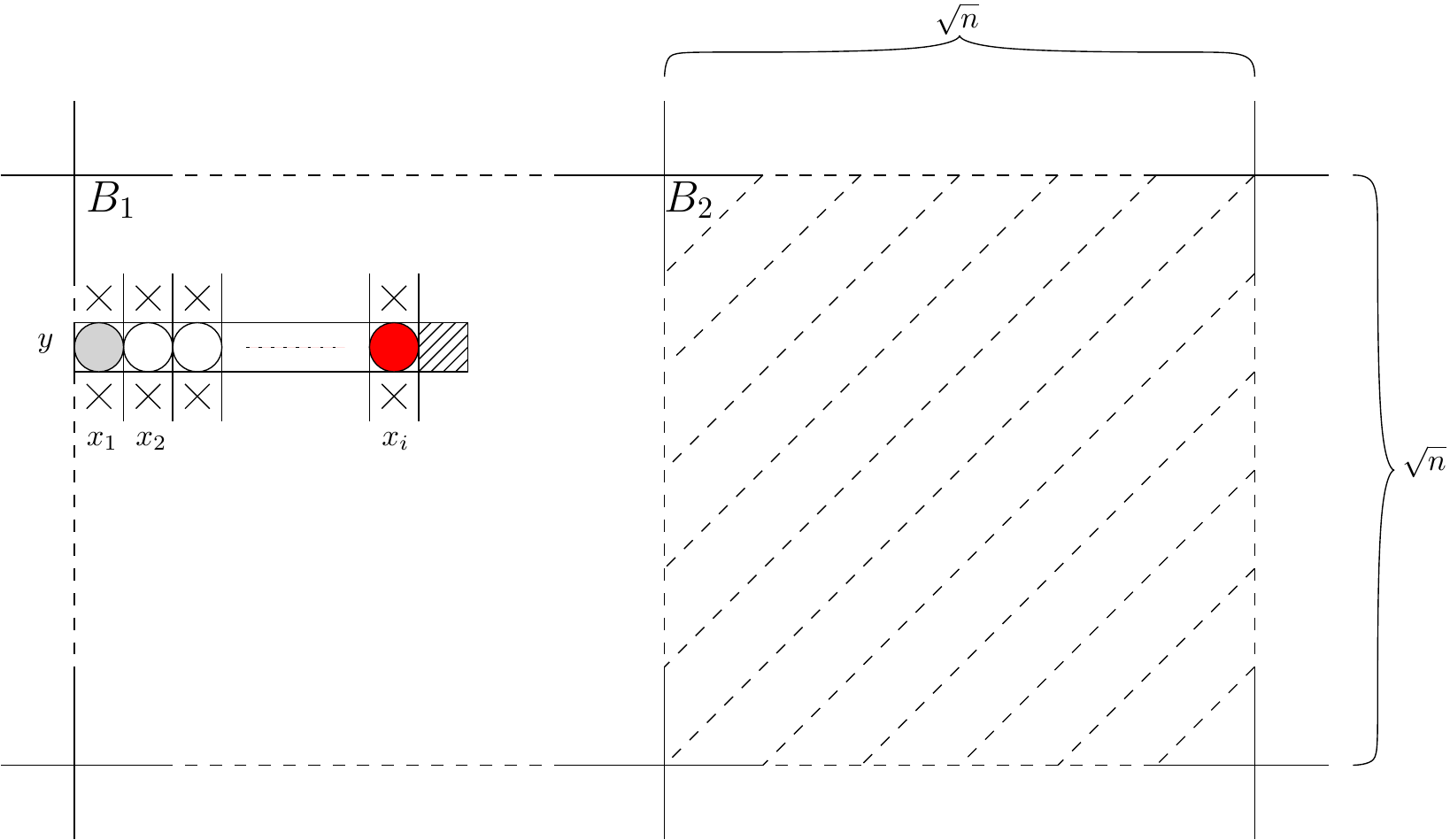}}  \quad 
			\subcaptionbox{$l$ pushes one move towards $B_1$.}
			{\includegraphics[scale=0.4]{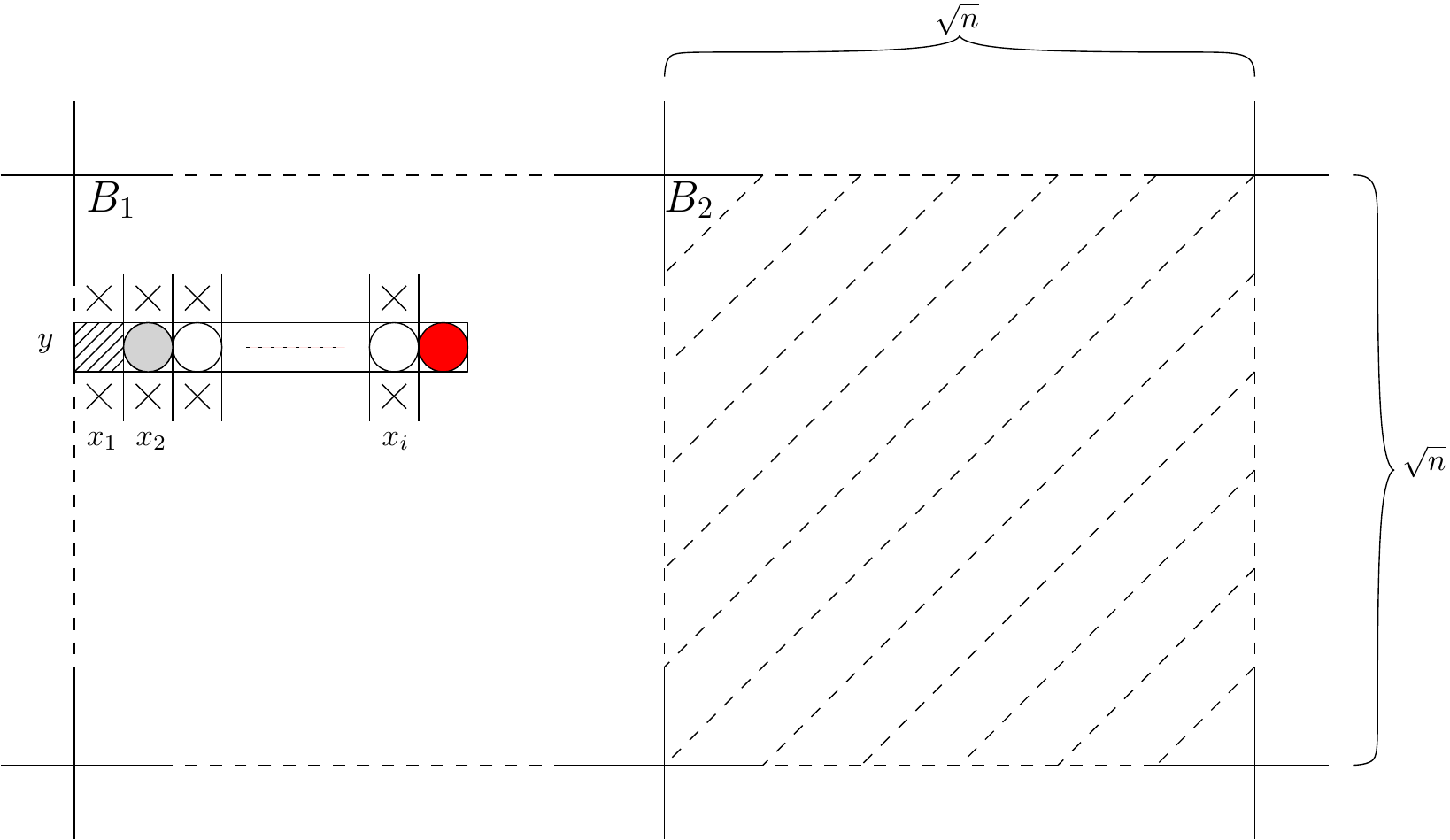}}
			\caption{Case 2. A line $l$ of length $i <\sqrt{n}$ of a child  component.}
			\label{fig:CompressingPreservesCase1b}
		\end{figure}		
		\item \textbf{Case 3.} Similar of \textbf{Case 2} in which there is two horizontal lines, $l_1$ and $l_2$, where $l_1$ starts from the leftmost column $x$ and ends at $x+i$ of $B_l$, and $l_2$ occupies $(x+i+2,y), \ldots ,(x + \sqrt{n},y)$. Now, $l_1$ pushes one move to fill the empty cell $(x+i+1, y)$,  a new empty cell has been created at $(x, y)$ and then both lines combines into a single line in of length $\sqrt{n}-1$, as in Figure \ref{fig:CompressingPreservesCase1c}. Still,  this move dose not violate \textit{connectivity} of the whole shape. 
		\begin{figure}[th!]
			\centering
			\subcaptionbox{Two horizontal lines occupy row $y$, both of lengths less than $\sqrt{n}$.}
			{\includegraphics[scale=0.4]{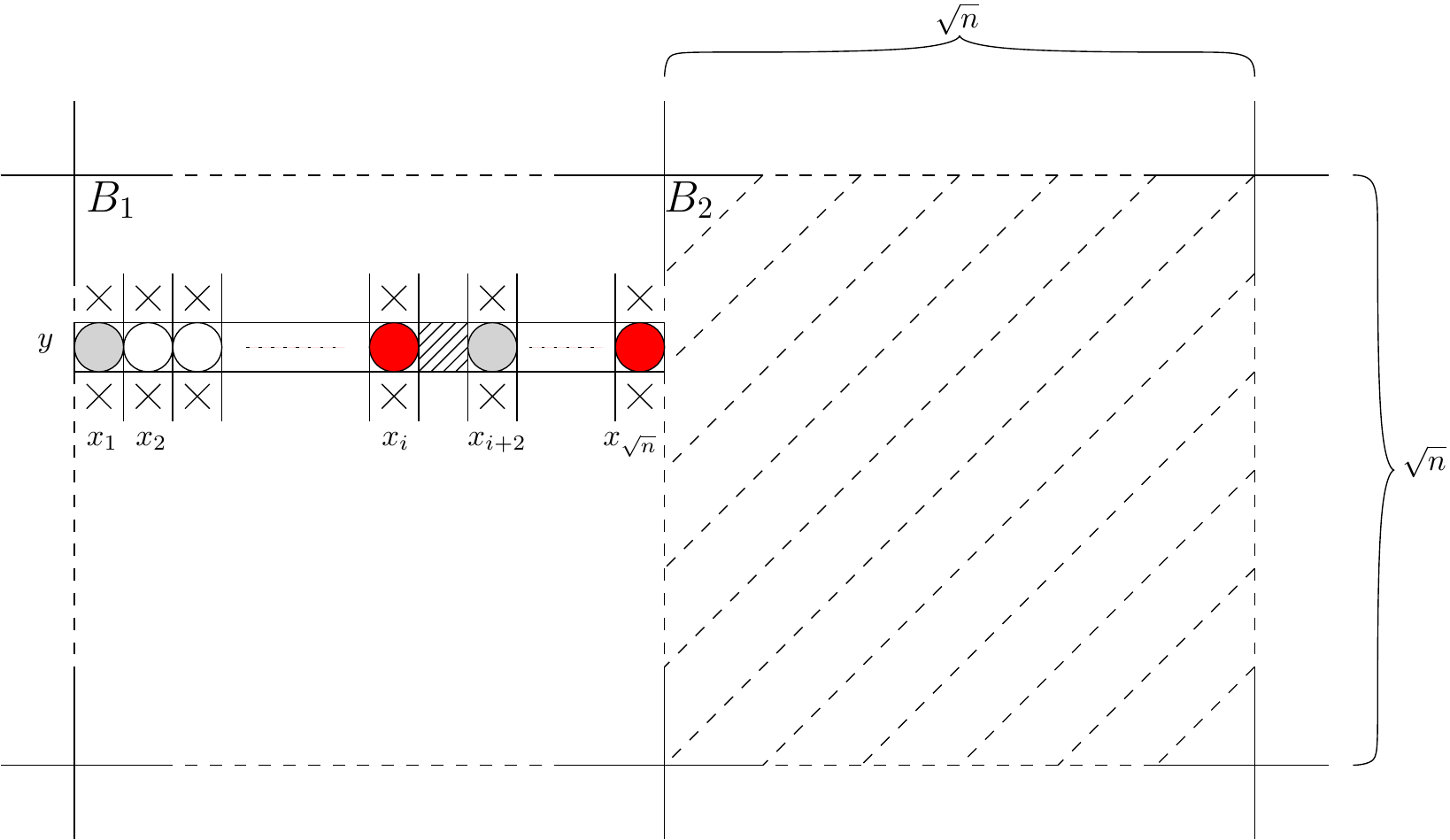}}  \quad 
			\subcaptionbox{$l$ pushes one move towards $B_l$.}
			{\includegraphics[scale=0.4]{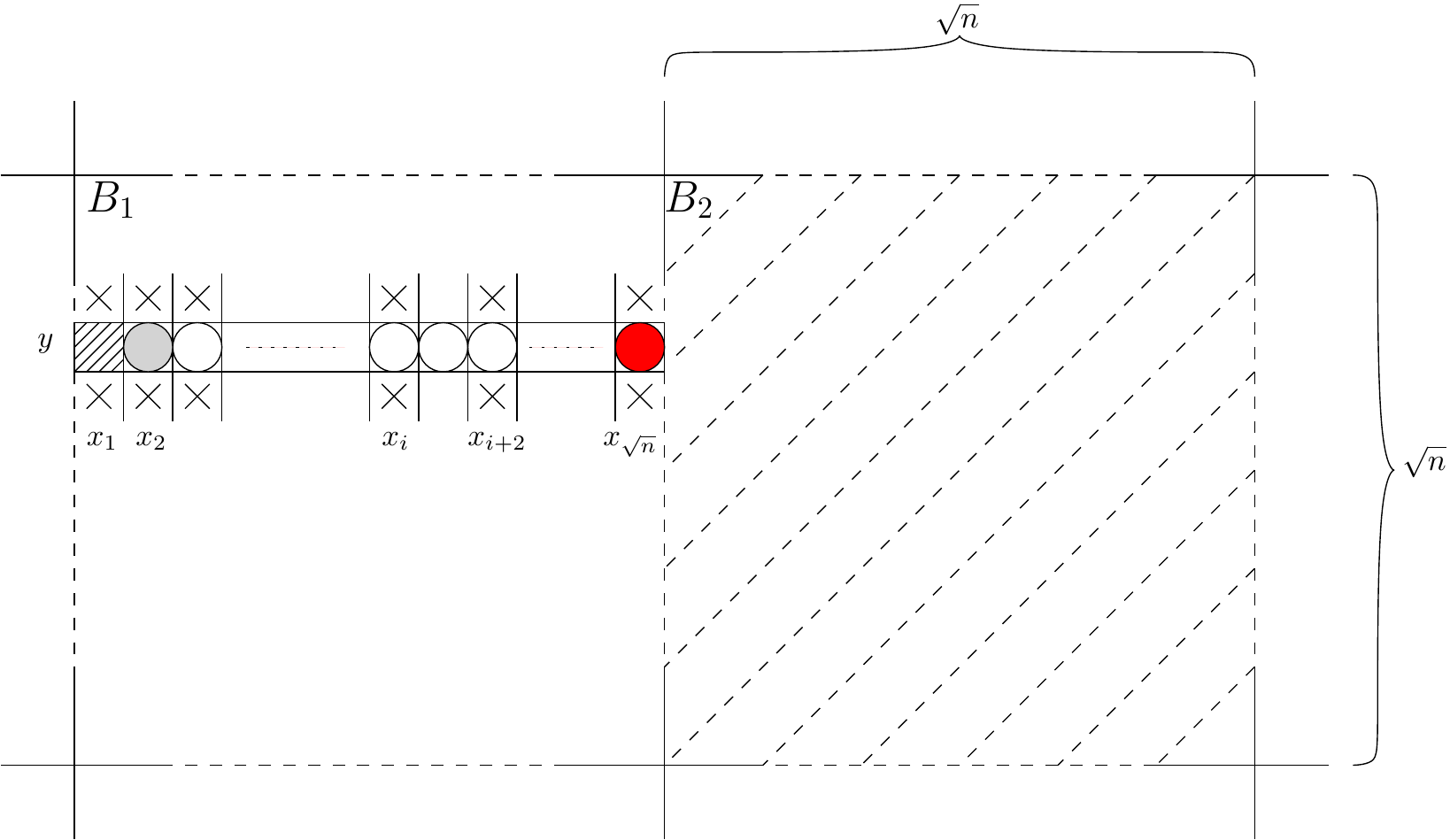}}
			\caption{Case 3. Two lines of a child component occupy a row, both of lengths less than $\sqrt{n}$.}
			\label{fig:CompressingPreservesCase1c}
		\end{figure} 
	\end{itemize}
	
	As mentioned earlier, when a component $C_l \in B_l$ moves to merge with its parent  $C_p \in B_p$, no line exceeds the four boundary of $B_p$. This shall preserves connectivity as well, and the following cases show how {\sc Compress} accomplishes this task, if $l$ occupies a row in $B_p$.   
	
	\begin{itemize}    	
		\item \textbf{Case 4.} The line  $l $ of length $i < \sqrt{n}$  starts from the leftmost column $x$ and ends at $x + i$, where there is an empty cell to the right at $(x+i+1, y)$. Once $l$ is pushed a single move to the right, $l$  fills in that empty cell and occupies positions $(x+1, y) , \ldots, (x+i+1, y)$. Therefore, the length of $l$ increases by 1, while the connectivity is preserved. See an example of this move in Figure \ref{fig:CompressingPreservesConnectivity_2}. 		
		\begin{figure}[th!]
			\centering
			\subcaptionbox{The line $l$ starts from a boundary between $(B_l,B_p)$ and ends at $(x+i, y)$, where $i < \sqrt{n}$.}
			{\includegraphics[scale=0.4]{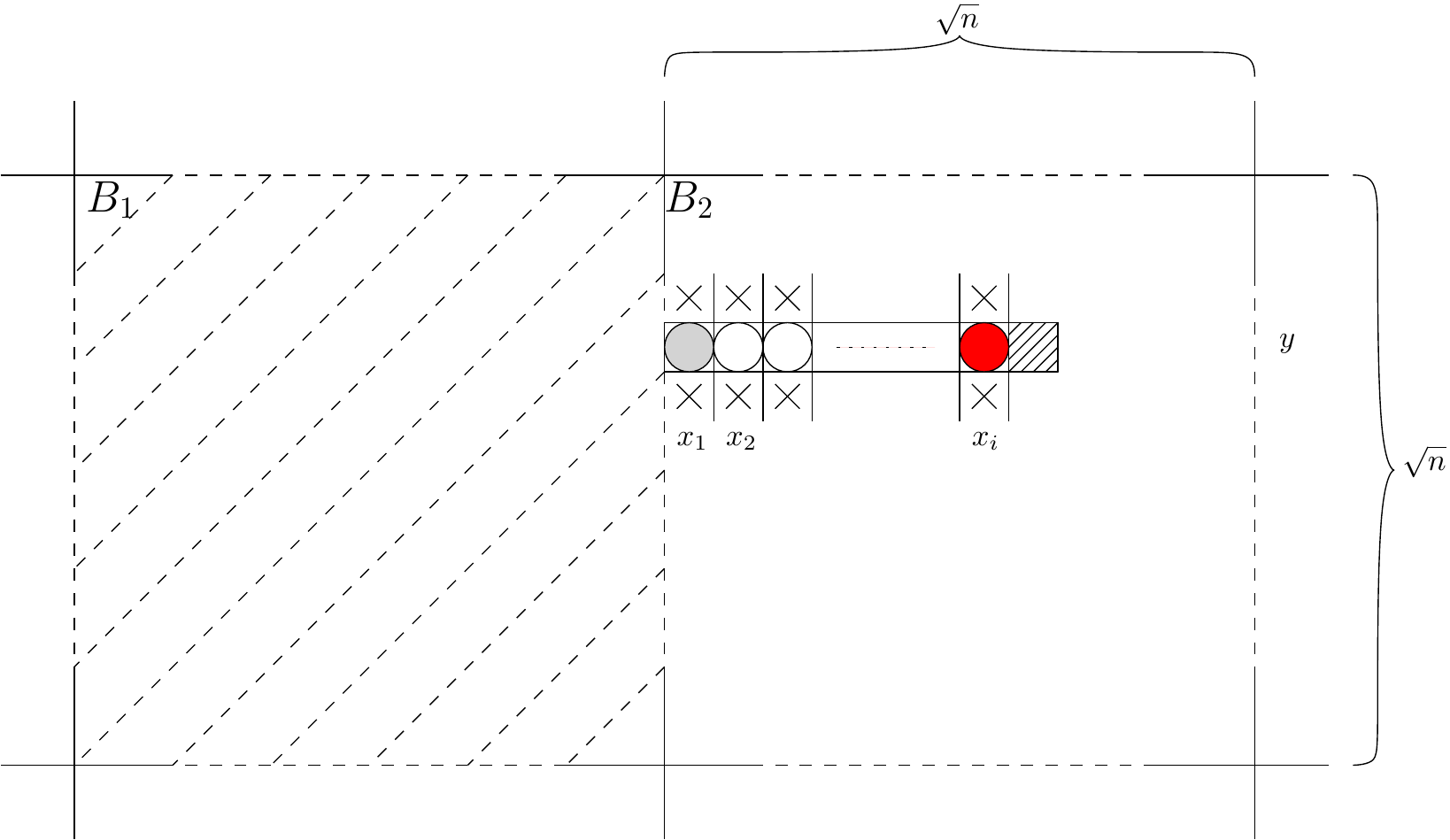}} \quad
			\subcaptionbox{$l$ is moved one positon right to occupy the empty cell to its right.}
			{\includegraphics[scale=0.4]{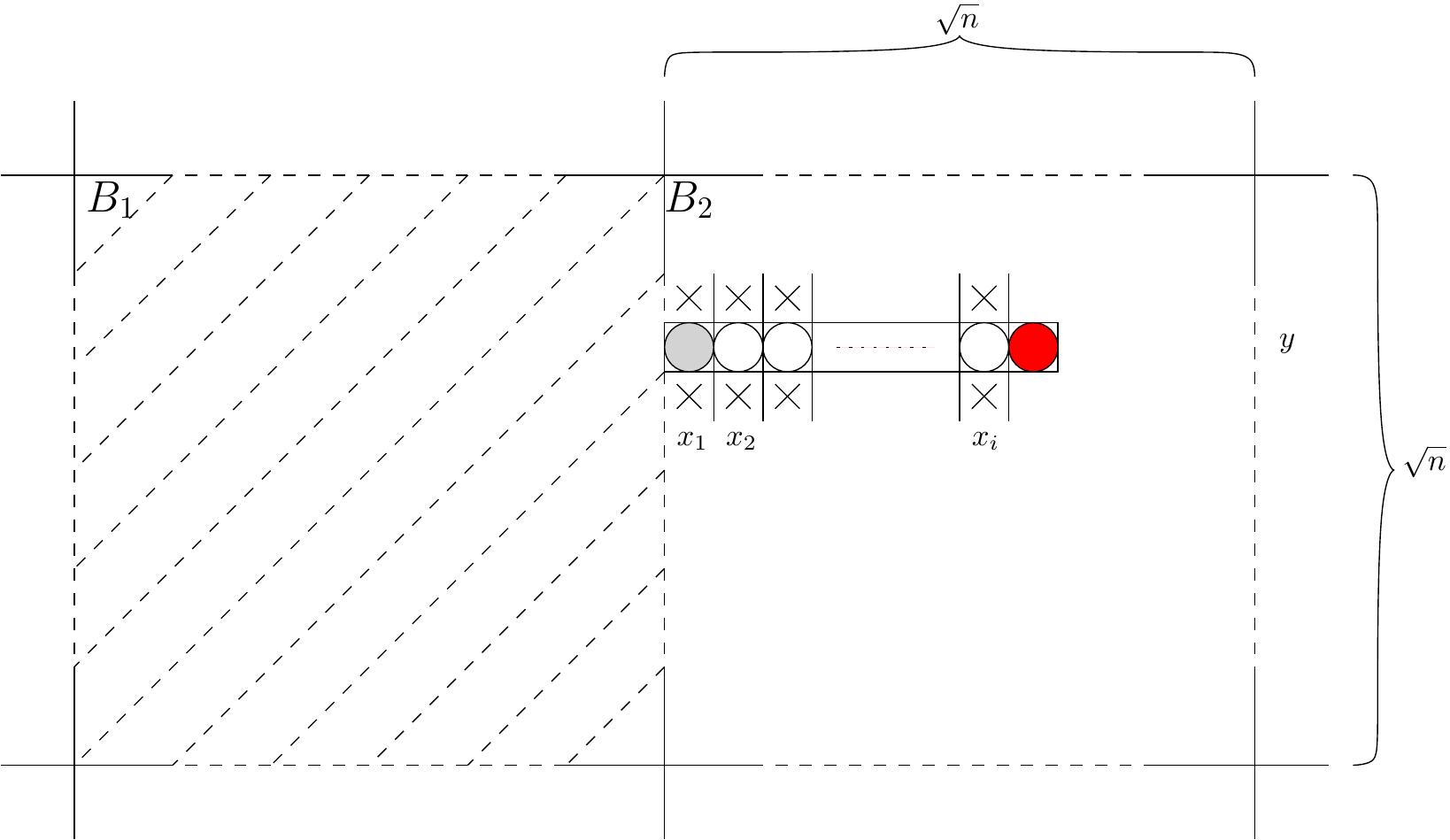}}
			\caption{Case 4. A line $l$ of length $i <\sqrt{n}$ in a parent component.}
			\label{fig:CompressingPreservesConnectivity_2}
		\end{figure}
		
		\item \textbf{Case 5.} The line $l$ of length $\sqrt{n}$ starts from the left and finishes at the right boundary of $B_p$. Once $l$ is pushed towards the right, it turns to fill empty cells at the right boundary of $B_p$, starting from the rightmost column to the left. The line $l$ needs two moves per node to change its orientation. Figures \ref{fig:CompressingPreservesConnectivity_3} and  \ref{fig:CompressingPreservesConnectivity_4} depicts two different examples of filling a boundary. Hence, this case preserves \textit{connectivity} of the whole shape. 
		\begin{figure}[th!]
			\centering
			\subcaptionbox{A line $l$ of length $i =\sqrt{n}$ occupyis a whole row in $B_p$.}
			{\includegraphics[scale=0.38]{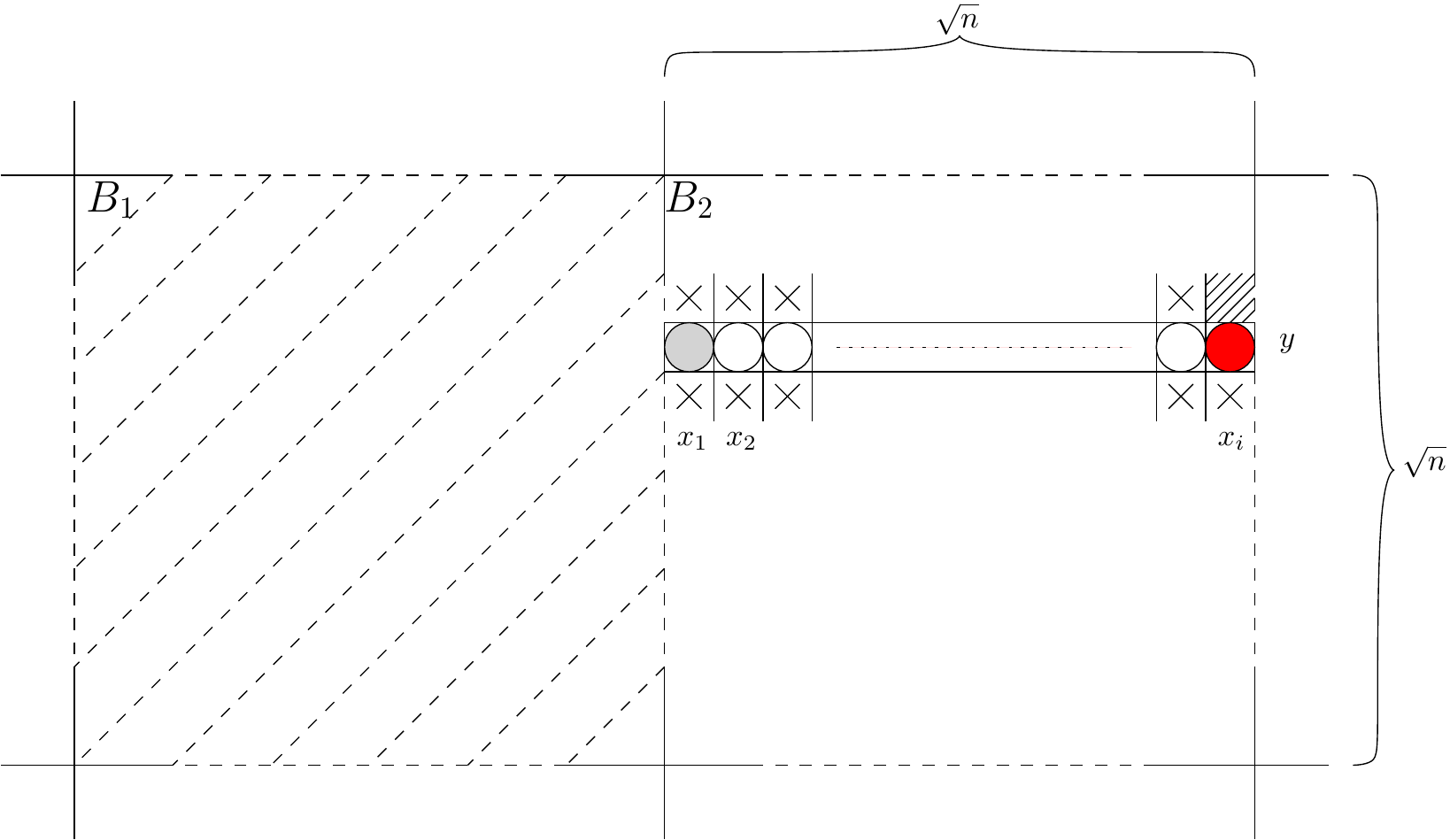}} \qquad 
			\subcaptionbox{$l$ starts to fill in an empty cell $(x+i, y+1)$ at the boundary of $B_p$ by pushing the node at $(x+i,y)$ one move up.}
			{\includegraphics[scale=0.38]{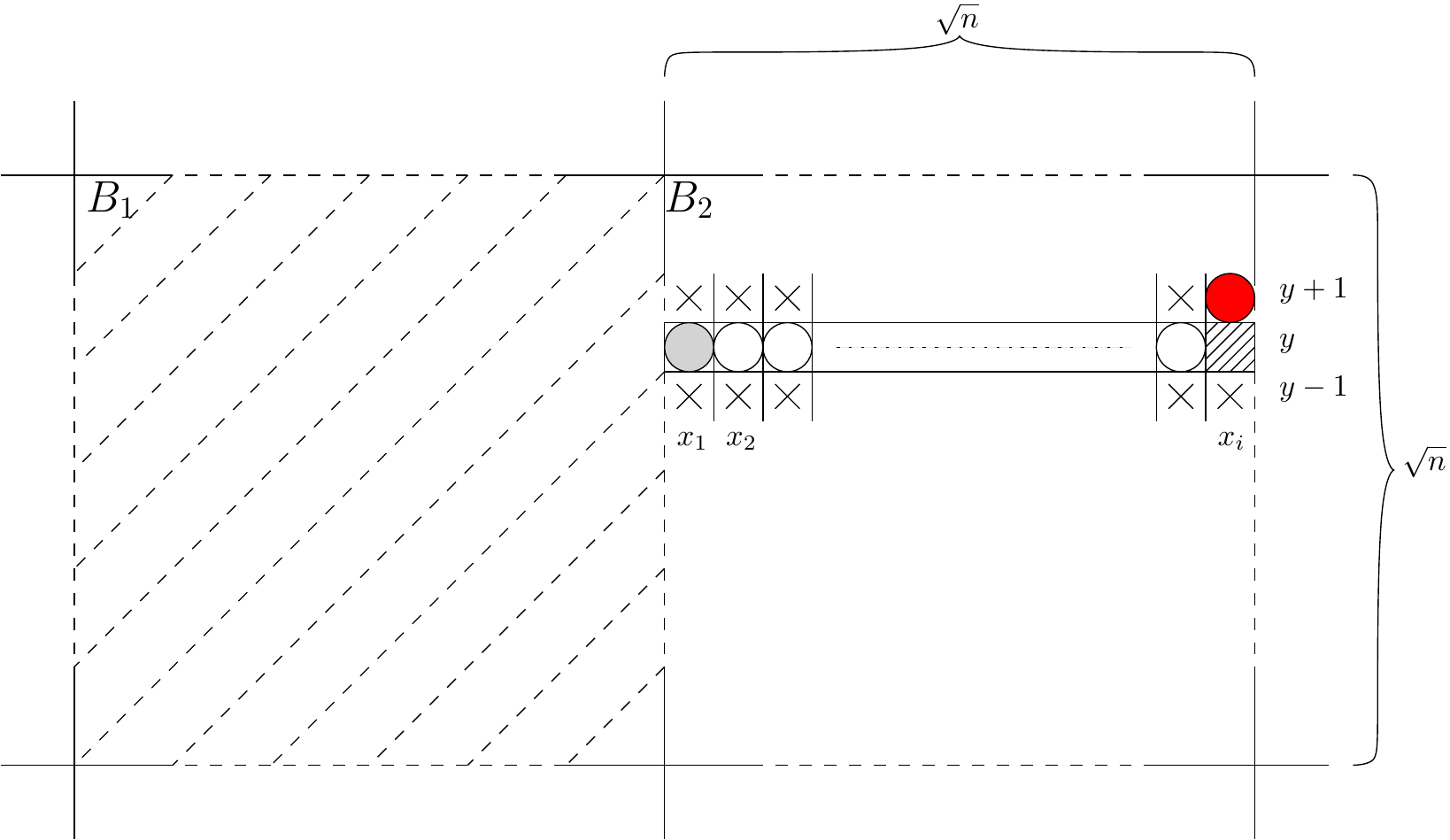}} \\
			\subcaptionbox{$l$ pushes one move right.}
			{\includegraphics[scale=0.4]{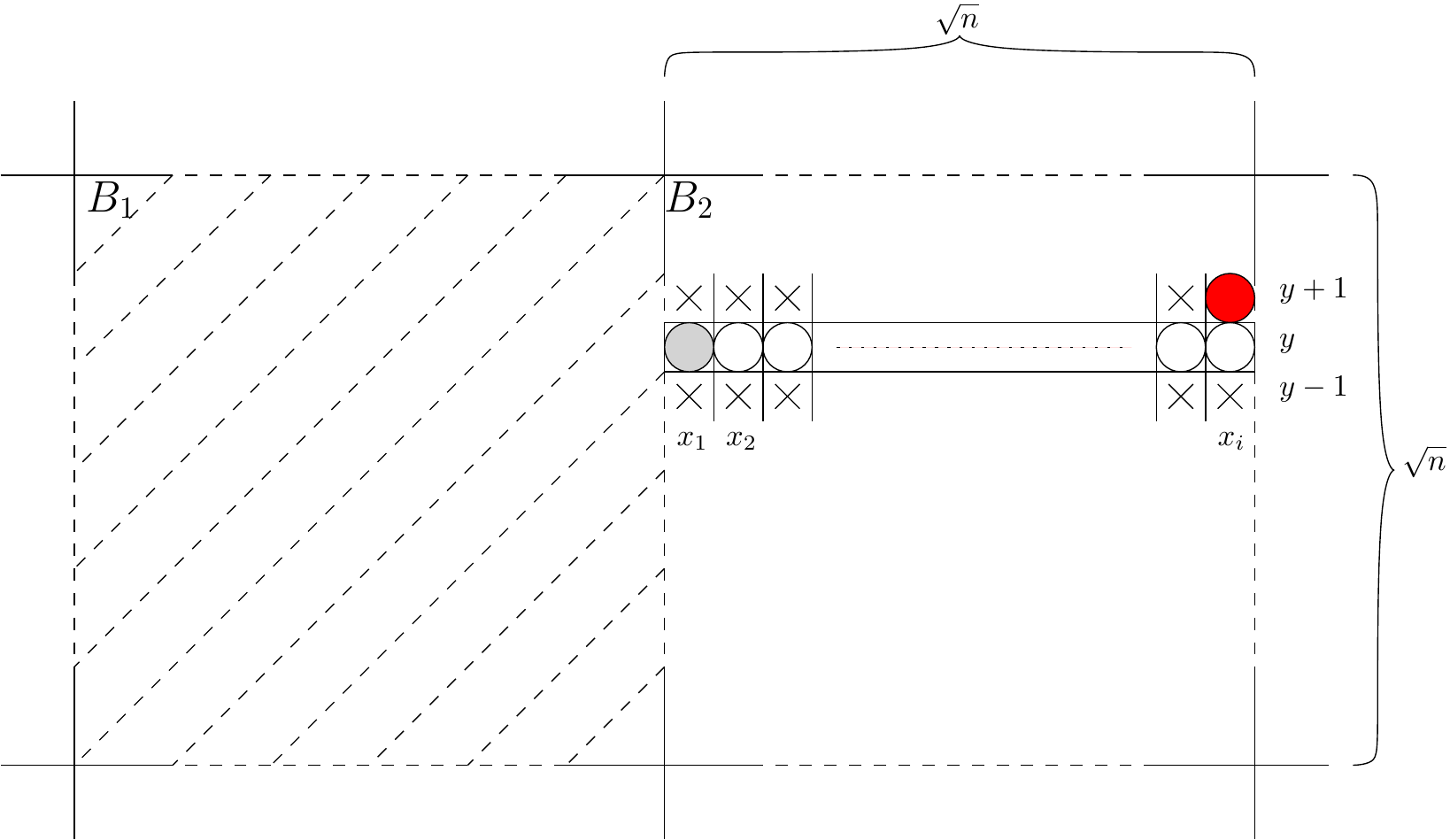}} 	   	
			\caption{Case 5 - Example 1. A line $l$ of length $\sqrt{n}$ of a parent component occupies the whole dimension of a sub-box., where there is empty cell at the rightmost column.}
			\label{fig:CompressingPreservesConnectivity_3}
		\end{figure}
		\begin{figure}[th!]
			\centering
			\subcaptionbox{A line $l$ of length $i =\sqrt{n}$ occupyis a whole row in $B_p$.}
			{\includegraphics[scale=0.38]{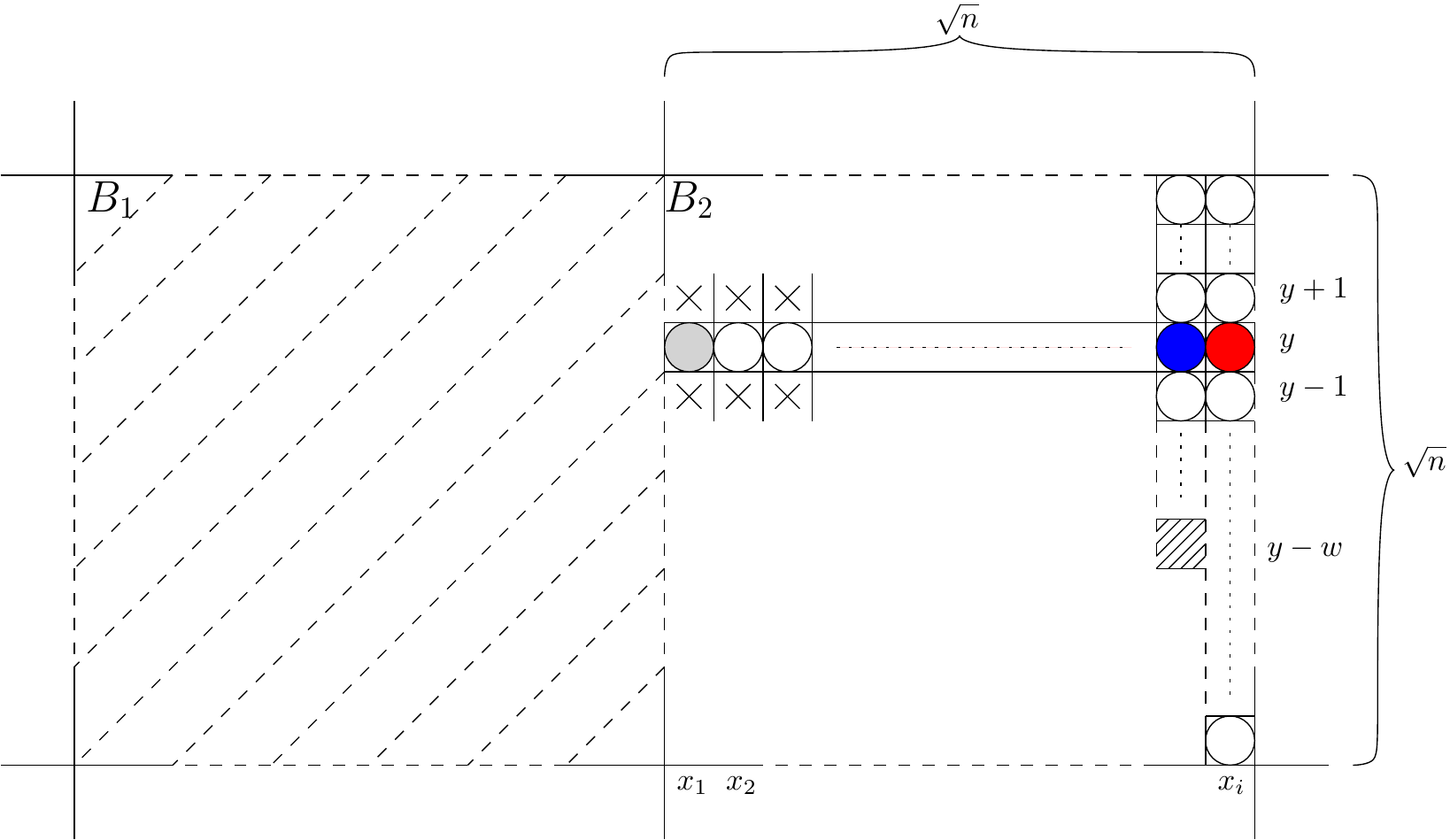}} \qquad 
			\subcaptionbox{$l$ turns to fill in the first empty cell at the right boundary of $B_p$, cell $(x+i-1, y+w)$. Then, $l$ pushes one move right to occupy the new empty cell $(x+i-1, y)$. }
			{\includegraphics[scale=0.38]{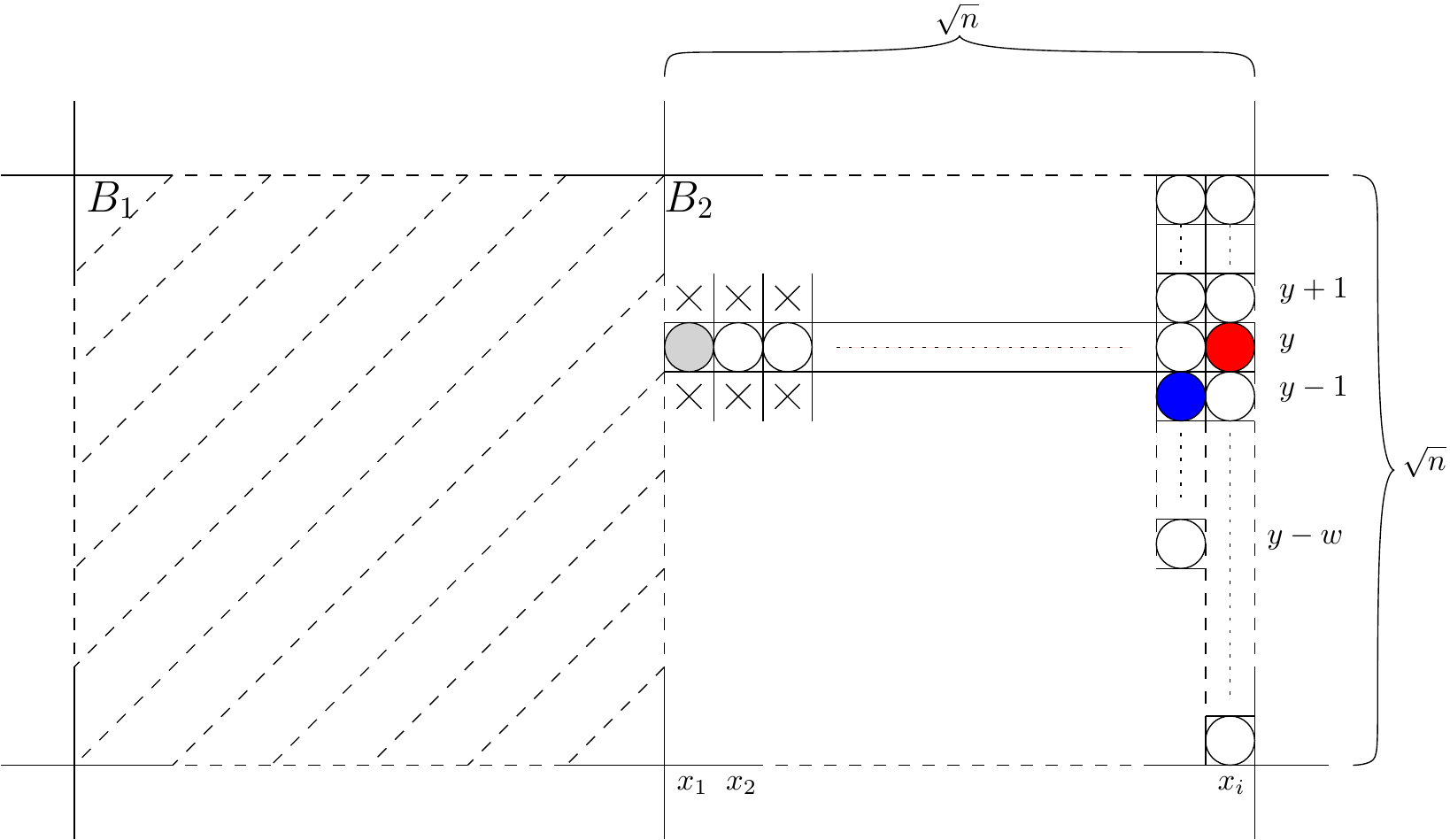}}
			\caption{Case 5 - Example 2. A line $l$ of length $\sqrt{n}$ of a parent component occupies the whole dimension of a sub-box., where there is no empty cell at the rightmost column. In this case,  $l$ fills in an empty cell at the column $x+i-1 $ of $B_p$.}
			\label{fig:CompressingPreservesConnectivity_4}
		\end{figure}
	\end{itemize}
	Finally, in all above cases, $l$ pushes one move towards the right without breaking \textit{connectivity} of $S_I$.  As an immediate observation: whenever a line $l \subset S_I$ inside a sub-box of dimension $\sqrt{n}$, for all $1 \le l \le \sqrt{n}$, that starts (\textit{perpendicularly}) from a boundary \textit{pushes} one move towards the opposite boundary between $(B_l, B_p)$, the global \textit{connectivity} of the whole shape is preserved. Further, this holds also for all $l$ lines that are pushing one move from $B_l$ towards $B_p$, sequentially one after another at any order, starting from the furthest-to-nearest line from that boundary between $B_l$ and $B_p$. Therefore, this must hold for a finite number of line moves a leaf $C_l$ requires to merge with its parent $C_p$ in $B_p$.
\end{proof} 

\newpage
\subsection{Running Time} 

Now, we are ready to analyse the time complexity of {\sc Compress}. The following lemmas provide a rough upper bound for all possible shape configurations. Given a uniform partitioning of any initial connected shape $S_1$ of order $n$, let us first show the total steps required to compress a leaf component $C_l$ in a sub-box $B_l$ into a parent $C_p$ occupying an adjacent sub-box $B_p$, in a worst-case.

\begin{lemma} \label{lem:RunTime_Of_One_Component}
	Given a pair of components $C_l, C_p$ of $k_l$ and $k_p$ nodes, $1 \le k_l+k_p \le n$, occupying adjacent sub-boxes $B_l, B_p$ of size $\sqrt{n}$ each, receptively.  Then, $C_l$ requires at most $O(n)$ steps to move from $B_l$ and compress into $C_p$ in $B_p$, without breaking connectivity. 
\end{lemma}
\begin{proof}
	Assume $B_l, B_p$ are connected diagonally (see Definition \ref{def:Sub-boxesConnectivity}), where a component $C_l$ occupies $\sqrt{n}$ lines in $B_l$ and $C_p$ consists of $\sqrt{n}$ lines in $B_p$ as well. $C_l$ pushes from $B_i$ via an intermediate sub-box $B_m$ towards $B_p$. Then, the $\sqrt{n}$ lines of $C_l$ moves a distance of at most $\sqrt{n}$ to cross the boundary between $B_l$ and $B_p$, in a total of at most $n$ moves to completely occupy $B_m$. Again, $C_l$ takes additional $n$  to move into $B_p$ and join $C_p$. Moreover, assume that $C_l$ requires additional $2n$ steps to fill in a boundary at $B_p$. Therefore and by Lemma \ref{lem:CompressingPreservesConnectivity},   $C_l$ compresses into $C_p$ in a total of at most:
	\begin{align*}
	t &= n + n + 2n = 4n \\
	&= O(n),
	\end{align*}  
	moves, while preserving connectivity of the shape during transformations.  
\end{proof}

The compression cost of this transformation could be very low taking only one move or being very high in some cases up to linear steps. To simplify the analysis, we divide the total cost of \emph{UC-Box} into charging phases. We then manage to upper bound the cost of each charging phase independently of the sequential order of compressions.

\begin{lemma} \label{lem:RunningTime}
	{\sc Compress} compresses any connected shape $S_I$ of order $n$ into a $\sqrt{n}\times\sqrt{n}$ square shape, in $O(n\sqrt{n})$ steps without breaking connectivity. 
\end{lemma}
\begin{proof}
   Let us compute a spanning tree $T=(V,E)$ of the associated graph $G(S_I)$, where nodes $V$ correspond to elements linked by edges $E$ representing \textit{relation connectivity} between them. Recall that the partitioning process of $S_I$ into small $ \sqrt{n} \times \sqrt{n} $ sub-boxes shall provide at most $O(\sqrt{n})$ occupied sub-boxes (proved in \cite{AMP19}). Observe that each component inside theses occupied sub-boxes matches a subtree in $T$. In each charging phase, the strategy compresses a single or multiple components of at most $O(\sqrt{n})$ nodes distance $O(\sqrt{n})$, which incurring a total cost of at most $O(n)$ (the worst-case is analysed in Lemma \ref{lem:RunTime_Of_One_Component}). Once this computed, a single or multiple subtrees of $\sqrt{n}$ nodes are removed form $T$. By repeating the same argument for at most $O(\sqrt{n})$ charging phases, then we arrive at the case where all nodes are removed from $T$,  which means that all components have been compressed into a single sub-box in a total cost at most $O(n\sqrt{n})$ moves, while the whole connectivity of the shape is not broken (consult Lemma \ref{lem:CompressingPreservesConnectivity}).   
\end{proof}

Similar to Lemma \ref{lem:RunningTime} but of different perspective, assume that $S_I$ is hidden of which we cannot see the actual configuration. Colour black all the $O(\sqrt{n})$ occupied sub-boxes by $S_I$. Each black sub-box consists of $n$ cells in a total of $n\sqrt{n}$ cells for all black occupied sub-boxes. Given that, in each charging phase the strategy moves $\sqrt{n}$ lines $\sqrt{n}$ distance of a total cost at most $O(n)$ moves to compress all nodes inside a black sub-box. This might happen in any order throughout the transformation. As the cost $O(n)$  is mostly sufficient to compress all nodes inside a single black sub-box and by Lemma \ref{lem:CompressingPreservesConnectivity},  a total of at most $O(\sqrt{n})$ charging  phases are fairly enough to compress all components inside the $O(\sqrt{n})$ occupied black sub-boxes, in a maximum total cost $O(n\sqrt{n})$ moves, while preserving connectivity during the transformations.    

There are a number of connected shapes which can be divided, by some partitionings, into $n$ connected components. This kind of dividing brings a wort-case complexity in which {\sc Compress} meets its maximum cost, due to several reasons. First, it splits the shape into  the maximum possible number of components $n$. Moreover, the diameter of the shape is spread over the largest space to cover $n$ (rows or columns). Unlike other dense connected shapes of shorter diameters, outspread shapes are harder to compress due to the lack of long lines and the additional cost required for individuals and short lines. The following lemma shows that there are a finite number of specific shapes that has  $n$ components produced by some artificial partitionings.

\begin{lemma} \label{lem:Finit_Worst_shapes}
	There are a finite number of initial shapes denoted $ \mathscr{S_I} $ that can be divided into $n$ components by some uniform partitionings. It holds that {\sc Compress} compresses any instance  $A \in \mathscr{S_I}$ into a single square sub-box in a total of $O(n\sqrt{n})$ steps, while preserving connectivity during its course.  
\end{lemma}
\begin{proof}
	Given $A \in \mathscr{S_I}$ of $n$ nodes with a particular partitioning positioned to divide $A$ into $n$ connected components. See partitioning examples of a zigzag line in Figure \ref{fig:UnusualShapeOfNnodes} and diagonal zigzag in Figure \ref{fig:worstCaseSahpe}. By Lemma \ref{lem:ConCompInSubBox}, a sub-box can have at most $2\sqrt{n}$ components, and with a given partitioning, $A$ can occupy at most $2\sqrt{n}/n = 2\sqrt{n} = O(\sqrt{n})$  sub-boxes. As $A$ is connected, each occupied sub-box contains at most $\sqrt{n}/2$ components of size 1 each. 
	
	\begin{figure}[th!]
		\centering
		{\includegraphics[scale=0.5]{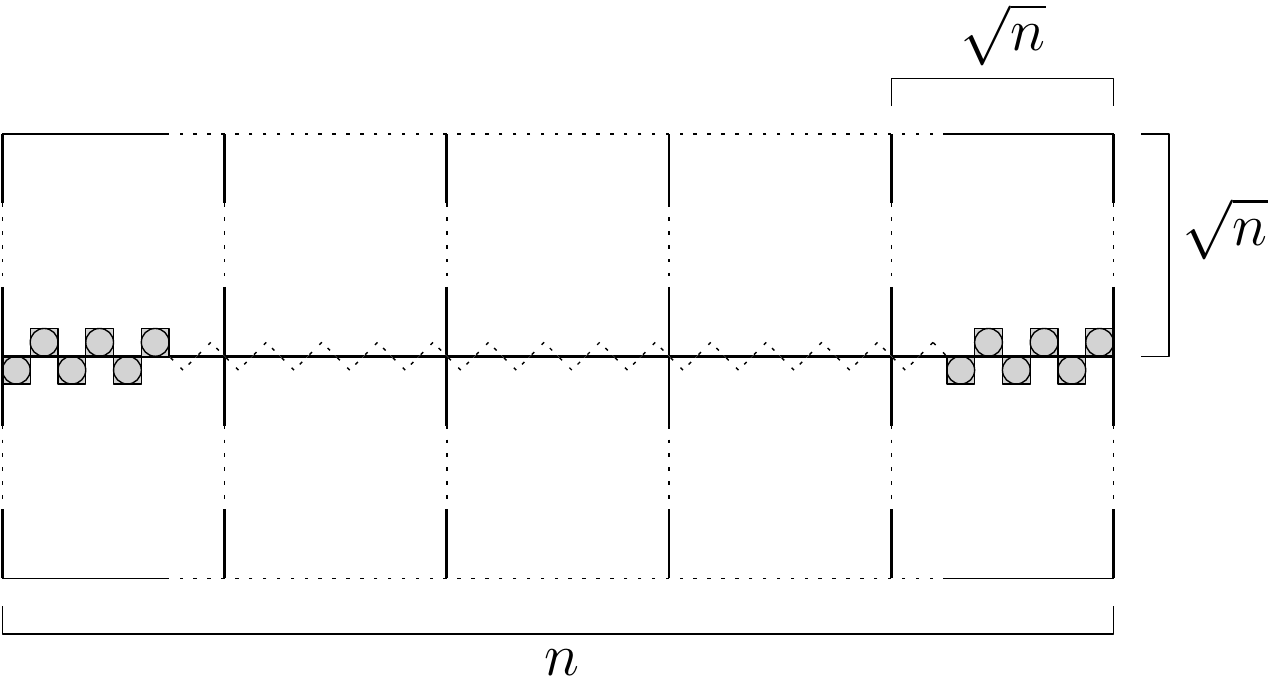}}	
		\caption{A zigzag line with a partitioning positioned to cross the middle through every two nodes of $A \in \mathscr{S_I}$.}
		\label{fig:UnusualShapeOfNnodes}
	\end{figure}
	
	\begin{figure}[th!]
		\centering
		{\includegraphics[scale=0.5]{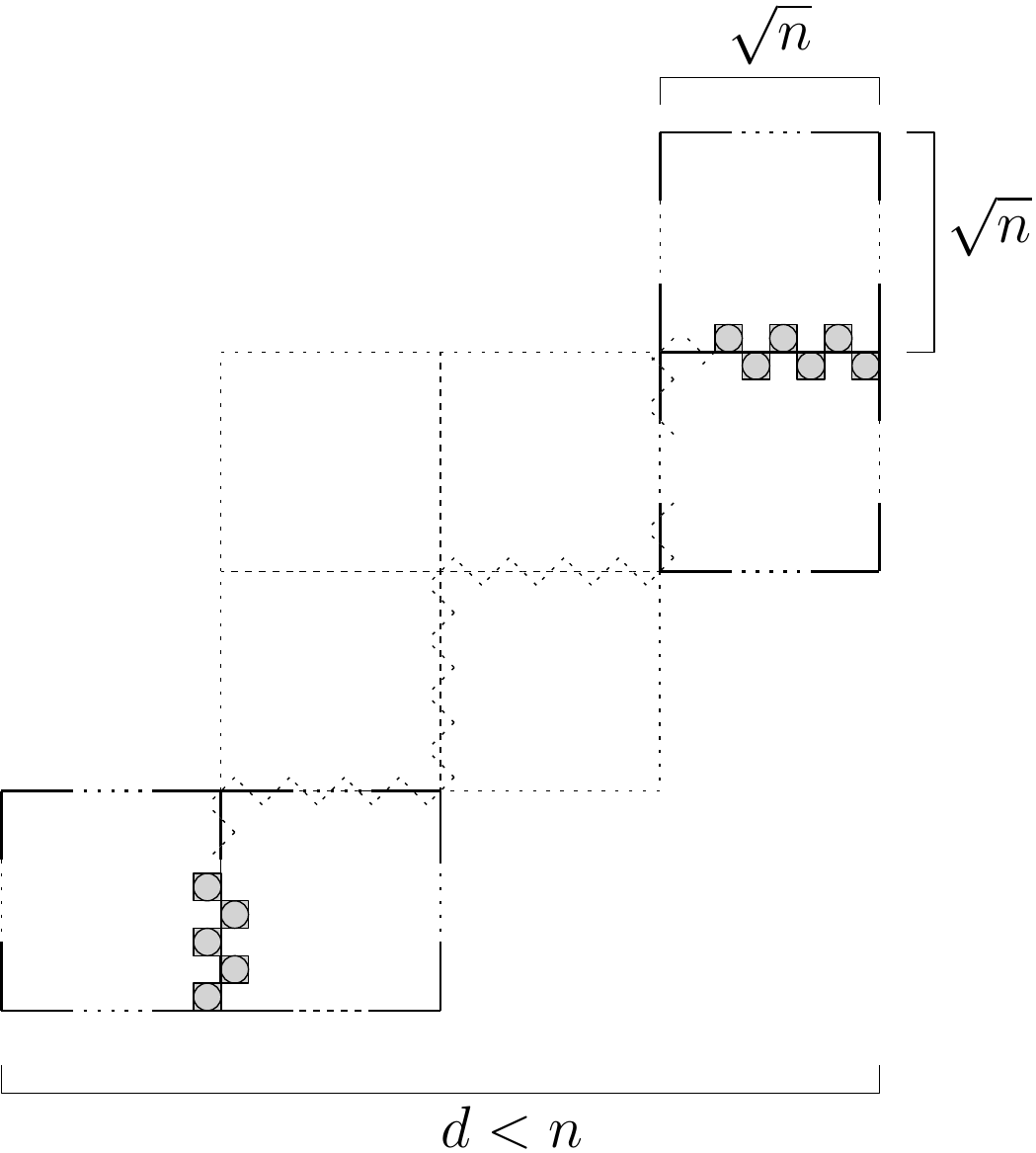}}	
		\caption{A diagonal zigzag line with a partitioning positioned to cross the middle through every two nodes in $A \in \mathscr{S_I}$ of dimension $d < n$.}
		\label{fig:worstCaseSahpe}
	\end{figure}
	
	We investigate how the current strategy behaves in the worst scenario. In any given charging phase $i$, for all $1 \le i \le \sqrt{n}$, {\sc Compress} compresses $\sqrt{n}$ lines of a single of multiple components  to their parent by moving them $\sqrt{n}$ distance in a total of $O(n)$ steps, with preserving connectivity. More, the compression may be via two diagonal sub-boxes occurring at most $2 \cdot i\sqrt{n}$. Additional cost is also given for rearrangements of at most  $2\cdot\sqrt{n}/2 = \sqrt{n}$ moves. Therefore, the charging phase $i$ takes a total moves $t_1$ of at most:
	\begin{align*}
	t = & \sum_{i=1}^{\sqrt{n}} i + (2 \cdot i\sqrt{n}) + \sqrt{n} =  \frac{\sqrt{n}(\sqrt{n}+1)}{2}+ (2 \cdot i\sqrt{n}) +\sqrt{n} = \frac{n +\sqrt{n}}{2} + (2 \cdot i\sqrt{n}) +\sqrt{n} \\
	& =  \frac{n +\sqrt{n} + (4 \cdot i\sqrt{n}) +2\sqrt{n} }{2}  =  \dfrac{n + 3\sqrt{n} + (4 \cdot i\sqrt{n})}{2} =   \dfrac{5n + 3\sqrt{n}}{2}\\
	& = O (n).    
	\end{align*}
	
	For the upper bound, we will assign the cost $t$ for each of the $2\sqrt{n}$ occupied sub-boxes in those particular shapes of Figures \ref{fig:UnusualShapeOfNnodes} and \ref{fig:worstCaseSahpe}. Hence, the total running time $T$ in moves is as follows:
	
	\begin{align*}
	T&  = t \cdot 2\sqrt{n} \\
	& = \dfrac{5n + 3\sqrt{n}}{2} \cdot 2\sqrt{n} =  \dfrac{10n\sqrt{n} + 6n}{2} = 5n\sqrt{n} +3n\\
	& = O(n\sqrt{n}). 
	\end{align*}
	
	By Lemma \ref{lem:CompressingPreservesConnectivity}, {\sc Compress} compresses any shape $A \in \mathscr{S_I}$ of $n$ nodes with a particular partitioning that dividing $A $ into $n$ components in at most $  O(n\sqrt{n}) $ steps with preserving the whole connectivity of the shape during the transformations.
\end{proof}

By Lemma \ref{lem:CompNice}, the resulting compressed square shape of  {\sc Compress} is a nice shape. Hence, Lemma \ref{lem:CompressingPreservesConnectivity}, Lemma \ref{lem:RunningTime}, and reversibility of nice shapes (from \cite{AMP19}), we therefore have: 

\begin{theorem} \label{theo:UC_Box}
	For any pair of connected shapes $(S_I,S_F)$ of the same order $n$, \emph{UC-Box} transforms $S_I$ into $S_F$ (and $S_F$ into $S_I$)  in $O(n\sqrt{n} )$ steps, while preserving connectivity during its course. 
\end{theorem}

\section{Lower Bounds}
\label{sec:Lower_Bounds}

In this section, we discuss a necessary minimum cost to transform the diagonal of order $n$ into a line exploiting the parallelism of line moves in a two dimensional grid.  The \textit{Input} is a diagonal shape $S_D$ of $n$ nodes occupying  $ (x_{1},y_{1}),$ $ (x_{2},y_{2}), \ldots, (x_{n},y_{n}) $, and the \textit{output} is a straight line $S_L$ of $n$ nodes occupying $n$ consecutive cells at a column $y_i$ or row $x_i$, for all $1 \le i \le n$.  Observe that $S_D$ matches the maximum number of steps a transformation takes to transform it into $S_L$, due to the inherent distance between these two pairs of shapes.

Given a complete graph $G=(V,E)$ in which $V$ is a set of nodes in $S_D$ and $E$ is non-negative edge weights (Manhattan distance between  nodes). Then, a simplification of this problem is collect all nodes on $S_D$ at the bottom-most node. That is, every node in $S_D$ must perform one or more hops through other nodes and end up at the bottom-most node. When going through a node the two or more nodes can continue traveling together and exploit parallelism.  

Any such solution to the problem forms a spanning tree $T \subseteq G$, where every leaf to root path corresponds to the hops of a specific node until it reached the end.
The cost of each subtree is: $c(T)$ is the total sum of the distances of its edges plus the cost of nodes $c(V)$. Every edge $E(u,v)$ has a cost equal to the distance of moving $u$ to $ v $,  where each node has a cost of paying for each internal node of a subtree the number of nodes in its subtree. The latter cost is due to not being able to exploit parallelism whenever turning, and any hop requires another turn. The cost due to distances is just:
\begin{align}
c(E) = \sum_{e\in E(T)}^{} cost(e),
\label{eqEdgesCost}
\end{align}
The cost of internal nodes is equal to:
\begin{align}
c(V) = \sum_{i =1}^{d(T)} i \cdot v\in d(T)_i,
\label{eqVetrexCost}
\end{align}
Where $ d(T) $ is the depth of tree $ T $ and $d(T)_i$ is the number of nodes at level $i$. The total cost in number of moves given by such a tree $ T $ is the sum of \ref{eqEdgesCost} and \ref{eqVetrexCost}: 
\begin{align}
c(T) = c(E) + c(V).
\label{eqTotalCost}
\end{align}
Now, the two sums seem to give some trade-off. If the depth is very small, then the cost due to distances seems to increase (e.g., if all nodes travel into one hop, they all pay their distances and the cost is quadratic). This approach is similar to any sequential transformation of individual movements which pays a cost of $\Theta(n^2)$ to transform $S_D$ into $S_L$. The summation of the total individual distances is, $\Sigma \Delta =0+1+2+\ldots+(n-1)=\Theta(n^2)$, independently of whether connectivity is preserved or not during transformations. This is because of the inherent individual distance between $S_D$ and $S_L$.
On the other hand, the tree $T$ of very large depths looks as a spanning line where a lot of parallelism must be exploited. The distance in this case would cost only $ c(E)= n-1 $. While the sum of turns at each node becomes quadratic, $ c(V)= n^2 $. Therefore, we observe that more balanced trees of logarithmic depth, such as binary trees, manage to balance both sums and give total cost $ n \log n $.
Due to the trade-off, it does not seem easy to lower bound in the general case. Further, it does not seem easy to lower bound the edges-sum even by some parameters depending on the depth (so that both sums will be using similar parameters). It might not even be related to that parameter. Therefore, we tried to further simplify the problem by restricting the solutions to extremely limited depths. Below, we have successfully managed to establish some special-case lower bounds for this problem.

It can be easily seen that no uniform strategy can achieve better bound than the $ O(n\log n) $-time strategy of \cite{AMP19}, by simply increasing the number of lines that are merging in every phase to decrease the number of phases. Hence, we have the following proposition.

\begin{proposition}
	Any strategy represented by a balanced tree performs $\Omega(n \log n)$ moves. 
\end{proposition}
\begin{proof}
	Observe that such a strategy is essentially trying to increase the degrees of the nodes of a balanced tree and decrease its depth. For example, take any merging parameter $ k \ge 2 $. Notice that the $ O(n\log n) $-time transformation of \cite{AMP19} (called \emph{DL-Doubling}) has k = 2, as it is merging pairs of lines and get $ \log n $ phases. So, in every phase $ i $ we are going to partition the $ L $ lines into $ L/k $ groups of $ k $ consecutive lines each and merge the lines within each group into a single line.
	
	First , in phase 1, $ L = n  $, and we are partitioning into $ n/k $ groups. For each group we are paying at least $ k^2 $ asymptotically to merge the lines in it. Therefore for phase 1 we pay $ (n/k)k^2 = nk $ (this is similar also to the $ O(n\sqrt{n} )$-time transformation in \cite{AMP19}, but there it only did it once and gathered all the lines to the bottom and not in any further phases). Then in phase 2 $ L = n/k $, we are partitioning into $ L/k = n/k^2 $ groups. Each group is paying at least $ k^3 $ asymptotically, because the distance between consecutive lines has now increased to $ k $ (roughly). Thus gives again cost at least $ nk $. This should hold for the other phases.
	
	Now, Observe that this strategy gives $ \log_{k} n = \log n \log k $ phases. If each is paying $ nk $, then the total cost is $ (nk)(\log n \log k) = n \log n(k \log k) $, which for all $  k \ge 2 $ is at least $ 2 n \log n = \Omega(n \log n) $. This would be helpful because it excludes any attempts to get a better than the $ O(n \log n )$-time transformation by simply playing with the degrees of the tree in a uniform way (which in turn decreases its depth and thus the number of phases).
\end{proof}

\subsection{An $\Omega(n \log n)$ Lower Bound for The 2-HOP Tree}

We start to study a special case lower bound for all solutions that represented by a tree $T$ of a minimum depth. Assume any such solution moves all nodes in only one-way via shortest paths towards the target node. As a node joins other nodes, they do not split after that during the transformation. Let $d(T)$ denotes the depth of the tree. For $ d(T) = 1 $, the tree becomes a star, and the total cost is quadratic in this case, due to the summation of individual distances $c(E) =0+1+2+\ldots+(n-1)=\Theta(n^2)$.  

Then, we investigate the tree $T$ of depth at most 2, $ d(T) = 2 $. Observe that for any node in the tree we are paying ``\textit{asymptotically}'' at least the square of number of children that it has. The reason is that at most 2 of its children can be nodes at distance 1, then at most 2 can be nodes of distance 2, at most 2 of distance $ i $ in general due to the neighbouring properties of the diagonal. Thus, it gives a total cost for any such tree which is similar to $ c(T) = \sum_i d(u_i)^2 $, where $ d(u_i)$ is the number of children of $u_i$. That is, the squares of the degrees of all internal nodes, excluding their parent (the root $u_i$). When taking into account all nodes, this gives a graph-theoretic measure related to chemical compounds, known as the \textit{Zagreb index} \cite{Gutman72,GRBT75}. A great amount of bounds have been established for it, but none of which could be directly used in our case. 

Let $T$ of $k$ nodes be a tree of depth 1, as shown in Figure \ref{fig:K_Tree}. Then, the total asymptotic cost of the tree $c(T)$ is at least:
	\begin{align}
	c(T) \ge \sum_{i =0}^{k} d(u_i)^2,
	\label{eq:TotalCostOfTree}
	\end{align}
	Where $ d(u_i)$ is the degree of node $u_i$. 
\begin{figure}[th!]
	\centering
	{\includegraphics[scale=0.7]{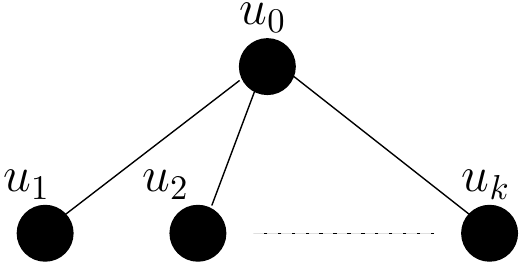}} 
	\caption{A tree $T$ of $k$ nodes.}
	\label{fig:K_Tree}
\end{figure}

Given the tree $T$ of $k$ we show the first case of a minimum total cost $T$ must pay if it has a node $u_i \in T$ with degree at least $n \log n$;

\begin{lemma} \label{lem:1stCase}
	If $\exists$ $ d(u_i) \ge \sqrt{n \log n}$, then $c(T) \ge n \log n$, for all $ 0 \le i \le k$. 
\end{lemma}
\begin{proof}
	The proof is straightforward. Consider the tree $T$ of $k$ nodes in Figure \ref{fig:K_Tree}. If $k \ge \sqrt{n \log n}$, then the tree shall have a minimum total cost of $c(T) \ge d(u_0)^2 = k^2 = {(\sqrt{n \log n})}^2 =  n \log n$. In general, if there exists a node $u_i \in T$, for all $ 0 \le i \le k$, such that $d(u_i) \ge \sqrt{n \log n}$, then the total cost of the tree must be at least $c(T) \ge n \log n$.  
\end{proof} 

Now, let us assume that all nodes in the tree have degrees less than $n \log n$. Thus, we show the lower bound of Lemma \ref{lem:1stCase} holds in this case.

\begin{lemma} \label{lem:2ndCase}
	Let  $ d(u_i) < \sqrt{n \log n}$ $\forall i$, where $ 0 \le i \le d(u_0) = k$. Then, $c(T) > n \log n$. 
\end{lemma}
\begin{proof}
	Given a tree $T$ of $n$ nodes that has depth of 2, and a subtree $T^{\prime} \subseteq T$ of $k$ nodes as in Figure \ref{fig:K_Tree}. So, let  $ d(u_i) < \sqrt{n \log n}$ , for all $ 0 \le i \le d(u_0) = k$. Assume without loss of generality that the nodes $u_i \in T^{\prime}$, for all $1 \le i \le d(u_0) =k$, are ordered in non-increasing degrees from left to right (increasing order $i$), that is, $d(u_1) \ge d(u_2) \ge \ldots \ge d(u_k)$. Hence, there are $n- (k+1) \in T$ nodes remaining to be assigned. As $d(u_1) \in T^{\prime}$ is the maximum, it must hold that, $d(u_1) \ge \frac{n-(k+1)}{k}$, thus $ \frac{n-(k+1)}{k} \le d(u_1) < \sqrt{n \log n} $. 
	
	Next, there are $n- (k+1) - d(u_1) \in T$ nodes need to be allocated. As $d(u_2) \in T^{\prime}$ is the maximum among the rest, it must hold that $d(u_2) \ge \frac{n-(k+1)- d(u_1)}{k-1}$, thus $ \frac{n-(k+1)- d(u_1)}{k-1} \le d(u_2) < \sqrt{n \log n} $. In general, if a node $d(u_i) \in T^{\prime}$ is the maximum, then the following must hold that,
	
	\begin{align}
	d(u_i) \ge \frac{n- \Big( \sum_{j=0}^{i-1} d(u_i) \Big) -1}{k-(i-1)},
	\label{eq:GeneralCase}
	\end{align}
	Thus,
	\begin{align}
	\frac{n- \Big( \sum_{j=0}^{i-1} d(u_i) \Big) -1}{k-(i-1)} \le d(u_i) < \sqrt{n \log n},
	\end{align}	
	Now, plug $i = 1$ and $k = d(u_0)$ in (\ref{eq:GeneralCase}) yields, 
	\begin{align}
	d(u_1) &\ge \frac{n- d(u_0) -1}{d(u_0)} 
	> \frac{n- \sqrt{n \log n} -1}{\sqrt{n \log n}},
	\end{align}	
	When $i = 2$, we will get,
	\begin{align}
	d(u_2) &\ge \frac{n- (d(u_0) + d(u_1)) -1}{d(u_0)} 
	> \frac{n- 2\sqrt{n \log n} -1}{\sqrt{n \log n} -1 },
	\end{align}
	For all $ 1 \le i \le d(u_0) = k$, we shall obtain,
	\begin{align}
	d(u_i) &\ge \frac{n-\Big( \sum_{j=0}^{i-1} d(u_i) \Big) -1}{d(u_0)-(i-1)} 
	> \frac{n- i\sqrt{n \log n} -1}{\sqrt{n \log n} -(i-1) }, \label{eq:CostOfu_i}
	\end{align}
	Then, we plug (\ref{eq:CostOfu_i}) into (\ref{eq:TotalCostOfTree}) of Observation \ref{oberv:TreeCost}, which implies,
	\begin{align}
	c(T) &> \sum_{i =0}^{d(u_0)} {\Bigg[ \frac{n- i\sqrt{n \log n} -1}{\sqrt{n \log n} -(i-1) } \Bigg]}^2 > {(\sqrt{n \log n})}^{-1} \sum_{i =0}^{d(u_0)} {(n- i \sqrt{n \log n}-1)}^2 \\
	&\simeq {(\sqrt{n \log n})}^{-1} \sum_{i =0}^{d(u_0)} {(n- i \sqrt{n \log n})}^2\\ &= {(\sqrt{n \log n})}^{-1} \sum_{i =0}^{d(u_0)} (n^2 + i^2 \cdot n \log n - 2 i\cdot n \sqrt{n \log n})\\
	&= {(\sqrt{n \log n})}^{-1} \Bigg[ d(u_0)\cdot n^2 + \sum_{i =0}^{d(u_0)} n( i^2 \log n - 2 i \sqrt{n \log n}) \Bigg] \\
	&= {(\sqrt{n \log n})}^{-1} \Bigg[ d(u_0)\cdot n^2 + n\sum_{i =0}^{d(u_0)} ( i^2 \log n - 2 i \sqrt{n \log n}) \Bigg].
	\label{eq:BoundCost}
	\end{align}
	We need to bound the summation of (\ref{eq:BoundCost}):
	\begin{align}
	\sum_{i =0}^{d(u_0)} n( i^2 \log n - 2 i \sqrt{n \log n}) &= \sum_{i =0}^{d(u_0)} i^2 \log n -  \sum_{i =0}^{d(u_0)}  2 i \sqrt{n \log n} \\ &= \Bigg( \log n  \sum_{i =0}^{d(u_0)} i^2  \Bigg) - \Bigg( 2 \sqrt{n\log n}  \sum_{i =0}^{d(u_0)} i \Bigg)\\
	&= \log n \Bigg( \dfrac{{d(u_0)}^3}{3} + \dfrac{{d(u_0)}^2}{2} + \dfrac{d(u_0)}{6} \Bigg)\\ &- 2 \sqrt{n\log n} \cdot \dfrac{d(u_0)(d(u_0+1))}{2}\\
	&\simeq \log n \Big( {d(u_0)}^3 + {d(u_0)}^2 + d(u_0) \Big) \\ &- 2  \sqrt{n\log n} \cdot {d(u_0)}^2.
	\label{eq:BoundCost_1} 
	\end{align}
	
	Now, plug (\ref{eq:BoundCost_1}) into (\ref{eq:BoundCost}), then it will give a total cost of the tree $c(T)$ that asymptotically bounded on:
	\begin{align}
	c(T) &> {(\sqrt{n \log n})}^{-1} \Big( n^2 \cdot d(u_0)  + n \log n \cdot {d(u_0)}^3 - \sqrt{n \log n} \cdot {d(u_0)}^2\Big)\\
	&= \dfrac{ n^2 \cdot d(u_0)  + n \log n \cdot {d(u_0)}^3}{\sqrt{n \log n}} - {d(u_0)}^2 \\
	&> \dfrac{ n^2 \cdot d(u_0)  + n \log n \cdot {d(u_0)}^3}{\sqrt{n \log n}} - n \log n 
	\end{align}
	
	Finally, since $d(u_0) > 1$, it implies that,
	\begin{align}
	c(T) &> \dfrac{n^2}{\sqrt{n \log n}} - n \log n = \dfrac{n^2\cdot n \log n}{\sqrt{n \log n} \cdot \log n} - n \log n = \dfrac{n^2}{\sqrt{n \log n} } - n \log n\\
	&= \dfrac{n \cdot n}{\sqrt{n} \sqrt{ \log n} } - n \log n = \dfrac{n \cdot \sqrt{n}}{ \sqrt{ \log n} } - n \log n = \dfrac{n \log n \cdot \sqrt{n}}{\log n \cdot \sqrt{ \log n} } - n \log n \\
	&> \dfrac{n \log n \cdot \sqrt{n}}{{\log}^2 n } - n \log n = \Bigg( \dfrac{\sqrt{n}}{{\log}^2 n} - 1 \Bigg) n \log n = \\
	& = \Omega (n \log n).
	\end{align}
	
\end{proof} 

As a result, both Lemmas \ref{lem:1stCase} and \ref{lem:2ndCase} show that the total cost of any spanning tree $c(T)$ of $n$ nodes and depth at most $d(T) \le 2$ is always bounded by $\Omega(n \log n)$.

\begin{theorem}
	Any 2-HOP spanning tree $c(T)$ of $n$ nodes and depth at most $d(T) \le 2$ has a total cost $c(T)$ of $\Omega(n \log n)$.  
\end{theorem}

\subsection{A conditional $\Omega(n \log n)$ Lower Bound - One way transformation}

Now we present another special case lower bound for transformations that are exploiting line moves. Again, our techniques is based on one-way assumption in which all nodes move in one direction via shortest paths towards the target node (e.g., from top to bottommost node in the diagonal). Whenever a node joins other nodes, they continue travelling together and do not split thereafter.

Let $S_{D}$ be a diagonal connected shape occupies of order $ n $ nodes (lines of length 1)  on positions $ (x_{1},y_{1}),$ $ (x_{2},y_{2}), \ldots, (x_{n},y_{n}) $. The argument starts by deciding a potential target position of the final straight line, $S_L$. Assume a potential placement of $S_L$ horizontally on the bottommost row $y_1$ or vertically at the leftmost column $x_1$ of the shape. With this, and without loss of generality, we therefore assume that lines only move down and leftwards. This is convenient as they always push a minimum distance towards the target potential placement, i.e, in our assumption at row $y_1$ or column $x_1$ of $S_L$. 

Enclose each individual node of $S_D$ into a square box of dimension $d = 1$ to have a total of $n$ boxes, see the black squares boxes in Figure \ref{fig:Lower_Bound}. Then, double the dimension of the square boxes to surround every two nodes in a total of $n/2$ boxes of $d=2$, such as the red squares boxes in Figure \ref{fig:Lower_Bound}. Repeat doubling dimensions each of different colours $\log n$ times, until arriving at 1 square box of $d = n$, which contains all nodes of $S_D$. Assume that $n$ is a power of 2, therefore, the total number of all square boxes shall be exactly $n + n/2 + n/4 + \ldots +1 = 2n - 1$ boxes, where there are $n$ boxes of $d =1$, $n/2$ boxes of $d= 2$, $\ldots$ and 1 box of $d= n$.

\begin{figure}[th!]
	\centering
	{\includegraphics[scale=0.5]{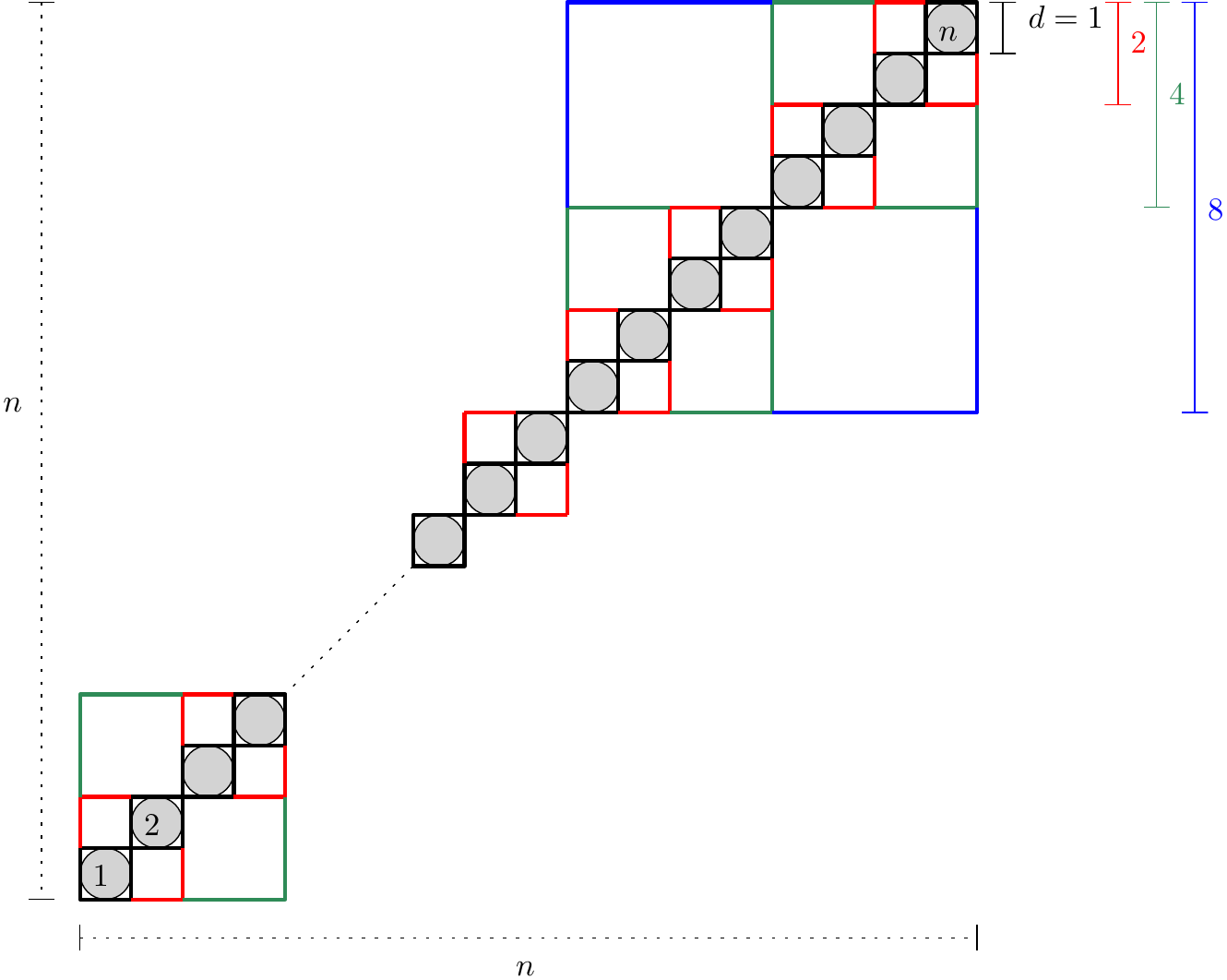}} 
	\caption{An initial diagonal shape $S_{D}$ on positions $ (x_{1},y_{1}),$ $ (x_{2},y_{2}), \ldots, (x_{n},y_{n}) $ enclosed into $ \log n $ boxes of dimensions $1,2,4, \ldots, n$.}
	\label{fig:Lower_Bound}
\end{figure}       

Now, observe that such a transformation at any order during its course, must pay at least $n$ steps to push $n$ nodes out from their black boxes of dimension 1. Likewise, when a line $l_1$ of 1 node occupying a cell $(x, y)$ (e.g, Figure \ref{fig:Lower_Bound_Case_1} (a)) is pushed one step to cross the boundary of its black box of dimension 1, no one will be pushed for free to move from any box of any size. The same argument follows, when a line $l_2$ of 2 nodes at $(x-1,y)$ and $(x-1, y-1)$ pushes 2 steps, say to the left, and leaves its red box of size $2\times2$, then no line is pushed to leave their red box of dimension 2 for free, see an example in Figure \ref{fig:Lower_Bound_Case_1} (b). More formally, by this observation, any transformation exploiting linear-strength pushing mechanism requires at least $d \cdot n/d$ steps, where the dimension $d = 2^k$ for all $0 \le k \le \log n$, to evacuate all lines from $n/d$ boxes of dimension $d$, without pushing any other lines for free in any arrangements during its course.    

\begin{figure}[th!]
	\centering
	\subcaptionbox{}
	{\includegraphics[scale=1]{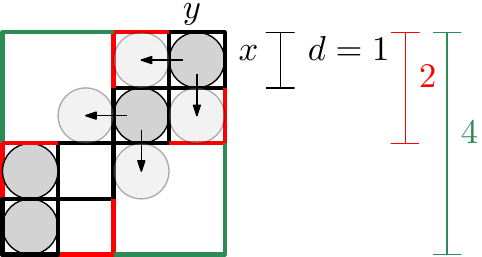}}  \qquad \qquad
	\subcaptionbox{}
	{\includegraphics[scale=1]{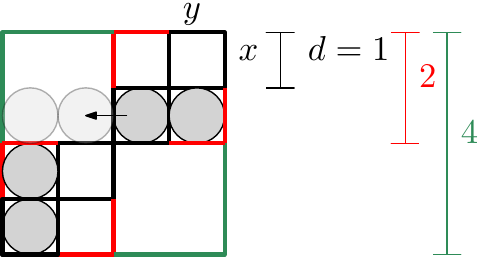}} 
	\caption{Artificial boxes of dimensions 1, 2 and 4 enclosing nodes of $S_D$.}
	\label{fig:Lower_Bound_Case_1}
\end{figure} 

There is another case might happen during transformations of any strategy based on line moves. Consider a square box of dimension $ d = 2^k$, for all $0 \le k \le \log n$, consists of four sub-boxes of dimension $2^{k-1}$  each, say without loss of generality, a blue box of size $8 \times 8$ holds four green sub-boxes of length-4 dimension, as depicted in Figure \ref{fig:Lower_Bound_Case_2}. Here, one can say that the line of length 4 in the top-right corner pushes 4 steps towards the left, which consequently moves the one in top-left for no cost. In such case, we should not forget the cost of forming the two length-4 line is prepaid previously. That is, the strategy already paid a cost of forming them initially from sub-boxes of dimensions 1 and 2. Further, one of the length-4 line is incurred the transformation a cost of 4 steps at least, to change its direction completely to occupy new 4 consecutive columns, that is, to line up vertically with the other length-4 line. Recall that any line of length $2^k$ occupying a box of $d = 2^k$ and crosses a boundary of that box vertically or horizontally, no line is pushed originally for free. This holds for any initial sub-lines of lengths less than $2^k$. 
\begin{figure}[th!]
	\centering
	{\includegraphics[scale=0.6]{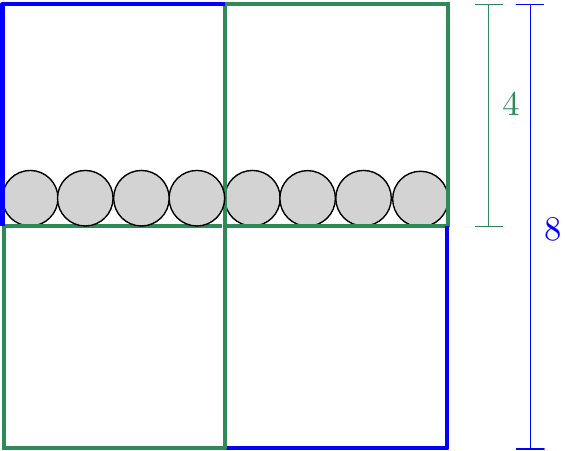}}  
	\caption{An example of the special case when nodes push others for free (on the same row of column).}
	\label{fig:Lower_Bound_Case_2}
\end{figure}    
With this, we can then calculate the total minimum steps that must be paid to evacuate all lines (of various lengths) from the $2n - 1$ square boxes. In each box, any strategy has to pay a minimum number of steps equals to the box's dimension $ d = 2^k$, for all $0 \le k \le \log n$, at any order during transformations. Thus the total minimum steps will  be $(1\cdot n) + (2\cdot n/2) + (4\cdot n/4) + \ldots + (n\cdot 1) = n + n + \ldots + n $. Now, since we have $\log n $ different dimensions, we obtain a total of $n \log n$ minimum number of steps. Hence, any transformation exploits linear-strength pushing mechanism asks for at least $\Omega (n \log n)$ steps to form all $n$ nodes at the potential placement and transform $S_{D}$ into $S_{L}$. 

Then, we try to apply a recursive transformation to check whether this  will  yield a better lower bound. That is, let $S_{L}$ be an initial straight line of $n$ nodes (say horizontal) which occupies the bottommost row $y_1$ and $S_D$ is a target diagonal of order $ n $  (lines of length 1) occupies  positions $ (x_{1},y_{1}),$ $ (x_{2},y_{2}), \ldots, (x_{n},y_{n}) $.  By reversibility, the pair $(S_D, S_{L})$ are transformable to each other, such that if $S_D \rightarrow  S_{L} $ (``$\rightarrow$'' means ``is transformed to'') then $ S_{L} \rightarrow S_{D} $ via a sequence of line moves. Then the cost of $S_D \rightarrow  S_{L} $ is equivalent to $ S_{L} \rightarrow S_{D} $.  

\begin{figure}[th!]
	\centering
	{\includegraphics[scale=0.7]{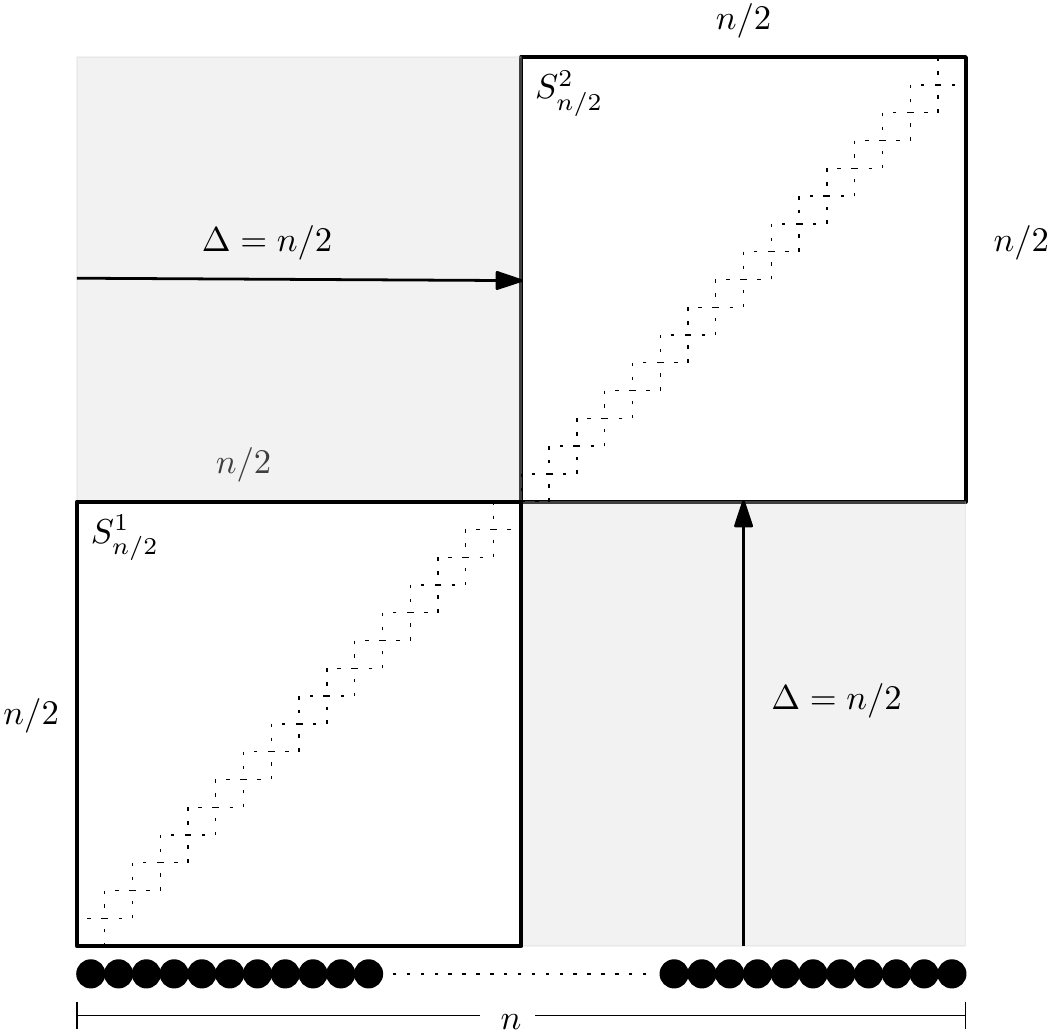}} 
	\caption{All nodes are initially placed on the bottom. Two independent sets $S_{n/2}^{1} $ and $S_{n/2}^{2} $ are defined.}
	\label{fig:Lower_Bound_3}
\end{figure}  

We define two independent sets, $S_{n/2}^{1} $ and $S_{n/2}^{2} $, each of which contains arbitrary $n/2$ nodes during configurations, such that $S_{n/2}^{1} \ne S_{n/2}^{2}$ and $S_{n/2}^{1}  \cap S_{n/2}^{2} = \phi$, see Figure \ref{fig:Lower_Bound_3}. Given a transformation $A$, then pick any $n/2 $ nodes randomly chosen from $S_L$. At any time, $A$ must pay a cost of at least $n/2$ for these specific $n/2$ nodes to cross a boundary of the $S_{n/2}^{2}$ box and get inside it, through the two shaded areas. This cost is based of the minimum distance any group of $n/2$ nodes have to pay, in order to reach their final positions inside the $S_{n/2}^{2}$ box. 

\begin{figure}[th!]
	\centering
	{\includegraphics[scale=0.7]{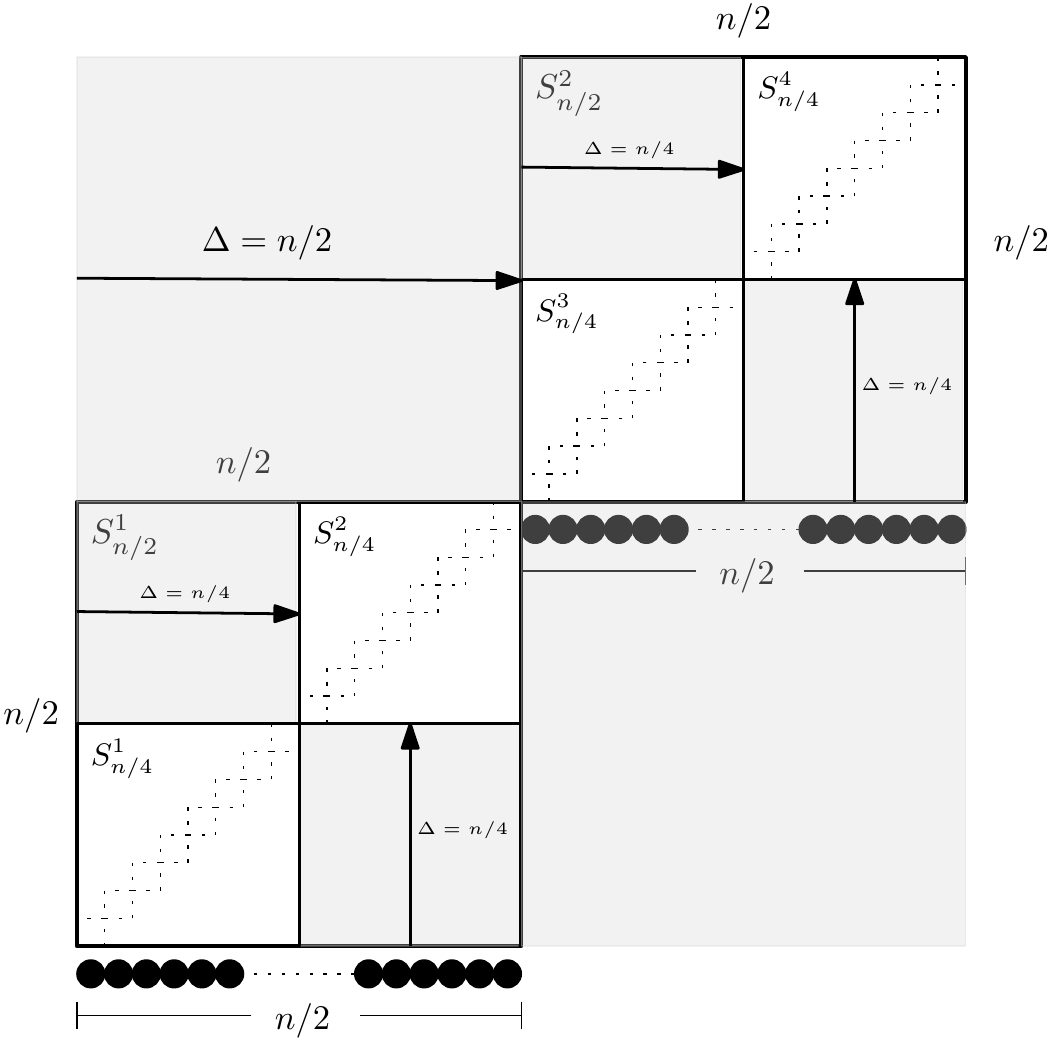}} 
	\caption{The two sets $S_{n/2}^{1} $ and $S_{n/2}^{2} $ are divided into four independent sets $S_{n/4}^{1}, S_{n/4}^{2} \subset S_{n/2}^{1} $ and $S_{n/4}^{3}, S_{n/4}^{4} \subset S_{n/2}^{2} $. }
	\label{fig:Lower_Bound_3_1}
\end{figure}

Similarly, we split the two boxes into half, by defining four independent sets $S_{n/4}^{1}, S_{n/4}^{2} \subset S_{n/2}^{1} $ and $S_{n/4}^{3}, S_{n/4}^{4} \subset S_{n/2}^{2} $. At any time, $A$ chooses two random group of nodes, each of size $n/4$. Then, $A$ has to pay a cost of at least $2 \cdot n/4$ for these nodes to cross boundaries and get inside the $S_{n/4}^{2}$ and  $S_{n/4}^{4}$ boxes. See Figure \ref{fig:Lower_Bound_3_1}. Repeat the same argument for the rest of $ \log n$ charging phases, until each nodes occupies its final target  positions in the diagonal $S_D$. As every node in each independent set will eventually reach its final position, then it will be contained into $\log n$ boxes. Therefore, the total amortized cost will be $(1 \cdot n/2) + (2\cdot n/4) + (4.\cdot n/8) + \ldots + (n/2 \cdot 1)  = n \log n / 2$.   As a result, we state that:

\begin{theorem}
	Any transformation strategy exploiting line moves requires $\Theta (n \log n)$ steps to transform the diagonal into a line. 
\end{theorem}

\section{Conclusions and Open Problems }
\label{sec:Conclusion}

We have presented efficient transformations for the line-pushing model introduced in \cite{AMP19} and some first lower bounds for restricted sets of transformations. Our first transformation works on the family of all Hamiltonian shapes and matches the running time of the best known transformations ($O(n\log n)$) while additionally managing to preserve connectivity throughout its course. We then gave the first universal connectivity preserving transformation for this model. Its running time is $O(n\sqrt{n})$ and works on any pair of connected shapes of the same order. Our $\Omega(n\log n)$ lower bounds match the best known upper bounds, still they are valid only for restricted sets of transformations.

This work opens a number of interesting problems and research directions. An immediate next goal is whether it is possible to develop an $O(n\log n)$-time universal connectivity-preserving transformation. If true, then a natural question is whether a universal transformation can be achieved in $o(n\log n)$-time (even when connectivity can be broken) or whether there exists a general $\Omega(n \log n)$-time matching lower bound. As a first step, it might be easier to develop lower bounds for the connectivity-preserving case. There are also a number of interesting variants of the present model. One is a centralised parallel version in which more than one line can be moved concurrently in a single time-step. Another, is a distributed version of the parallel model, in which the nodes operate autonomously through local control and under limited information.

\bibliographystyle{plainurl}
\bibliography{ref} 
\end{document}